\documentclass[12pt,oneside,reqno,a4paper]{amsart}

\usepackage[latin1]{inputenc}
\usepackage{amsthm}
\usepackage{amsmath}
\usepackage{amsfonts}
\usepackage{amssymb}
\usepackage{graphicx}
\usepackage[colorlinks]{hyperref}
\usepackage{cleveref} 
\usepackage{thm-restate}
\usepackage[square,numbers]{natbib}
\usepackage[paper=a4paper, top=3cm,bottom=2.5cm,left=2.5cm,right=2.5cm]{geometry}
\usepackage[normalem]{ulem}

\hypersetup{
    colorlinks = true,
    linkcolor = blue,
    citecolor = cyan
}

\newcommand{\Weizsacker}{Weizs\"{a}cker}
\newcommand{\TFW}{\textnormal{TFW}}
\newcommand{\loc}{\textnormal{loc}}
\newcommand{\unif}{\textnormal{unif}}

\newcommand{\id}{\textnormal{ d}}

\newcommand{\wt}{\widetilde}

\newcommand{\R}{\mathbb{R}^{3}}

\newcommand{\Holder}{H\"{o}lder}
\newcommand{\Schrodinger}{Schr\"{o}dinger}

\newcommand{\smcb}{ \textstyle{\frac{1}{|\cdot|}}  }
\newcommand{\smfrac}[2]{ \textstyle{ \frac{#1}{#2} } }

\newcommand{\ou}{\overline{u}}
\newcommand{\ophi}{\overline{\phi}}
\newcommand{\om}{\overline{m}}

\newcommand{\refer}[1]{\textbf{#1}}
\newcommand{\spt}{\textnormal{spt}}

\newcounter{mydef}
\newcounter{defcount}

\newtheorem{theorem}[mydef]{Theorem}
\newtheorem{definition}[defcount]{Definition}
\newtheorem{lemma}[mydef]{Lemma}
\newtheorem{proposition}[mydef]{Proposition}

\newtheorem*{assumption*}{Assumption}
\newtheorem{corollary}[mydef]{Corollary}
\newtheorem*{theorem*}{Theorem}
\newtheorem*{proposition*}{Proposition}
\newtheorem*{lemma*}{Lemma}

\theoremstyle{definition}
\newtheorem{remark}{Remark}

\numberwithin{theorem}{section}
\numberwithin{lemma}{section}
\numberwithin{corollary}{section}
\numberwithin{proposition}{section}
\numberwithin{equation}{section}

\def\Xint#1{\mathchoice
{\XXint\displaystyle\textstyle{#1}}%
{\XXint\textstyle\scriptstyle{#1}}%
{\XXint\scriptstyle\scriptscriptstyle{#1}}%
{\XXint\scriptscriptstyle\scriptscriptstyle{#1}}%
\!\int}
\def\XXint#1#2#3{{\setbox0=\hbox{$#1{#2#3}{\int}$}
\vcenter{\hbox{$#2#3$}}\kern-.5\wd0}}

\def\dashint{\Xint-}

\title{Locality of the Thomas--Fermi--von Weizs\"{a}cker Equations}

\author{F. Q. Nazar}
\author{C. Ortner}
\date{\today}

\thanks{FQN is funded by the MASDOC doctoral training centre, EPSRC grant
  EP/H023364/1. CO was supported by ERC Starting Grant 335120 and the Leverhulme
  Trust through a Philip Leverhulme Prize.}

\address{F. Q. Nazar \\ Mathematics Institute \\ Zeeman Building \\
  University of Warwick \\ Coventry CV4 7AL \\ UK}
\email{F.Q.Nazar@warwick.ac.uk}

\address{C. Ortner\\ Mathematics Institute \\ Zeeman Building \\
  University of Warwick \\ Coventry CV4 7AL \\ UK}
\email{C.Ortner@warwick.ac.uk}

\begin{document}

\maketitle

\begin{abstract}
  We establish a pointwise stability estimate for the Thomas--Fermi--von
  Weiz-s\"{a}cker (TFW) model, which demonstrates that a local perturbation of a
  nuclear arrangement results also in a local response in the electron density
  and electrostatic potential. The proof adapts the arguments for existence and
  uniqueness of solutions to the TFW equations in the thermodynamic limit by
  Catto {\it et al.} (1998). 
  
  To demonstrate the utility of this combined locality and stability result we
  derive several consequences, including an exponential convergence rate for
  the thermodynamic limit, partition of total energy into exponentially
  localised site energies (and consequently, exponential locality of forces),
  and generalised and strengthened results on the charge neutrality of local
  defects.
%
%
%
%
\end{abstract}

\maketitle

\section{Introduction}
%
%
%
Locality properties of electronic structure models are a key premise in certain state of the art numerical algorithms. A well-established example
is ``near-sightedness'', a locality property of the density matrix which gives
rise to linear scaling algorithms for Kohn--Sham type models
\cite{IsmailBeigi_LocalityoftheDensityMatrix, KohnNearsightedness,
  GoedeckerLinearScaling, RazoukDecayPropertiesSpectralProjectors}. A stronger
notion is the locality of the mechanical response, which is a fundamental
premise underpinning the construction of interatomic potentials and of
multi-scale algorithms such as hybrid QM/MM schemes \cite{Csanyi} (here, it is
termed ``strong locality'').  This latter category of locality is less well
studied, the only result in this direction being the locality of
non-selfconsistent tight binding models \cite{Huajie}.


The aim of the present work is to establish the locality properties satisfied by
the Thomas--Fermi--von \Weizsacker\, (TFW) model. Our main technical result to
achieve this is the following pointwise stability estimate for the TFW
equations, which establishes the locality of the electron response to changes in
the nuclear configuration.  Compared with \cite{Huajie} it is noteworthy that
our result takes Coulomb interaction fully into account. Rigorous
  statements, under different conditions, are given in Theorems
  \ref{Theorem - One inf pointwise stability estimate alt} and \ref{Theorem -
    Exponential Est Integral RHS}.

\begin{theorem*}
  For $i = 1,2$ let $m_{i} \in L^\infty(\R)$ represent nuclear charge
  distributions satisfying
  \begin{align*}
    m_i \geq 0 \qquad \text{and} \qquad
    \lim_{R \to \infty} \frac{1}{R} \inf_{x \in \R} \int_{B_{R}(x)} m_{i}(z) \id z = + \infty.
  \end{align*}
  Let the corresponding ground state electron densities and electrostatic
  potentials, denoted by $u_{i}, \phi_i : \R \to \mathbb{R}$, satisfy the TFW
  equations,
  \begin{align*}
    - \Delta u_{i} &+ \frac{5}{3} u_{i}^{7/3} - \phi_{i} u_{i} = 0,  \\
    - \Delta \phi_{i} &= 4\pi (m_{i} - u_{i}^{2}).
  \end{align*}
  Then there exists $C, \gamma > 0$ such that for all $y \in \R$
  \begin{align}
    |(u_{1} - u_{2})(y)| + |(\phi_{1} - \phi_{2})(y)| \leq C \left( \int_{\R} |(m_{1} - m_{2})(x)|^{2} e^{- 2 \gamma |x - y|} \id x \right)^{1/2}. \label{goal 1}
  \end{align}
\end{theorem*}

In the remainder of the article we explore some of the consequences of this
locality result: In \Cref{Proposition - Infinite Finite Ground state comparison}
we obtain new estimates on finite-domain approximations which yield
  exponential decay of surface energies as well as an exponential convergence
  rate for the thermodynamic limit. In \Cref{Corollary -
  Exponential Estimates Consequences} we show that (\ref{goal 1}) gives rise to
rigorous results that match, and substantially generalise, the Thomas--Fermi
theory of impurity screening in metals \cite{Ashcroft,Kittel}. In \Cref{Theorem
  - Neutrality Estimate} we strengthen existing results on the neutrality of the
TFW model \cite{C/Ehr}. In all these results, general (condensed) nuclear arrangements are treated.

A striking application of \eqref{goal 1} is that it allows us to decompose
energy into local contributions from which we obtain local site energy
potentials: Given a countable collection of nuclei
$Y =(Y_j)_{j \in \mathbb{N}} \subset \R$ we construct an energy density
$\mathcal{E}(Y; x)$ which allows us to define the TFW energy
$\int_\Omega \mathcal{E}(Y; x)\,{\rm d}x$ of an arbitrary volume
$\Omega \subset \R$ in a meaningful way. This then motivates us to define site
energies
\begin{displaymath}
  E_j(Y) := \int_{\R} \varphi_j(x) \mathcal{E}(Y; x)\, {\rm d}x,
\end{displaymath}
where $(\varphi_j)_{j \in \mathbb{N}}$ is a smooth partition of unity of $\R$,
which can be constructed in such a way that $E_j$ are permutation and
isometry invariant and most crucially, $E_j$ are {\em local} in the sense
that
\begin{align}
  \left|  \frac{\partial E_{j}(Y)}{\partial Y_{k}} \right| \leq C e^{-\gamma |Y_{j} - Y_{k}|}, \label{goal 2}
\end{align}
for some $C, \gamma > 0$. The rigorous statement of this result is given in
\Cref{Theorem - Forcing Est}. An analogous result has recently been proven for a
tight binding model in~\cite{Huajie}.

This result not only gives a strong justification for the construction of
classical short-ranged interatomic potentials in metals, but in fact it allows
us to treat the TFW mechanical response as if it emanated from such a classical
potential. For example, (i) the analysis of the Cauchy--Born continuum limit
\cite{CauchyBorn} applies directly to the TFW model; and (ii) we can generalise
in \cite{Paper2} the analysis of variational problems for the mechanical
response to defects in an infinite crystal \cite{EOS}.

The remainder of this article is organised as follows: In Section \ref{Section 2
  - The TFW Model} we recall the definition of the TFW model and summarise the
relevant existing results. In Section \ref{Section 3 - Main Results} we state
the main technical results, including the rigorous statement of the stability
result \eqref{goal 1}. In Section \ref{Section 4 - Applications} we present
applications. Concluding remarks are made in \Cref{Section 5 - Conclusion and
  Outlook}, followed by the detailed proofs of the results in \Cref{Section 6 -
  Proofs}.



\begin{remark}
  We conclude the introduction with a remark about the analytical context of
  this work. The TFW equations for the electron density and potential is a
  coupled \Schrodinger --Poisson system. Other systems of this class can be
  found in semiconductor physics
  \cite{Markowich_SemiconductorEquations,Wang/Zhou,Zhao/Zhao}. We also note that
  Thomas--Combes estimates give conditions under which eigenfunctions of a
  \Schrodinger \, operator decay exponentially
  \cite{Agmon_LecturesofExponentialDecay,Combes/Thomas}. While the results
  obtained for these systems are similar, the corresponding equations have
  different structure, hence the analytical techniques used to study them differ
  considerably.
  
  The closest existing result to (\ref{goal 1}) we have found is \cite[Theorem
  4.6]{BlancScreening}, which shows the exponential decay of the electron
  density away from the boundary of a crystal. Both (\ref{goal 1}) and
  \cite[Theorem 4.6]{BlancScreening} utilise the uniqueness of the TFW equations
  to prove stability estimates. In \Cref{sec:tdl-estimates}, we use
  \Cref{Proposition - Infinite Finite Ground state comparison} to generalise
  \cite[Theorem 4.6]{BlancScreening}.
\end{remark}

\subsection*{Acknowledgements}
We thank Chen Huang for his valuable comments on impurity screening in the
Thomas--Fermi model, and Eric Cances and Virginie Ehrlacher for helpful
  discussions about different analytical techniques for the TFW model.

\section{The TFW Model}
\label{Section 2 - The TFW Model}
For $p \in [1,\infty]$ we define the function spaces
\begin{align*}
L^{p}_{\loc}(\R) &:= \{ \, f: \R \to \mathbb{R} \, | \, \forall \, K \subset \R \text{ compact}, f \in L^{p}(K) \, \} \quad \text{ and} \\
L^{p}_{\unif}(\R) &:= \{ \, f \in L^{p}_{\loc}(\R) \, | \, \sup_{x \in \R} \| f \|_{L^{p}(B_{1}(x))} < \infty \, \}.
\end{align*}
For $k \in \mathbb{N}$, $H^{k}_{\loc}(\R), H^{k}_{\unif}(\R)$ are defined analogously. For a multi-index $\alpha = (\alpha_{1},\alpha_{2},\alpha_{3})$, we define the partial derivative $\partial^{\alpha} = \partial_{1}^{\alpha_{1}} \partial_{2}^{\alpha_{2}} \partial_{3}^{\alpha_{3}}$. Throughout this paper, $\alpha, \beta$ denote three-dimensional multi-indices.

The Coulomb interaction, for $f, g \in L^{6/5}(\R)$, is given by
\begin{align}
D(f,g) = \int_{\R} \int_{\R} \frac{f(x)g(y)}{|x-y|} \id x \id y = \int_{\R} \left( f * \smcb \right)(y) g(y) \id y. \label{eq: Coulomb energy def}
\end{align}
This is finite due to the Hardy--Littlewood--Sobolev estimate \cite{Aubin_HLS}
\begin{align}
|D(f,g)| \leq C \| f \|_{L^{6/5}(\R)} \| g \|_{L^{6/5}(\R)}. \label{eq: HLS embedding eq}
\end{align}

Let $m \in L^{6/5}(\R), m \geq 0$ denote the charge density of a finite nuclear cluster, 
then the corresponding TFW energy functional is defined, for $v \in H^{1}(\R)$,
by
\begin{align}
E^{\TFW}(v,m) = C_{\textnormal{W}} \int_{\R} |\nabla v|^{2} + C_{\textnormal{TF}} \int_{\R} v^{10/3} + \frac{1}{2} D( m - v^{2}, m - v^{2}). \label{eq: TFW energy def, with constants}
\end{align}
The function $v$ corresponds to the positive square root of the electron
density. The first two terms of $E^{\TFW}(v,m)$ model the kinetic energy of the
electrons while the third term models the Coulomb energy. We remark that this definition of the
Coulomb energy is only valid for smeared nuclei. We can rescale the energy to
ensure $C_{\textnormal{W}} = C_{\textnormal{TF}} = 1$.

To construct the electronic ground state for an infinite arrangement of
  nuclei (e.g., crystals), we restrict admissible nuclear charge densities to
  $m \in L^{1}_{\unif}(\R), m \geq 0$, satisfying
\begin{align*}
\textnormal{(H1) }& \sup_{x \in \R} \int_{B_{1}(x)} m(z) \id z < \infty, \\
\textnormal{(H2) }& \lim_{R \to \infty} \inf_{x \in \R} \frac{1}{R} \int_{B_{R}(x)} m(z) \id z = \infty.
\end{align*} 
The property (H1) guarantees that no clustering of infinitely many nuclei occurs
at any point in space whereas (H2) ensures that there are no large regions that
are devoid of nuclei.

Let $R_{n} \uparrow \infty$ and define the truncated nuclear distribution
$m_{R_{n}} = m \cdot \chi_{B_{R_{n}}(0)}$, then the minimisation problem
\begin{align*}
I^{\TFW}(m_{R_{n}}) = \inf \left\{ \, E^{\TFW}(v,m_{R_{n}})  \, \bigg| \, v \in H^{1}(\R), v \geq 0, \int_{\R} v^{2} = \int_{\R} m_{R_{n}} \, \right \}
\end{align*}
possesses a unique minimiser $u_{R_{n}}$. The charge constraint ensures that the system is neutral. Further, $u_{R_{n}}$ solves the corresponding Euler--Lagrange equation, which can be expressed as the coupled system
\begin{subequations}
\label{eq: u phi Rn pair}
\begin{align}
- \Delta u_{R_{n}} &+ \frac{5}{3} u^{7/3}_{R_{n}} - \phi_{R_{n}} u_{R_{n}} = 0, \label{eq: u phi eqs, uRn} \\
- \Delta \phi_{R_{n}} &= 4\pi (m_{R_{n}} - u^{2}_{R_{n}}). \label{eq: u phi eqs, phiRn}
\end{align}
\end{subequations}
By the proof of \cite[Corollary 2.7, Theorem 6.10]{C/LB/L}, it follows that
\begin{align}
\| u_{R_{n}} \|_{H^{1}_{\unif}(\R)} + \| \phi_{R_{n}} \|_{L^{2}_{\unif}(\R)} \leq C, \label{eq: uRn phiRn uniform ests}
\end{align}
where C is independent of $R_{n}$. Consequently, (\ref{eq: uRn phiRn uniform ests}) implies that along a subsequence $(u_{R_{n}},\phi_{R_{n}})$ converge to
$(u,\phi)$. Passing to the limit in (\ref{eq: u phi Rn pair}) yields the
following result.
%

\begin{theorem}[{\cite[Theorem 6.10]{C/LB/L}}]
\label{Theorem - C/LB/L Main Result}
Let $m \in L^{1}_{\unif}(\R), m \geq 0$ satisfy \textnormal{(H1)--(H2)}, then there exists a unique solution $(u,\phi) \in L^{\infty}(\R) \times L^{1}_{\unif}(\R)$, up to the sign of $u$, of
\begin{subequations}
\label{eq: u phi eq pair}
\begin{align}
- \Delta u &+ \frac{5}{3} u^{7/3} - \phi u = 0, \label{eq: u inf eq} \\
- \Delta \phi &= 4\pi (m - u^{2}), \label{eq: phi inf eq}
\end{align}
\end{subequations}
in the distributional sense. In addition, $\inf u > 0$.
\end{theorem}

\begin{definition}
For any nuclear configuration $m$ satisfying \textnormal{(H1)--(H2)}, we refer to $(u,\phi)$ solving \textnormal{(\ref{eq: u phi eq pair})} as the ground state corresponding to $m$. \qed
\end{definition}
\begin{remark}
A concise proof of uniqueness of the TFW equations is given in \cite{Blanc_Uniqueness} under the assumption that $m$ is smooth and hence $u,\phi \in W^{1,\infty}(\R)$, which simplifies the earlier proof given in \cite{C/LB/L}. \qed
\end{remark}
In the next section, we discuss results that can be obtained by generalising the proof of \Cref{Theorem - C/LB/L Main Result}.

\section{Main Results}
\label{Section 3 - Main Results}

\subsection{Uniform regularity estimates}
\label{Subsection - Uniform Regularity Estimates} 
In the proof of our main results in the next section we employ regularity
estimates that refine those of \cite{C/LB/L}, and may be of independent
interest.

Other than \Cref{Proposition - General Regularity Est}, our estimates rely on uniform variants of (H1)--(H2) and it turns out that
(H2) may then be simplified without loss of generality; see \Cref{Lemma - H2
  equiv lemma} for more details. Given $M, \omega_{0}, \omega_{1} > 0$, let $\omega = (\omega_{0},\omega_{1})$ and define the class of
nuclear configurations
\begin{align}
\mathcal{M}_{L^{2}}(M,\omega) = \bigg\{ \, m \in L^{2}_{\unif}(\R) \, \bigg| \,\,\, &\| m \|_{L^{2}_{\unif}(\R)} \leq M, \nonumber \\ \, & \forall \, R > 0 \,\, \inf_{x \in \R} \int_{B_{R}(x)} m(z) \id z \geq \omega_{0} R^{3} - \omega_{1} \, \bigg\}. \label{eq: M L2 def}
\end{align}

As each nuclear distribution $m \in \mathcal{M}_{L^{2}}(M,\omega)$ satisfies (H1)--(H2), \Cref{Theorem - C/LB/L Main Result} guarantees the
existence of corresponding ground states $(u,\phi)$. We adapt the proof of
existence of \Cref{Theorem - C/LB/L Main Result} to show that the uniformity in upper and
lower bounds on $m \in \mathcal{M}_{L^2}(M, \omega)$ yields regularity estimates
and lower bounds on these ground states which are also uniform.



\begin{proposition}
\label{Proposition - General Regularity Est}
For any nuclear distribution $m:\R \to \mathbb{R}_{\geq 0}$, satisfying
\begin{align*}
\| m \|_{L^{2}_{\unif}(\R)} \leq M,
\end{align*}
there exists $(u,\phi)$ solving \textnormal{(\ref{eq: u phi eq pair})} and satisfying $u \geq 0$ and
\begin{align}
&\| u \|_{H^{4}_{\unif}(\R)} \, \leq C(1 + M^{15/4}), \label{eq: gen u H2 unif est} \\
&\| \phi \|_{H^{2}_{\unif}(\R)} \, \leq C(1 + M^{3/2} ). \label{eq: gen phi H2 unif est}
\end{align}
\end{proposition} 

\begin{proposition}
\label{Proposition - Uniform inf u estimate}
There exists $c_{M,\omega} > 0$ such that for all $m \in
\mathcal{M}_{L^{2}}(M,\omega)$ the corresponding ground state $(u,\phi)$ is unique and the electron density
$u$ satisfies
\begin{align}
\inf_{x \in \R} u(x) &\geq c_{M,\omega} > 0. \label{eq: u inf est}
\end{align}

\end{proposition}

Assuming higher regularity of the nuclear distributions implies higher regularity of the ground state. We therefore define, for $k \in \mathbb{N}_{0}$,
\begin{align*}
\mathcal{M}_{H^{k}}(M,\omega) = \bigg\{ \, m \in H^{k}_{\unif}(\R) \, \bigg| \,\,\, &\| m \|_{H^{k}_{\unif}(\R)} \leq M, \\ \, \, & \forall \, R > 0 \,\, \inf_{x \in \R} \int_{B_{R}(x)} m(z) \id z \geq \omega_{0} R^{3} - \omega_{1} \, \bigg\}.
\end{align*}
Arguing by induction and applying the uniform lower bound (\ref{eq: u inf est})
yields the following result.

\begin{corollary}
\label{Corollary - General Est Ck version}
Suppose $k \in \mathbb{N}_{0}$ and $m \in \mathcal{M}_{H^{k}}(M,\omega)$, then
the corresponding  solution $(u,\phi)$ to \textnormal{(\ref{eq: u phi eq pair})} satisfies
\begin{align}
\| u \|_{H^{k+4}_{\unif}(\R)} +
\| \phi \|_{H^{k+2}_{\unif}(\R)} \leq C(k,M,\omega). \label{eq: corr k, u phi extra reg est}
\end{align}
\end{corollary}

\subsection{Pointwise stability and locality}
\label{Subsection - Pointwise Stability Estimates}
We now discuss our main result, a pointwise stability estimate for
\textnormal{(\ref{eq: u phi eq pair})} which reveals a generic locality of the
TFW interaction.  To establish this result, we adapt the proof of uniqueness of
the TFW equations in \cite{C/LB/L,Blanc_Uniqueness}, specialising the class of
test functions to 
\begin{align}
H_{\gamma} &= \bigg \{ \, \xi \in H^{1}(\R) \, \bigg| \, |\nabla \xi(x)| \leq \gamma |\xi(x)| \, \forall \, x \in \R \, \bigg \} \label{eq: H gamma test space}
\end{align}
for some $\gamma > 0$.
Observe that $e^{-\wt \gamma |\cdot|} \in H_{\gamma}$ for any $0 < \wt \gamma \leq \gamma$.


\begin{theorem}
\label{Theorem - One inf pointwise stability estimate alt}
Let $m_{1} \in \mathcal{M}_{L^{2}}(M,\omega)$, and let $(u_{1}, \phi_{1})$
Also, let $m_{2} : \R \to \mathbb{R}_{\geq 0}$ satisfy $\| m_{2} \|_{L^{2}_{\unif}(\R)} \leq M'$ and suppose there exists $(u_{2},\phi_{2})$ solving \textnormal{(\ref{eq: u phi eq pair})} corresponding to $m_{2}$, satisfying $u_{2} \geq 0$ and
\begin{align}
\| u_{2} \|_{H^{4}_{\unif}(\R)} &+ \| \phi_{2} \|_{H^{2}_{\unif}(\R)} \leq C(M'). \label{eq: u2 phi2 reg est}
\end{align}
Further, there exist $C = C(M,M',\omega), \gamma = \gamma(M,M',\omega) >0$ such that
for any $\xi \in H_{\gamma}$
\begin{align}
\int_{\R} \bigg( \sum_{|\alpha_{1}| \leq 4} |\partial^{\alpha_{1}} (u_{1} - u_{2})|^{2} + \sum_{|\alpha_{2}| \leq 2} |\partial^{\alpha_{2}} (\phi_{1} - \phi_{2})|^{2} \bigg) \xi^{2} \leq C \int_{\R} (m_{1} - m_{2})^{2} \xi^{2}. \label{eq: w and psi partial xi global onesided est}
\end{align}
In particular, for any $y \in \R$,
\begin{align}
\sum_{|\alpha| \leq 2} |\partial^{\alpha} (u_{1} - u_{2})(y)|^{2} + |(\phi_{1} - \phi_{2})(y)|^{2} \leq C \int_{\R} |(m_{1} - m_{2})(x)|^{2} e^{-2\gamma |x - y|} \id x. \label{eq: w and psi pointwise rhs exp integral onesided est}
\end{align}
\end{theorem} 

\begin{remark}
  Since we do not assume that $m_{2} \in \mathcal{M}_{L^{2}}(M',\omega')$, we can
  not guarantee the uniqueness of the corresponding solution
  $(u_{2},\phi_{2})$. \qed
\end{remark}

We can generalise \Cref{Theorem - One inf pointwise stability estimate alt} to
obtain higher-order pointwise estimates, but this requires \emph{both}
$\inf u_{1}, \inf u_{2} > 0$, hence we need to assume
$m_{1}, m_{2} \in \mathcal{M}_{H^{k}}(M,\omega)$ for some
$k \in \mathbb{N}_{0}$.

\begin{theorem}
\label{Theorem - Exponential Est Integral RHS}
Let $k \in \mathbb{N}_{0}$ and $m_{1}, m_{2} \in \mathcal{M}_{H^{k}}(M,
\omega)$. Consider the corresponding ground states $(u_{1},\phi_{1}), (u_{2},\phi_{2})$ and define
\begin{align*}
w = u_{1} - u_{2}, \quad \psi = \phi_{1} - \phi_{2}, \quad R_{m} = 4 \pi (m_{1} - m_{2}).
\end{align*}
Then, there exist $C = C(k,M,\omega), \gamma = \gamma(M,\omega) > 0$ such that for any $\xi \in H_{\gamma}$
\begin{align}
\int_{\R} \bigg( \sum_{|\alpha_{1}| \leq k+4} |\partial^{\alpha_{1}} w|^{2} + \sum_{|\alpha_{2}| \leq k+2} |\partial^{\alpha_{2}} \psi|^{2} \bigg) \xi^{2} \leq C \int_{\R} \sum_{|\beta| \leq k} |\partial^{\beta} R_{m}|^{2} \xi^{2}. \label{eq: w and psi partial xi global est}
\end{align}
In particular, for any $y \in \R$,
\begin{align}
\sum_{|\alpha_{1}| \leq k+2} |\partial^{\alpha_{1}} w(y)|^{2} + \sum_{|\alpha_{2}| \leq k} |\partial^{\alpha_{2}} \psi(y)|^{2} \leq C \int_{\R} \sum_{|\beta| \leq k} |\partial^{\beta} R_{m}(x)|^{2} e^{-2\gamma |x - y|} \id x. \label{eq: w and psi pointwise rhs exp integral est}
\end{align}
\end{theorem}

\begin{remark}
It is possible to generalise \Cref{Theorem - One inf pointwise stability estimate alt} to treat nuclei described by a non-negative measure $m$ on $\R$ satisfying
\begin{align*}
\sup_{x \in \R} m(B_{1}(x)) \leq M, \tag{H1'}
\end{align*}
and there exist $\omega_{0} > 0, \omega_{1} \geq 0$ such that for all $R > 0$
\begin{align*}
\inf_{x \in \R} m(B_R(x)) \geq \omega_{0} R^3 - \omega_{1}. \tag{H2'}
\end{align*}
The existence and uniqueness of a corresponding ground state $(u,\phi)$ is guaranteed by~\cite[Theorem 6.10]{C/LB/L}. We believe that the arguments used to show \cite[Lemma 5.5]{C/LB/L} and Theorem~\ref{Theorem - One inf pointwise stability estimate alt} can be adapted to show pointwise estimates similar to \eqref{eq: w and psi partial xi global onesided est}--\eqref{eq: w and psi pointwise rhs exp integral onesided est} when $m_{1}, m_{2}$ satisfy \textnormal{(H1')--(H2')} and that $m_{1} - m_{2}$ is absolutely continuous with respect to the Lebesgue measure on $\R$, with a density belonging to $L^{2}_{\unif}(\R)$.

This result is not sufficient to consider the response of the ground state to a perturbation of point nuclei, though it may be possible to treat this using an approximation to the identity or by applying similar techniques. \qed
%
\end{remark}
\section{Applications}
\label{Section 4 - Applications}

\subsection{Thermodynamic limit estimates} 
\label{sec:tdl-estimates}
The following result provides an estimate for comparing the infinite ground state
with its finite approximation, over compact sets, thus providing explicit rates
of convergence for the thermodynamic limit. This is discussed in
\Cref{rem:rdl-estimate}.

Interpreted differently, the result yields estimates on the decay of the
perturbation from the bulk electronic structure at a domain boundary,
generalising the exponential decay estimate \cite[Theorem 4.6]{BlancScreening}
to arbitrary open $\Omega \subset \R$ and general
$m \in \mathcal{M}_{L^{2}}(M,\omega)$.

\begin{proposition}
\label{Proposition - Infinite Finite Ground state comparison}
Let $m \in \mathcal{M}_{L^{2}}(M,\omega)$ and $(u,\phi)$ be the corresponding
ground state. Further, let $\Omega \subset \R$ be open and suppose there exists
$m_{\Omega}: \R \to \mathbb{R}_{\geq 0}$ such that 
$m_{\Omega} = m$ on $\Omega$ and $\| m_{\Omega} \|_{L^{2}_{\unif}(\R)} \leq M$
(e.g., $m_\Omega = m \chi_\Omega$), then there exists
$(u_{\Omega},\phi_{\Omega})$ solving \textnormal{(\ref{eq: u phi eq pair})} with
$m=m_{\Omega}$, $u_{\Omega} \geq 0$ and
$C = C(M, \omega), \gamma = \gamma(M,\omega) > 0$, independent of $\Omega$, such
that for all $y \in \Omega$
\begin{align}
\sum_{|\alpha| \leq 2} |\partial^{\alpha} (u - u_{\Omega})(y)| + |(\phi - \phi_{\Omega})(y)| \leq C e^{-\gamma \textnormal{dist}(y,\partial \Omega)}.  \label{eq: w and psi Rbuf est}
\end{align} 
\end{proposition} 

\begin{remark}
  \label{rem:rdl-estimate}
  Let $R>0$ and $R_n \uparrow \infty$, then applying \Cref{Proposition -
      Infinite Finite Ground state comparison} with $\Omega = B_{R_n}(0)$ and
    $m_\Omega = m_{R_n}$ gives a rate of convergence for the finite
    approximation $(u_{R_{n}},\phi_{R_{n}})$, solving (\ref{eq: u phi Rn pair}),
    to the ground state $(u,\phi)$,
  \begin{align}
    \| u - u_{R_{n}} \|_{W^{2,\infty}(B_R)(0)} 
    + \| \phi - \phi_{R_{n}} \|_{L^{\infty}(B_R)(0)} \leq 
    C e^{- \gamma ( R_{n} - R )}. 
    \label{eq: w and psi Rn local exp conv est}
  \end{align}
  This strengthens the result that $(u_{R_{n}},\phi_{R_{n}})$ converges to
  $(u,\phi)$ pointwise almost everywhere along a subsequence \cite{C/LB/L}.
\end{remark}

\subsection{Locality of the Charge Response}
The following result shows that the decay properties of the nuclear perturbation
are inherited by the response of the ground state.

\begin{corollary}
\label{Corollary - Exponential Estimates Consequences}
Let $k \in \mathbb{N}_{0}$ and $m_{1}, m_{2} \in \mathcal{M}_{H^{k}}(M, \omega)$. Consider the corresponding ground states $(u_{1},\phi_{1}), (u_{2},\phi_{2})$ and define
\begin{align*}
w = u_{1} - u_{2}, \quad \psi = \phi_{1} - \phi_{2}, \quad R_{m} = 4 \pi (m_{1} - m_{2}).
\end{align*}
\begin{enumerate}
\item \textnormal{(Exponential Decay)} If $R_{m} \in H^{k}(\R)$ and
  $\spt(R_{m}) \subset B_{R}(0)$, or there exists $\gamma' > 0$ such that
$\sum_{|\beta| \leq k}| \partial^{\beta} R_{m}(x)| \leq C e^{-\gamma'|x|},$
then there exist $C=C(k,M,\omega), \gamma = \gamma(M,\omega) > 0$ depending
also on $R$ or $\gamma'$ such that
\begin{align}
\sum_{|\alpha_{1}| \leq k+2} |\partial^{\alpha_{1}} w(x)| + \sum_{|\alpha_{2}| \leq k} |\partial^{\alpha_{2}} \psi(x)| \leq C e^{- \gamma |x|}. \label{eq: w psi local exp est}
\end{align}

\item \textnormal{(Algebraic Decay)} If there exist $C, r > 0$ such that
$\sum_{|\beta| \leq k}| \partial^{\beta} R_{m}(x)| \leq C (1 + |x|)^{-r}$
then there exists $C = C(r,k,M,\omega) > 0$ such that
\begin{align}
\sum_{|\alpha_{1}| \leq k+2} |\partial^{\alpha_{1}} w(x)| + \sum_{|\alpha_{2}| \leq k} |\partial^{\alpha_{2}} \psi(x)| &\leq C (1 + |x|)^{-r}. \label{eq: w psi decay est}
\end{align}

\item \textnormal{(Global Estimates)} If $R_{m} \in H^{k}(\R)$ 
then there exists $C = C(k,M,\omega) > 0$ such that
\begin{align}
\| w \|_{H^{k+4}(\R)} + \| \psi \|_{H^{k+2}(\R)} \leq C \| R_{m} \|_{H^{k}(\R)}. \label{eq: w psi Sobolev est}
\end{align}
\end{enumerate}
\end{corollary}

%
%
%
%
%

\begin{remark}
  For some of our comparison results, we require only
  $m_1 \in \mathcal{M}_{L^2}(M,\omega)$ but impose weaker assumptions on $m_2$. This would
  not generalise \Cref{Corollary - Exponential Estimates Consequences} since any
  of the decay assumptions in (1--3) already imply that
  $m_2 \in \mathcal{M}_{L^2}(M, \omega')$ for some $\omega'$.
\end{remark}


\begin{remark}
The estimate (\ref{eq: w psi local exp est}) can be used to study the full non-linear response of the ground state to a nuclear impurity. We compare this to the results from the Thomas--Fermi (TF) \cite{Ashcroft,Kittel,Resta} and TFW \cite{DonovanMarchScreening,StoppingPower,WeizsackerCorrectionNumerics,PartiallyLinearizedTFW} theories of screening.

Consider a nuclear arrangement $m_{1} \in \mathcal{M}_{L^{2}}(M,\omega)$ and
model a nuclear impurity at the origin with positive charge $Z$ by $Z \eta(x)$,
where $\eta \in C^{\infty}_{\textnormal{c}}(\R), \eta \geq 0$ and
$\int \eta = 1$.  Then define the perturbed system by
$m_{2} = m_{1} + Z \eta \in \mathcal{M}_{L^{2}}(M_{1},\omega_{1})$ and consider
the corresponding TFW ground states $(u_{1},\phi_{1})$ and $(u_{2},\phi_{2})$,
respectively. From (\ref{eq: w psi local exp est}) of \Cref{Corollary -
  Exponential Estimates Consequences} it follows that
\begin{align}
\sum_{|\alpha|\leq 2} |\partial^{\alpha} (u_{1} - u_{2})(x)| + |(\phi_{1} - \phi_{2})(x)| \leq C Z e^{- \gamma |x|}, \label{eq: My TFW screening est}
\end{align}

We now compare (\ref{eq: My TFW screening est}) with existing results from the TF and TFW theories of screening. These models consider the formal linear response $(n,V)$ of the electron density and potential to a nuclear impurity at the origin, modelled by the Dirac distribution $Z \delta_{0}$, in a uniform electron gas. In both models, $V$ satisfies the linear equation
\begin{align*}
- \Delta V = 4 \pi [ n + Z \delta_{0} ],
\end{align*}
while $n$ solves either the linearised TF or TFW equations. In the TF model, $V$ and $n$ are shown to satisfy \cite[Page 112]{Kittel}, \cite[Page 342]{Ashcroft}
\begin{align}
V(x) = Z \frac{e^{-k_{s}|x|}}{|x|}, \qquad n(x) = - \frac{Z k_{s}^{2}}{4 \pi} \frac{e^{-k_{s}|x|}}{|x|}, \label{eq: TF screening est}
\end{align}
where $k_{s} > 0$ is a material-dependent constant called the inverse screening length. In the TFW model, $V$ and $n$ satisfy \cite{DonovanMarchScreening,WeizsackerCorrectionNumerics,PartiallyLinearizedTFW,StoppingPower}
\begin{align}
V(x) &= \frac{Z}{4 \alpha \beta |x|} e^{-\alpha|x|} \left( (\alpha + \beta)^{2} e^{\beta |x|} - (\alpha - \beta)^{2} e^{-\beta |x|}  \right), \nonumber \\ n(x) &= - \frac{(\alpha^{2} - \beta^{2})^{2}Z}{ \alpha \beta |x|} e^{-\alpha|x|} \left( e^{\beta |x|} - e^{-\beta |x|}  \right), \label{eq: TFW screening est}
\end{align}
where $\alpha \in \mathbb{R}, \beta \in \mathbb{C}$ satisfy $0 < |\beta| < \alpha$. The constants $\alpha, \beta$ depend on the material and the coefficient $C_{W}$, which appears in the definition of the TFW energy (\ref{eq: TFW energy def, with constants}). There is a critical value of $C_{W}$ below which $\beta > 0$ and above which $\beta$ is complex, the latter case introduces oscillations in the potential and electron density. In either case, both the TF and TFW models exhibit screening due to the presence of the exponential term appearing in (\ref{eq: TF screening est})--(\ref{eq: TFW screening est}).

The lack of a factor of the form $\smfrac{1}{|x|}$ in (\ref{eq: My TFW screening est}) can be attributed to using a smeared nuclear description for the impurity as opposed to a point description in (\ref{eq: TF screening est})--(\ref{eq: TFW screening est}). Other than this, the similarity of (\ref{eq: My TFW screening est}) to (\ref{eq: TF screening est}) suggests that the constant $\gamma$ in (\ref{eq: My TFW screening est}) may be interpreted as the inverse screening length. In this paper we show there exists $\gamma > 0$ satisfying (\ref{eq: My TFW screening est}), however we do not provide any estimates for its value.

The estimate (\ref{eq: My TFW screening est}) shows that screening occurs in the
TFW model, without any approximations made to the model and without any
restrictions on the nuclear configurations (other than (H1)--(H2)). It should
be noted that although (\ref{eq: My TFW screening est}) agrees with existing
results from the TF theory of screening, in metals often the effects of
screening are weaker. For metals, instead of an exponentially decaying
screening factor, Friedel oscillations are observed
\cite{FriedelOscillations/FermiSurfaces, MarchManyBody, FriedelOscillationsEncyclopedia}. In
this case, the screening factor behaves as $|x|^{-r} f(|x|)$, where $f:
\mathbb{R}_{\geq 0} \to [-1,1]$ is an oscillating function and the decay rate $r
> 0$ depends on the Fermi surface of the metal. The {\em generic} exponential screening factor in (\ref{eq: My TFW screening est}) demonstrates that the TFW model significantly overscreens charges. \qed
\end{remark}

\subsection{Neutrality of defects}
An immediate consequence of \Cref{Corollary - Exponential Estimates
  Consequences} is the neutrality of nuclear perturbations in the TFW
equations. This result applies to all nuclear configurations belonging to
$\mathcal{M}_{L^{2}}(M,\omega)$. In particular \Cref{Theorem - Neutrality
  Estimate}(3) strengthens the result of \cite{C/Ehr}, which requires
   $m_1-m_2 \in L^1(\R) \cap L^2(\R)$ and thus excludes typical point defects; see
 \Cref{rem:discussion-neutrality} for more details.

\begin{theorem}
\label{Theorem - Neutrality Estimate}
Let $m_{1}, m_{2} \in \mathcal{M}_{L^{2}}(M,\omega)$ and define $\rho_{12} := m_{1} - u_{1}^{2} - m_{2} + u_{2}^{2}$.
\begin{enumerate}
\item 
  If $\spt(m_{1} -m_{2}) \subset B_{R'}(0)$, or there exist $C, \wt \gamma > 0$ such that
$|(m_{1} - m_{2})(x)| \leq C e^{- \wt \gamma|x|}$,
then $\rho_{12} \in L^{1}(\R)$ and there exist $C, \gamma > 0$ such that, for all $R > 0$, 
\begin{align}
\bigg | \int_{B_{R}(0)} \rho_{12} \bigg | &\leq C e^{- \gamma R}. \label{eq: local pert BR exp est}
\end{align}

\item 
  If there exists $C, r > 0$ such that
$|(m_{1}- m_{2})(x)| \leq C (1 + |x|)^{-r}$ 
then there exists $C > 0$ such that, for all $R > 0$,
\begin{align}
\bigg | \int_{B_{R}(0)} \rho_{12} \bigg | &\leq C (1+ R)^{2 - r}. \label{eq: local pert BR alg est}
\end{align}

\item 
  If $m_{1} - m_{2} \in L^{2}(\R)$ \textnormal{(}e.g., $r > 3/2$ in \textnormal{(2)}\textnormal{)} then
  $\rho_{12} \in L^{2}(\R)$ and
\begin{align}
\lim_{\varepsilon \to 0} \frac{1}{|B_{\varepsilon}(0)|} \int_{B_{\varepsilon}(0)} \widehat{\rho}_{12}(k) \id k = 0, \label{eq: global neutrality Leb point eq}
\end{align}
where $\widehat{\rho}_{12}$ denotes the Fourier transform of $\rho_{12}$.
\end{enumerate}
\end{theorem}

\begin{remark} \label{rem:discussion-neutrality}
  In a forthcoming article \cite{Paper2}, we construct a variational problem to
  study the response of a crystal due to a local defect, using the TFW
  energy. Arguing as in \cite{EOS}, we shall show that any minimising
  displacement decays away from the defect at the rate $|x|^{-2}$, which
  corresponds to case (2) with $r=2$. In this case (\ref{eq: local pert BR alg
    est}) only provides a uniform bound for the charge as opposed to a decay
  estimate. However, as $r > 3/2$ the global neutrality result (\ref{eq: global
    neutrality Leb point eq}) holds for the relaxed system.
  
  The neutrality estimates of \Cref{Theorem - Neutrality Estimate} strengthen
  those of \cite{C/Ehr} in the following ways. Firstly, our result considers a
  perturbation of a general nuclear arrangement as opposed to a perfect
  crystal. This allows us, in \cite{Paper2}, to consider the response of
  extended defects such as dislocations. In addition, we only require that the
  nuclear perturbation belongs to $L^{2}(\R)$, which we prove rigorously in
    \cite{Paper2}, whereas in \cite{C/Ehr} the nuclear defect is assumed to lie
  in $L^{1}(\R) \cap L^{2}(\R)$, which fails for typical point defects.  \qed
\end{remark}

\subsection{Energy locality}
We now show that the locality result, \Cref{Theorem - Exponential Est Integral
  RHS}, can be used to describe the energy contribution of each individual
nucleus. In effect, we will derive a {\em site energy potential} for the TFW
model, which has the surprising consequence that, for the study of mechanical
response, TFW can be treated as a classical short-ranged interatomic
potential. Our result gives credence to the construction of interatomic
potentials and the assumption of {\em strong locality} used in hybrid quantum
mechanics/molecular mechanics (QM/MM) simulations \cite{Csanyi}.

Let $\eta \in C^{\infty}_{\textnormal{c}}(B_{R_{0}}(0))$ be radially symmetric
and satisfy $\eta \geq 0$ and $\int_{\R} \eta = 1$ describe the charge density
of a single (smeared) nucleus, for some fixed $R_{0} > 0$. For any countable
collection of nuclear coordinates
$Y = (Y_{j})_{j \in \mathbb{N}} \in (\R)^\mathbb{N}$, let the corresponding
nuclear configuration be defined by
\begin{align}
  m_{Y}(x) = \sum_{j \in \mathbb{N}} \eta(x - Y_{j}). \label{eq: m Y ass def}
\end{align}
A natural space of nuclear coordinates, related to the $\mathcal{M}_{L^2}$
space is
\begin{align}
  \mathcal{Y}_{L^{2}}(M,\omega) := \{ \, Y \in (\R)^\mathbb{N}  \, | \, m_{Y} \in \mathcal{M}_{L^{2}}(M,\omega)  \, \}. \label{eq: Y space def}
\end{align}
This space contains many condensed phases, such as crystals containing point defects, dislocations and grain boundaries. It does not include arrangements with arbitrarily large voids such as surfaces or cracks. However, as the TFW model for surfaces has been established \cite{Blanc_Uniqueness}, it may be possible to obtain locality estimates for surfaces and cracks using the TFW model.

We remark that there exists $R' = R'(R_{0},\omega) > 0$ such that for any
$Y \subset \mathcal{Y}_{L^{2}}(M,\omega)$
\begin{align}
\bigcup_{j \in \mathbb{N}} B_{R'}(Y_j) = \R. \label{eq: B R' cover ass def}
\end{align}

%
%

For any $Y \in \mathcal{Y}_{L^{2}}(M,\omega)$ there exists a unique ground state
$(u,\phi)$ corresponding to $m = m_{Y}$. Naively, we might define the energy
  stored in a region $\Omega \subset \R$ by
\begin{align}
\int_{\Omega} |\nabla u|^{2} + \int_{\Omega} u^{10/3} + \frac{1}{2} \int_{\Omega} \left( (m-u^{2}) * \smcb \right)(m-u^{2}), \label{eq: inf E TFW energy}
\end{align}
however, difficulties arise due to the fact that $(m - u^{2}) * \smcb$ is not
well-defined.
%
Instead, we give two alternative definitions for the energy density for an infinite system:
\begin{align}
\mathcal{E}_{1}(Y;\cdot) &:= |\nabla u|^{2} + u^{10/3} + \frac{1}{2} \phi(m-u^{2}), \label{eq: TFW energy 1} \\
\mathcal{E}_{2}(Y;\cdot) &:= |\nabla u|^{2} + u^{10/3} + \frac{1}{8\pi}  |\nabla \phi|^{2},
\label{eq: TFW energy 2}
\end{align}
which both satisfy
$\mathcal{E}_{1}(Y;\cdot), \mathcal{E}_{2}(Y;\cdot) \in L^{1}_{\unif}(\R)$.

Suppose now that $\Omega \subset \R$ is a charge-neutral volume
\cite{YuTrinkleBaderVolumes}, that is, if $n$ is the unit normal to $\partial \Omega$,
then $\nabla \phi \cdot n = 0$ on $\partial \Omega$. Recalling from (\ref{eq: phi inf eq}) that
\begin{align*}
- \Delta \phi = 4 \pi ( m - u^{2} )
\end{align*}
we deduce that
\begin{align*}
\frac{1}{8 \pi} \int_{\Omega} |\nabla \phi|^{2} = \frac{1}{8\pi} \int_{\Omega} (- \Delta \phi) \phi + \int_{\partial \Omega} \phi \nabla \phi \cdot n = \frac{1}{2} \int_{\Omega} \phi (m - u^{2}),
\end{align*}
and hence
\begin{align*}
\int_{\Omega} \mathcal{E}_{1}(Y;x) \id x = \int_{\Omega} \mathcal{E}_{2}(Y;x) \id x.
\end{align*}
In particular, for finite neutral systems $m_{R_{n}} = m \cdot \chi_{B_{R_{n}}(0)},$ where $m \in \mathcal{M}_{L^{2}}(M,\omega)$ and $(u_{R_{n}},\phi_{R_{n}})$ denoting the corresponding solution, the following energies agree on $\Omega = \R$
\begin{align}
& \int_{\R} \bigg( |\nabla u_{R_{n}}|^{2} + u_{R_{n}}^{10/3} + \frac{1}{2} \left( (m_{R_{n}}-u_{R_{n}}^{2}) * \smcb \right)(m_{R_{n}}-u_{R_{n}}^{2}) \bigg) \nonumber \\ & \qquad = \int_{\R} \bigg( |\nabla u_{R_{n}}|^{2} + u_{R_{n}}^{10/3} + \frac{1}{2} \phi_{R_{n}} (m_{R_{n}}-u_{R_{n}}^{2}) \bigg) \nonumber \\ & \qquad = \int_{\R} \bigg( |\nabla u_{R_{n}}|^{2} + u_{R_{n}}^{10/3} + \frac{1}{8\pi} |\nabla \phi_{R_{n}}|^{2} \bigg). \label{eq: energy three agree}
\end{align}
We prove this claim is in \Cref{Remark - Energy Definitions are equal} in \Cref{Section 6 - Proofs}.
Thus, we have derived two energy densities, $\mathcal{E}_1, \mathcal{E}_2$,
which are meaningful and well-defined also for infinite configurations.

In order to define site energies, we require a partition of $\R$.  For each
$j \in \mathbb{N}$ let
$\varphi_{j}(Y;\cdot) \in C^1(\R), \varphi_{j}(Y;\cdot) \geq 0$ satisfying the
following conditions: there exist $C, \wt \gamma > 0$ such that for all
$Y \in \mathcal{Y}_{L^{2}}(M,\omega)$
\begin{subequations}
\label{eq: site-E-partition-conditions}
\begin{align}
  \sum_{j \in \mathbb{N}} \varphi_{j}(Y; x) &= 1, \label{eq: sum 1 eq} \\
  |\varphi_{j}(Y;x)| &\leq C e^{- \wt\gamma|x - Y_{j}|}, \label{eq: partition condition 1} 
  \quad \text{and} \\
  \left| \frac{\partial \varphi_{j}}{\partial Y_{k}}(Y;x) \right| &\leq C e^{- \wt\gamma|x - Y_{j}|} e^{- \wt\gamma|x - Y_{k}|}. \label{eq: partition condition 2}
\end{align}
\end{subequations}
We propose a canonical construction of such a partition in
  \Cref{rem:after-site-E-thm} below.

Given a family of partition functions satisfying (\ref{eq:
  site-E-partition-conditions}), we can define site energies
\begin{align}
E^{i}_{j}(Y) = \int_{\R} \mathcal{E}_{i}(Y;x) \varphi_{j}(Y;x) \id x, \label{E ell h def}
\end{align}
for $i = 1, 2$.  A consequence of Theorems \ref{Theorem - One inf pointwise
  stability estimate alt} and \ref{Theorem - Exponential Est Integral RHS} is
that $E^i_j(Y)$ are {\em local}: their dependence on the environment of nuclei
decays exponentially fast. This is made precise in the following theorem.


\begin{theorem}
\label{Theorem - Forcing Est}
Let $i \in \{ 1,2\}$, $Y \in \mathcal{Y}_{L^{2}}(M,\omega)$ and $\{ \varphi_{j} | j \in \mathbb{N} \}$ satisfy \textnormal{(\ref{eq: site-E-partition-conditions})}. Then for every $k \in \mathbb{N}$, $\partial_{Y_k} E^{i}_j$ exists and satisfies
\begin{align}
\left|\frac{\partial E^{i}_{j}(Y)}{\partial Y_{k}}\right| &\leq C e^{- \gamma |Y_{j} - Y_{k}|}, \label{eq: forcing est}
\end{align}
where $C = C(M,\omega)$, $\gamma = \gamma(M,\omega) > 0$.
\end{theorem}


The derivative $\partial_{Y_k} E^{i}_j$ can be interpreted
as the contribution of the atom at $Y_{k}$ to the force on the nucleus at
$Y_{j}$. In addition,
we show in \Cref{sec:proofs energy locality} that these site energies generate the correct total force
\begin{align}
\sum_{j \in \mathbb{N}} \frac{\partial E^{1}_{j}(Y)}{\partial Y_{k}} = \sum_{j \in \mathbb{N}} \frac{\partial E^{2}_{j}(Y)}{\partial Y_{k}} = \int_{\R} \phi(x) \, \frac{\partial m_Y(x)}{\partial Y_k} \id x. \label{eq: Forcing 1 2 equal est}
\end{align}

\begin{remark} \label{rem:after-site-E-thm} Two further canonical requirements
  on a site energy potential are permutation and isometry (rotation and
  translation) invariance. This can be obtained as follows:
  
  If the partition $(\varphi_j)_{j \in \mathbb{N}}$ is {\em permutation
    invariant}, that is, for any bijection $P: \mathbb{N} \to \mathbb{N}$,
  $Y \circ P = (Y_{P j})_{j \in \mathbb{N}}$, we have
  \begin{align}
    \varphi_{j}(Y \circ P;x) = \varphi_{P_{j}}(Y;x)
    \qquad \forall\, j \in \mathbb{N} \quad x \in \R,
    \label{eq: Permutation Invariance def}
  \end{align}
  then so are the site energies,
  \begin{displaymath}
    E^{i}_{j}(Y \circ P) = E^{i}_{j}(Y).
  \end{displaymath}
  
  If the partition is {\em isometry invariant}, that is, for any isometry
  $A: \R \to \R$, $AY = (AY_{j})_{j \in \mathbb{N}}$, we have
  \begin{align}
    \varphi_{j}(AY;x) = \varphi_{j}(Y;A^{-1}x)
    \qquad \forall\, j \in \mathbb{N}, \quad x \in \R, \label{eq: Isometry Invariance def}
  \end{align}
  then the site energies are also isometry invariant, 
  \begin{displaymath}
    E^{i}_{j}(AY) = E^{i}_{j}(Y).
  \end{displaymath}
  Both statements are proven in \Cref{Lemma - Site Energy Invariance}.

  A canonical class of partitions satisfying \eqref{eq:
    site-E-partition-conditions} as well as \eqref{eq: Permutation Invariance
    def}, \eqref{eq: Isometry Invariance def} can be constructed as follows: Let
  $\wt \varphi \in C^{1}(\R)$, $\wt \varphi \geq 0$, be radially symmetric and
  satisfy
  \begin{align}
    &|\wt \varphi (x)| + |\nabla \wt \varphi (x)| \leq C e^{-\wt \gamma |x|}, \label{eq: wt varphi est 1} \\
    &\wt \varphi(x) \geq c > 0 \quad \text{ on } B_{R'+1}(0). \label{eq: wt varphi est 2}
  \end{align}
  For example, this holds for $\wt \varphi(x) = e^{-\wt \gamma |x|^2}$ and for
  standard mollifiers with sufficiently wide support.
  

  Then, for $j \in \mathbb{N}$, we can define
  \begin{align}
    \varphi_{j}(Y;x) = \frac{ \wt \varphi(x - Y_{j}) }{\sum_{j' \in \mathbb{N}} \wt \varphi(x - Y_{j'}) }. \label{eq: exp partition def}
  \end{align}
  It is easy to see that this class of functions are well-defined and satisfies
  all requirements. \qed
\end{remark}

\begin{remark}
Alternative constructions of energy partitions include Bader volumes and charge-neutral volumes \cite{BaderAtomsBook,YuTrinkleBaderVolumes,MarchBookBaderVolumes}.
Bader volumes partition space into regions such that the flux of the electron density on the boundary is zero, while charge-neutral volumes are defined so that each region has zero charge. The construction of these volumes is not unique, like our definition of a partition. Bader volumes were developed as a means to define atoms within molecules \cite{BaderAtomsBook}.

With this in mind, using a partition we may assign a portion of the electron density to each nucleus in the system. We refer to a nucleus paired with it's associated partition of the electron density as an effective atom. Due to the screening that occurs in the TFW model, the interaction of two effective atoms decays exponentially as the distance between the nuclei grows. In comparison, the interaction of two neutrals atoms separated by a sufficiently large distance $r$ in the TF model has been shown to decay at the rate $r^{-7}$~\cite{BrezisLiebLongRangeTF}. This suggests that due to the overscreening of the TFW model, the interaction of the effective atoms is considerably weaker than is realistic. However, for the purpose of simulating quantum systems, in particular applying the strong locality principle \cite{Csanyi}, the weak long-range interaction of the TFW model is a desirable property.
%
%
\qed
\end{remark}
\begin{remark}
The estimate shown in \Cref{Theorem - Forcing Est} is a theoretical result which can be applied to simulate defective crystals, though we do not pursue this. Locality estimates have been established for the tight-binding model and subsequently used to construct QM/MM hybrid methods~\cite{Huajie, Huajie2}. \qed
\end{remark}

\section{Conclusion and Outlook}
\label{Section 5 - Conclusion and Outlook}
The two main results of this work, Theorems \ref{Theorem - One inf pointwise
  stability estimate alt} and \ref{Theorem - Exponential Est Integral RHS}, are
stability and exponential locality estimates for the TFW model, which apply to
general condensed nuclear configurations.


We have demonstrated in \Cref{Section 4 - Applications} that it can be used to
extend and strengthen a range of existing results on the TFW model.  A particular
strength of our results is that they apply to general nuclear configuration in
$\mathcal{M}_{L^{2}}(M,\omega)$, whereas the previous analyses of the TFW model
have focused on (near-)crystalline arrangements or the homogeneous electron
gas. This generality will be valuable when exploring the consequences of our
analysis for studying models for the mechanical response problem in
\cite{Paper2}, where we generalise \cite{EOS} to electronic structure models.


A further application, that we will develop in a forthcoming work is a study of
the Yukawa potential as a model approximation \cite{YukawaPaper}. Adapting Theorems \ref{Theorem - One inf pointwise stability estimate alt} and \ref{Theorem - Exponential Est Integral RHS}
we can consider the difference between the Coulomb and Yukawa ground states for
a given nuclear configuration and prove uniform error estimates in terms of the
screening parameter in the Yukawa model.

Two key difficulties in the analysis of electronic structure models are (i) 
the exchange and correlation of electrons due to the antisymmetry of the electronic wavefunction; and (ii) the interaction of charged particles (positive nuclei and
negative electrons) via the long-range Coulomb potential. Since the TFW model is
orbital-free it does not account for (i), however it fully incorporates Coulomb
interaction. In this regard it is perhaps surprising that the TFW model satisfies
the extremely generic locality property we obtained in Theorems \ref{Theorem - One inf pointwise stability estimate alt} and \ref{Theorem - Exponential Est Integral RHS}.

The Hartree--Fock and Kohn-Sham models take both effects into account and
whether these models permit a similarly strong notion of locality is an open
problem.  It is clear, however, that such results cannot be obtained in the
generality that we considered in the present paper. Since charged defects exist
in the reduced Hartree--Fock model \cite{Cances/Lewin_RHF} and as locality
implies neutrality, this suggests that a locality property cannot hold for
general condensed phase arrangements in the reduced Hartree--Fock model, which
is the simplest model in the Hartree--Fock/Kohn--Sham class.

\section{Proofs}
\label{Section 6 - Proofs}
This section contains the proofs of the main results. Proofs of results in Sections \ref{Subsection - Uniform Regularity Estimates}, \ref{Subsection - Pointwise Stability Estimates} and \ref{Section 4 - Applications} are found in Sections \ref{Subsection - Proofs of Uniform Regularity Estimates}, \ref{Subsection - Proofs of Pointwise Stability Estimates} and \ref{Subsection - Proofs of Applications} respectively.

The following is a preliminary result used in the construction of the space
$\mathcal{M}_{L^2}(M,\omega).$

\begin{lemma}
\label{Lemma - H2 equiv lemma}
Suppose $m : \R \to \mathbb{R}_{\geq 0}$ and $m \in L^{1}_{\loc}(\R)$, then
\textnormal{(H2)} is equivalent to the following statement: there exist $\omega_{0}, \omega_{1} > 0$ such that for all $R > 0$
\begin{align}
\inf_{x \in \R} \int_{B_{R}(x)} m(z) \id z \geq \omega_{0} R^{3} - \omega_{1}. \label{eq: H2 equiv any omega}
\end{align}
\end{lemma}

\begin{proof}[Proof of \Cref{Lemma - H2 equiv lemma}]
Clearly, (\ref{eq: H2 equiv any omega}) implies (H2), so suppose $m$ satisfies (H2), then there exists $R_{1} > 0$ such that
\begin{align*}
\inf_{x \in \R} \int_{B_{R_{1}}(x)} m(z) \id z \geq 1.  
\end{align*} 
For $R > 0$ and $x' \in \R$, let $Q_{R}(x') \subset \R$ denote the cube of side length $2R$ centred at $x'$, which contains $B_{R}(x')$. Also, let $R_{2} = \sqrt{3} R_{1}$, which ensures that $\overline{B_{R_{2}}(x)} \supset Q_{R_{1}}(x)$ for all $x \in \R$. Further, let $R \geq R_{2}$, hence $R = k R_{2}$, for some $k \geq 1$. Then
\begin{align}
\inf_{x \in \R} \int_{B_{R}(x)} m(z) \id z &= \inf_{x \in \R} \int_{B_{kR_{2}}(x)} m(z) \id z \geq \inf_{x \in \R} \int_{Q_{kR_{1}}(x)} m(z) \id z \nonumber \\ &\geq \lfloor k \rfloor^{3} \inf_{x' \in \R} \int_{Q_{R_{1}}(x')} m(z) \id z \geq \lfloor k \rfloor^{3} \nonumber 
  \\
&\geq \left( \frac{k}{2} \right)^{3} 
  = \frac{R^{3}}{8 R_{2}^{3}} =: \omega_{0} R^{3}. \label{eq: omega0 est}
\end{align} 
Now define $\omega_{1} := \omega_{0} R_{2}^{3} \geq 0$, then it follows from (\ref{eq: omega0 est}) that (\ref{eq: H2 equiv any omega}) holds for all $R > 0$. \qedhere 
\end{proof}


\subsection{Proofs of Uniform Regularity Estimates}
\label{Subsection - Proofs of Uniform Regularity Estimates}
The following lemma features in the proofs of both the existence and uniqueness of the TFW equations and is found in \cite{C/LB/L}.
\begin{lemma}
\label{Lemma - Eigenvalue Argument}
Let $a \in H^{1}_{\loc}(\R) \cap L^{\infty}(\R)$, then define the elliptic operator $L = - \Delta + a$. Suppose that there exists $u \in H^{1}_{\loc}(\R)$ satisfying $u > 0$ and $L u = 0$ in distribution. Then, the operator $L$ is non-negative, that is for all $\varphi \in H^{1}(\R)$
\begin{align}
\langle \varphi, L \varphi \rangle \geq 0. \label{eq: L eig non-neg result}
\end{align}
\end{lemma}

The proof is shown in \cite{C/LB/L} but is included here for completeness.
\begin{proof}[Proof of \Cref{Lemma - Eigenvalue Argument}]
Let $R > 0$ and define $\Omega = B_{R}(0)$ and consider $L$ as an operator on $L^{2}(\Omega)$ with domain $H^{2}(\Omega) \cap H^{1}_{0}(\Omega)$. Then $L$ is a self-adjoint operator with compact resolvent hence has a purely discrete spectrum. Since $a \in H^{1}_{\loc}(\R)$ it follows that the smallest eigenvalue $\lambda_{1}(\Omega)$ is simple and has a positive eigenfunction $v_{\Omega} \in H^{1}_{0}(\Omega)$ \cite[Theorem 8.38]{Gilbarg/Trudinger}. In addition, by standard elliptic regularity $v_{\Omega} \in H^{3}(\Omega) \hookrightarrow C^{1,1/2}(\overline{\Omega})$ \cite{Evans} and solves
\begin{align*}
\left ( - \Delta + a \right) v_{\Omega} = \lambda_{1}(\Omega) v_{\Omega}.
\end{align*}
Testing this equation with $u$ and using integration by parts, we obtain
\begin{align}
- \int_{\partial \Omega} \frac{\partial v_{\Omega}}{\partial n} u = \lambda_{1}(\Omega) \int_{\Omega} v_{\Omega} u. \label{eq: Eigenvalue integral eq}
\end{align}
As $v_{\Omega} > 0$ on $\Omega$ and $v_{\Omega}$ vanishes over $\partial \Omega$, it follows that $\frac{\partial v_{\Omega}}{\partial n} \leq 0$. It follows that the left-hand side of (\ref{eq: Eigenvalue integral eq}) is non-negative, hence $\lambda_{1}(\Omega) \geq 0$. As this holds for $\Omega = B_{R}(0)$, for any $R > 0$, we deduce that for all $\varphi \in C^{1}_{c}(\R)$ $\langle \varphi, L \varphi \rangle \geq 0.$ Using that $a \in L^{\infty}(\R)$ and the density of $C^{1}_{c}(\R)$ in $H^{1}(\R)$, it follows that for all $\varphi \in H^{1}(\R)$ $\langle \varphi, L \varphi \rangle \geq 0.$
\end{proof}

We now show uniform estimates for finite systems corresponding to truncated
nuclear distributions. This result is essentially \cite[Proposition
3.5]{C/LB/L}, however as we require uniform regularity estimates, we provide a
complete proof.

\begin{proposition}
\label{Proposition - Finite Regularity Est}
Let $m : \mathbb{R} \to \mathbb{R}_{\geq 0}$ satisfy 
\begin{align}
\| m \|_{L^{2}_{\unif}(\R)} \leq M, \label{eq: m L2 est}
\end{align}
and $R_{n} \uparrow \infty$, then define the truncated nuclear distribution $m_{R_{n}} = m \cdot \chi_{B_{R_{n}}(0)}$. 
The unique solution to the minimisation problem
\begin{align*}
I^{\TFW}(m_{R_{n}}) = \inf \left\{ \, E^{\TFW}(v,m_{R_{n}})  \, \bigg| \, v \in H^{1}(\R), v \geq 0, \int_{\R} v^{2} = \int_{\R} m_{R_{n}} 
\, \right \},
\end{align*}
yields a unique solution $(u_{R_{n}},\phi_{R_{n}})$ to \textnormal{(\ref{eq: u phi Rn pair})}
\begin{subequations}
\begin{align*}
- \Delta u_{R_{n}} &+ \frac{5}{3} u_{R_{n}}^{7/3} - \phi_{R_{n}} u_{R_{n}} = 0, \\
- \Delta \phi_{R_{n}} &= 4 \pi ( m_{R_{n}} - u_{R_{n}}^{2} ). 
\end{align*}
\end{subequations}
which satisfy the following estimates, with constant $C$ independent of $R_{n}$:
\begin{align}
\| u_{R_{n}} \|_{H^{4}_{\unif}(\R)} &\leq C(1 + M^{15/4}), \label{eq: uRn H4 unif est} \\
\| \phi_{R_{n}} \|_{H^{2}_{\unif}(\R)} &\leq C (1 + M^{3/2} ). \label{eq: phiRn H2 unif est}
\end{align}
\end{proposition}

\begin{proof}[Proof of \Cref{Proposition - Finite Regularity Est}]
If $m \equiv 0$, then for all $R_{n}$, clearly $u_{R_{n}} = \phi_{R_{n}} = m_{R_{n}} = 0$ satisfies (\ref{eq: u phi Rn pair}) and (\ref{eq: uRn H4 unif est})--(\ref{eq: phiRn H2 unif est}).

If $m \not\equiv 0$, then there exists a constant $R_{0} \geq 0$ such
that $R_{n} \geq R_{0}$ ensures that $\int_{\R} m_{R_{n}} > 0$. Recall
\begin{align*}
E^{\TFW}_{R_{n}} (v,m_{R_{n}}) &= \int |\nabla v|^{2} + \int v^{10/3} + \frac{1}{2} D( m_{R_{n}} - v^{2}, m_{R_{n}} - v^{2}).
\end{align*}
For each $R_{n}$, consider the minimisation problem
\begin{align*}
I^{\TFW}(m_{R_{n}}) = \inf \left\{ \, E^{\TFW}(v,m_{R_{n}})  \, \bigg| \, v \in H^{1}(\R), v \geq 0, \int_{\R} v^{2} = \int_{\R} m_{R_{n}} > 0 \, \right \}.
\end{align*}
The constraint $\int_{\R} v^{2} = \int_{\R} m_{R_{n}}$ ensures the system is charge neutral, and by \cite[Theorem 7.19]{Lieb_Summary} there exists a unique non-negative minimiser $u_{R_{n}} \in H^{1}(\R)$ to $I^{\TFW}(m_{R_{n}})$ solving
\begin{align}
- \Delta u_{R_{n}} + \frac{5}{3} u_{R_{n}}^{7/3} &- \left( (m_{R_{n}} - u_{R_{n}}^{2}) * \smcb \right) u_{R_{n}} = - \theta_{R_{n}} u_{R_{n}}, \label{eq: un initial eq} \\
&\int_{\R} u_{R_{n}}^{2} = \int_{\R} m_{R_{n}} > 0. \label{eq: un charge constraint eq}
\end{align}
Here $\theta_{R_{n}} > 0$ is the Lagrange multiplier associated with the charge constraint (\ref{eq: un charge constraint eq}) \cite{Lieb_Summary,C/LB/L}. Define $\phi_{R_{n}} : \R \to \mathbb{R}$ by
\begin{align}
\phi_{R_{n}} = \left( (m_{R_{n}} - u_{R_{n}}^{2}) * \smcb \right) - \theta_{R_{n}}, \label{eq: phi n definition}
\end{align}
so we can express (\ref{eq: un initial eq}) as the \Schrodinger--Poisson system (\ref{eq: u phi Rn pair})
\begin{align*}
- \Delta u_{R_{n}} &+ \frac{5}{3} u_{R_{n}}^{7/3} - \phi_{R_{n}} u_{R_{n}} = 0, \\
- \Delta \phi_{R_{n}} &= 4 \pi ( m_{R_{n}} - u_{R_{n}}^{2} ).
\end{align*}
Decompose
\begin{align*}
(m_{R_{n}} - u_{R_{n}}^{2}) * \smcb &= (m_{R_{n}} - u_{R_{n}}^{2}) * \left( \smcb \chi_{B_{1}(0)} \right) + (m_{R_{n}} - u_{R_{n}}^{2}) * \left( \smcb \chi_{B_{1}(0)^{\rm c}} \right),
\end{align*}
then as $u_{R_{n}} \in H^{1}(\R) \hookrightarrow L^{6}(\R)$ and $m \in L^{2}_{\unif}(\R)$  applying Young's inequality gives
\begin{align*}
\left \| (m_{R_{n}} - u_{R_{n}}^{2}) * \smcb \right \|_{L^{\infty}(\R)} \nonumber &\leq \| (m_{R_{n}} - u_{R_{n}}^{2}) \|_{L^{5/3}(\R)} \left \| \smcb \chi_{B_{1}(0)} \right \|_{L^{5/2}(\R)} \\ & \quad + \| (m_{R_{n}} - u_{R_{n}}^{2}) \|_{L^{7/5}(\R)} \left \| \smcb \chi_{B_{1}(0)^{\rm c}} \right \|_{L^{7/2}(\R)} \nonumber \\
&\leq C \left( (R_{n}^{3/10} + R_{n}^{9/14}) \| m_{R_{n}} \|_{L^{2}(\R)} + \| u_{R_{n}} \|_{H^{1}(\R)}^{2} \right) \nonumber \\
&\leq C \left( (R_{n}^{9/5} + R_{n}^{15/7}) \| m \|_{L^{2}_{\unif}(\R)} + \| u_{R_{n}} \|_{H^{1}(\R)}^{2} \right) \\
&\leq C \left( (R_{n}^{9/5} + R_{n}^{15/7}) M + \| u_{R_{n}} \|_{H^{1}(\R)}^{2} \right).
\end{align*}
By \cite[Lemma II.25]{Lieb/Simon_TF} we deduce that $(m_{R_{n}} - u_{R_{n}}^{2}) * \frac{1}{|\cdot|}$ is a continuous function vanishing at infinity. It follows that $\phi_{R_{n}} \in L^{\infty}(\R)$ and is also continuous. Also, $|\nabla \phi_{R_{n}}| \in L^{2}(\R)$
\begin{align}
\frac{1}{8\pi} \int_{\R} |\nabla \phi_{R_{n}}|^{2} 
= &\,\,\frac{1}{8\pi} \int_{\R} \phi_{R_{n}} \left( -\Delta \phi_{R_{n}} \right) \nonumber \\
= &\,\,\frac{1}{2} \int_{\R} \phi_{R_{n}} (m_{R_{n}}-u_{R_{n}}^{2}) 
  \nonumber \\
= &\,\,\frac{1}{2} \int_{\R} \phi_{R_{n}} (m_{R_{n}}-u_{R_{n}}^{2}) + \frac{\theta_{R_{n}}}{2} \int_{\R}(m_{R_{n}}-u_{R_{n}}^{2}) \nonumber \\
= &\,\, \frac{1}{2} \int_{\R} \left( \phi_{R_{n}} + \theta_{R_{n}} \right)(m_{R_{n}}-u_{R_{n}}^{2}) \nonumber \\
= &\,\,\frac{1}{2} \int_{\R} \left( (m_{R_{n}}-u_{R_{n}}^{2}) * \smcb \right)(m_{R_{n}}-u_{R_{n}}^{2}), \label{eq: energy 1 eq}
\end{align}
hence $\phi_{R_{n}} \in H^{1}_{\unif}(\R)$.
Now, consider $u_{R_{n}} \in H^{1}(\R)$, which solves
\begin{align}
- \Delta u_{R_{n}} &=  -\frac{5}{3} u_{R_{n}}^{7/3} + \phi_{R_{n}} u_{R_{n}}. \label{eq: Delta uRn = rhs eq }
\end{align}
The right-hand side can be estimated in $L^{2}(\R)$ by
\begin{align*}
\| \smfrac{5}{3} u_{R_{n}}^{7/3} - \phi_{R_{n}} u_{R_{n}}  \|_{L^{2}(\R)} &\leq
\frac{5}{3} \| u_{R_{n}}^{7/3} \|_{L^{2}(\R)} + \| \phi_{R_{n}} \|_{L^{\infty}(\R)} \| u_{R_{n}} \|_{L^{2}(\R)} \\
&\leq C \| u_{R_{n}} \|_{H^{1}(\R)}^{7/3} + \| \phi_{R_{n}} \|_{L^{\infty}(\R)} \| u_{R_{n}} \|_{H^{1}(\R)},
\end{align*}
which implies $u_{R_{n}} \in H^{2}(\R)$ as $\Delta u_{R_{n}} \in L^{2}(\R)$. By the Sobolev Embedding Theorem \cite{Evans} $u_{R_{n}} \in H^{2}(\R) \hookrightarrow C^{0,1/2}(\R)$, hence $u_{R_{n}}$ is continuous. Also, by \cite[Lemma 9]{B/B/L}, $u_{R_{n}}$ decays at infinity. We now justify this. Recall (\ref{eq: Delta uRn = rhs eq }) and since $u_{R_{n}} \geq 0$, we have $- \Delta u_{R_{n}} \leq \phi_{R_{n}} u_{R_{n}}$,
hence
\begin{align*}
- \Delta u_{R_{n}} + u_{R_{n}} \leq ( 1 + \phi_{R_{n}} ) u_{R_{n}}
\end{align*}
As $\phi_{R_{n}} \in L^{\infty}(\R)$ and $u_{R_{n}} \in H^{1}(\R)$, the right-hand side belongs to $L^{2}(\R)$ hence by the Lax-Milgram theorem there exists a unique $g_{R_{n}} \in H^{1}(\R)$ satisfying
\begin{align*}
- \Delta g_{R_{n}} + g_{R_{n}} = (1 + \phi_{R_{n}}) u_{R_{n}}.
\end{align*}
Moreover, using the Green's function $g_{R_{n}} = \smfrac{e^{-|\cdot|}}{|\cdot|} * (1 + \phi_{R_{n}}) u_{R_{n}}$ and since $\smfrac{e^{-|\cdot|}}{|\cdot|}, (1 + \phi_{R_{n}}) u_{R_{n}} \in L^{2}(\R)$ by \cite[Lemma II.25]{Lieb/Simon_TF} $g_{R_{n}}$ is continuous function that decays at infinity, hence $g_{R_{n}} \in L^{\infty}(\R)$. It follows from the comparison principle that $u_{R_{n}} \leq g_{R_{n}}$, so $u_{R_{n}} \in L^{\infty}(\R)$ and decays at infinity.

Using that $u_{R_{n}},\phi_{R_{n}} + \theta_{R_{n}}$ are continuous and decay at infinity, by arguing as in \cite{Solovej_Universality}, there exists a universal constant $C_{S} > 0$, independent of the nuclear distribution, satisfying
\begin{align}
0 &< \theta_{R_{n}} \leq C_{S}, \\
\frac{10}{9} u_{R_{n}}^{4/3} &\leq \phi_{R_{n}} + C_{S}. \label{eq: Solovej est}
\end{align}
As $u_{R_{n}} \geq 0$, from the Solovej estimate (\ref{eq: Solovej est}) we obtain the uniform lower bound
\begin{align}
\phi_{R_{n}} \geq - C_{S}. \label{eq: phi Rn lower bound}
\end{align}
We aim to show a uniform upper bound for $\phi_{R_{n}}$, which together with (\ref{eq: Solovej est}) will yield the uniform estimate
\begin{align}
\| u_{R_{n}} \|_{L^{\infty}(\R)}^{4/3} + \| \phi_{R_{n}} \|_{L^{\infty}(\R)} \leq C(M), \label{eq: u Rn phi Rn uniform L inf est}
\end{align}
which is independent of $R_{n}$. 

If $\phi_{R_{n}}$ is non-positive, then (\ref{eq: u Rn phi Rn uniform L inf est}) holds as
\begin{align*}
\| u_{R_{n}} \|_{L^{\infty}(\R)}^{4/3} + \| \phi_{R_{n}} \|_{L^{\infty}(\R)} \leq 2C_{S}.
\end{align*}
Instead, suppose that $\phi^{+}_{R_{n}}$ is non-zero at some point in $\R$. By (\ref{eq: phi n definition}) $\phi_{R_{n}}$ is a continuous function that converges to a negative limit at infinity, $\phi_{R_{n}}^{+} \in C_{c}(\R)$, hence there exists a point $x_{R_{n}} \in \R$ such that
\begin{align}
\phi_{R_{n}}^{+}(x_{R_{n}}) = \| \phi_{R_{n}}^{+} \|_{L^{\infty}(\R)} > 0. \label{eq: phi Rn sup xn = 0}
\end{align}
Without loss of generality, we assume $x_{R_{n}} = 0$.

We now show that $u_{R_{n}} > 0$ on $\R$, arguing by contradiction. Suppose that there exists $z \in \R$ such that $u_{R_{n}}(z) = 0$. Since $u_{R_{n}}$ is a non-negative, continuous function decaying at infinity, there exists $y_{n} \in \R$ such that
\begin{align*}
u_{R_{n}}(y_{n}) = \sup_{x \in \R}u_{R_{n}}(x).
\end{align*}
Let $R > |y_{n} - z|$, then by the Harnack inequality
\cite{Trudinger_MeasurableCoefficients}, we infer
\begin{align*}
0 \leq u_{R_{n}}(y_{n}) = \sup_{x \in B_{R}(y_{n})} u_{R_{n}}(x) \leq C(R) \inf_{x \in B_{R}(y_{n})} u_{R_{n}}(x) = u_{R_{n}}(z) = 0,
\end{align*}
so $u_{R_{n}} \equiv 0$. This contradicts the charge constraint (\ref{eq: un charge constraint eq}) $\int_{\R} u_{R_{n}}^{2} = \int_{\R} m_{R_{n}} > 0$, hence $u_{R_{n}} > 0$ on $\R$.

As $u_{R_{n}} \in H^{1}(\R) \cap L^{\infty}(\R), \phi_{R_{n}} \in H^{1}_{\unif}(\R) \cap L^{\infty}(\R)$ and $u_{R_{n}} > 0$, \Cref{Lemma - Eigenvalue Argument} implies that $L_{R_{n}} = - \Delta + \frac{5}{3} u_{R_{n}}^{4/3} - \phi_{R_{n}}$ is a non-negative operator.

Choose $\varphi \in C^{\infty}_{\textnormal{c}}(B_{1}(0))$ satisfying $0 \leq \varphi \leq 1$, $\varphi = 1$ on $B_{1/2}(0)$, $\int_{\R} \varphi^{2} = 1$ and $\int_{\R} |\nabla \varphi|^{2} =: c_{\varphi}$, then for $y \in \R$, define $\varphi_{y} \in C^{\infty}_{\textnormal{c}}(B_{1}(y))$ by $\varphi_{y} = \varphi(\cdot - y)$. As $L_{R_{n}}$ is non-negative (\ref{eq: L eig non-neg result}) implies
\begin{align*}
\langle \varphi_{y}, L_{R_{n}} \varphi_{y} \rangle = \int_{\R} |\nabla \varphi_{y}|^{2} + \int_{\R} \left( \frac{5}{3} u_{R_{n}}^{4/3} - \phi_{R_{n}} \right) \varphi_{y}^{2} \geq 0,
\end{align*}
which can be re-arranged and expressed using convolutions as
\begin{align}
\frac{5}{3} \left( u_{R_{n}}^{4/3} * \varphi^{2} \right) &\geq  \left( \phi_{R_{n}} * \varphi^{2} - \int_{\R} |\nabla \varphi|^{2} \right)_{+} \nonumber \\
&= \left( \phi_{R_{n}} * \varphi^{2} - c_{\varphi} \right)_{+} \label{eq: phi Rn first convolution est}
\end{align}
Observe that $\phi_{R_{n}} * \varphi^{2}$ solves
\begin{align}
- \Delta \left( \phi_{R_{n}} * \varphi^{2} \right) &= 4 \pi \left( m_{R_{n}} * \varphi^{2} - u^{2}_{R_{n}} * \varphi^{2}  \right). \label{eq: Delta phi Rn convolution est}
\end{align}
We estimate the first term using (\ref{eq: m L2 est})
\begin{align}
\left(m_{R_{n}} * \varphi^{2} \right)(x) &= \int_{B_{1}(x)} m_{R_{n}}(y) \varphi^{2}(x - y) \id y \nonumber \\ &\leq \int_{B_{1}(x)} m(y) \id y \leq C_{0} \| m \|_{L^{2}_{\unif}(\R)} \leq C_{0} M. \label{eq: m Rn convolution est}
\end{align}
For the second term, using the convexity of $t \mapsto t^{3/2}$ for $t \geq 0$ and the fact that $\int \varphi^{2} = 1$, applying Jensen's inequality and (\ref{eq: phi Rn first convolution est}) we deduce
\begin{align}
4 \pi \, u_{R_{n}}^{2} * \varphi^{2}(x) &\geq \frac{5}{3} u_{R_{n}}^{2} * \varphi^{2}(x) \nonumber \\
&= \frac{5}{3} \int_{\R} u_{R_{n}}^{2}(x-y) \varphi^{2}(y) \id y \nonumber \\
&= \frac{5}{3} \int_{\R} \left(u_{R_{n}}^{4/3}(x-y)\right)^{3/2} \varphi^{2}(y) \id y \nonumber \\
&\geq \frac{5}{3} \left( \int_{\R} u_{R_{n}}^{4/3}(x-y) \varphi^{2}(y) \id y \right)^{3/2} \nonumber \\
&= \frac{5}{3}(u^{4/3}_{R_{n}} * \varphi^{2})^{3/2} \geq \left( \phi_{R_{n}} * \varphi^{2} - c_{\varphi}  \right)^{3/2}_{+}. \label{eq: phi Rn convolution est}
\end{align}
Combining the estimates (\ref{eq: Delta phi Rn convolution est}) - (\ref{eq: phi Rn convolution est}) we conclude that
\begin{align*}
- \Delta \left( \phi_{R_{n}} * \varphi^{2} \right) + \left( \phi_{R_{n}} * \varphi^{2} - c_{\varphi} \right)^{3/2}_{+} \leq C_{0} M.
\end{align*}
 By (\ref{eq: phi n definition}), as $\phi_{R_{n}}$ is a continuous function that converges to a negative limit at infinity, $\phi_{R_{n}} * \varphi^{2}$ also shares these properties. Now consider the set
\begin{align*}
S = \{ \, x \in \R \, | \, \phi_{R_{n}} * \varphi^{2} - c_{\varphi} > 0 \, \},
\end{align*}
it follows that $S$ is open and bounded and that $\phi_{R_{n}} * \varphi^{2} - c_{\varphi} = 0$ on $\partial S$. Observe that the constant function $h = (C_{0}M)^{2/3}$ satisfies
\begin{align*}
- &\Delta h + h^{3/2}_{+} = C_{0}M \quad \text{ on } S, \\
0 = \, &\phi_{R_{n}} * \varphi^{2} - c_{\varphi} \leq h \quad \text{ in } \partial S,
\end{align*}
so by the maximum principle $\phi_{R_{n}} * \varphi^{2} \leq (c_{\varphi} + C_{0}^{2/3} M^{2/3})$
over $S$, and also on $S^{\rm c}$, hence
\begin{align}
\phi_{R_{n}} * \varphi^{2} \leq C_{1}(1 + M^{2/3}), \label{eq: phi Rn M convolution est}
\end{align}
where $C_{1} = \max\{ c_{\varphi}, C_{0}^{2/3} \}$ is independent of $M$.

Observe that
\begin{align*}
\phi_{R_{n}}^{+} * \varphi^{2} &= \phi_{R_{n}}^{-} * \varphi^{2} + \phi_{R_{n}} * \varphi^{2} \leq C_{S} + C_{1}(1 + M^{2/3})  = C(1 + M^{2/3}),
\end{align*}
and by estimating (\ref{eq: u phi eqs, phiRn}) directly, that
\begin{align*}
-\Delta \phi_{R_{n}}^{+} = -\Delta \phi_{R_{n}} \chi_{\{\phi_{R_{n}} > 0\}} = 4 \pi \left( m_{R_{n}} - u_{R_{n}}^{2} \right) \chi_{\{\phi_{R_{n}} > 0\}} \leq 4 \pi m_{R_{n}} \chi_{\{\phi_{R_{n}} > 0\}}  \leq 4 \pi m_{R_{n}}.
\end{align*}
As $0 \leq \varphi \leq 1$ and $\varphi = 1$ on $B_{1/2}(0)$, then
\begin{align}
\int_{B_{1/2}(0)} \phi^{+}_{R_{n}}(x) \id x \leq \left(\phi_{R_{n}}^{+} * \varphi^{2}\right)(0) \leq C(1 + M^{2/3}). \label{eq: phi + B1/2 int est}
\end{align}
Using a change of variables, (\ref{eq: phi + B1/2 int est}) can be expressed as
\begin{align*}
\int_{B_{1/2}(0)} \phi^{+}_{R_{n}}(x) \id x = \int_{0}^{1/2} \int_{S_{t}(0)} \phi^{+}_{R_{n}}(t \gamma) \id S_{t}(\gamma) \id t.
\end{align*}
Define $f: [0,1/2] \to \mathbb{R}$ by
\begin{align*}
f(t) = \int_{S_{t}(0)} \phi^{+}_{R_{n}}(t \gamma) \id S_{t}(\gamma)
\end{align*}
and suppose that for all $t \in (1/4,1/2)$
\begin{align*}
f(t) > 4 \int_{B_{1/2}(0)} \phi^{+}_{R_{n}}(x) \id x > 0,
\end{align*}
then
\begin{align*}
\int_{B_{1/2}(0)} \phi^{+}_{R_{n}}(x) \id x = \int_{0}^{1/2} f(t) \id t \geq \int_{1/4}^{1/2} f(t) \id t > \int_{B_{1/2}(0)} \phi^{+}_{R_{n}}(x) \id x,
\end{align*}
which gives a contradiction, hence for some $t \in (1/4,1/2)$
\begin{align}
\int_{S_{t}(0)} \phi^{+}_{R_{n}}(t \gamma) \id S_{t}(\gamma) \leq 4 \int_{B_{1/2}(0)} \phi^{+}_{R_{n}}(x) \id x \leq  C(1 + M^{2/3}). \label{eq: phi + St int est}
\end{align}
Since $t > 1/4$, (\ref{eq: phi + St int est}) implies
\begin{align*}
\dashint_{S_{t}(0)} \phi^{+}_{R_{n}}(t \gamma) \id S_{t}(\gamma) &= \frac{1}{|S_{t}(0)|} \int_{S_{t}(0)} \phi^{+}_{R_{n}}(t \gamma) \id S_{t}(\gamma) \\ &\leq \frac{ C + M^{2/3}   }{|S_{1/4}(0)|} \leq C ( 1 + M^{2/3} ) =: C_{1}(M).
\end{align*}
We now construct an upper bound for $\phi^{+}_{R_{n}}$ as follows. Let $\phi_{1}$ satisfy
\begin{align*}
- \Delta \phi_{1} &= 0 \qquad \,\, \text{ in } B_{t}(0), \\
\phi_{1} &= \phi^{+}_{R_{n}} \quad \text{ on } S_{t}(0).
\end{align*}
As $\phi_{1}$ is harmonic, it satisfies the mean value property
\begin{align}
\phi_{1}(0) \leq \dashint_{S_{t}(0)} \phi^{+}_{R_{n}}(t \gamma) \id S_{t}(\gamma) \leq C_{1}(M). \label{eq: phi 1 est}
\end{align}
Then consider the Dirichlet problem
\begin{align*}
- \Delta \phi_{2} &= 4 \pi m \quad \text{ in } B_{t}(0), \\
\phi_{2} &= 0 \qquad \,\,\, \text{ on } S_{t}(0).
\end{align*}
By Lax-Milgram, this has a unique weak solution $\phi_{2} \in H^{1}_{0}(B_{t}(0))$. By standard elliptic regularity theory \cite{Evans} $\phi_{2} \in H^{2}(B_{t}(0)) \hookrightarrow C^{0,1/2}(\overline{B_{t}(0)})$ and
\begin{align}
\| \phi_{2} \|_{C^{0,1/2}(\overline{B_{t}(0)})} &\leq C \| \phi_{2} \|_{H^{2}(B_{t}(0))} \leq C \| m \|_{L^{2}(B_{t}(0))} \leq C t^{3/2} \| m \|_{L^{2}_{\unif}(\R)} \leq C M. \label{eq: phi 2 est}
\end{align}
The constructed functions $\phi_{1}, \phi_{2}$ satisfy
\begin{align*}
- \Delta \phi^{+}_{R_{n}} &\leq -\Delta (\phi_{1} + \phi_{2}) \quad \text{ in } B_{t}(0), \\
\phi^{+}_{R_{n}} &= \phi_{1} + \phi_{2} \,\, \quad \qquad \text{ on } S_{t}(0),
\end{align*}
hence by the maximum principle $\phi^{+}_{R_{n}} \leq \phi_{1} + \phi_{2}$, in particular (\ref{eq: phi 1 est})--(\ref{eq: phi 2 est}) imply
\begin{displaymath}
  \| \phi^{+}_{R_{n}} \|_{L^{\infty}(\R)} = \phi^{+}_{R_{n}}(0) \leq \phi_{1}(0) 
  + \phi_{2}(0) \leq C(1 + M), 
\end{displaymath}
where the right-hand side is independent of $R_{n}$. Combining this with the lower bound (\ref{eq: phi Rn lower bound}) and the Solovej estimate (\ref{eq: Solovej est}), we obtain the estimate (\ref{eq: u Rn phi Rn uniform L inf est})
\begin{align*}
\| u_{R_{n}} \|_{L^{\infty}(\R)}^{4/3} + \| \phi_{R_{n}} \|_{L^{\infty}(\R)} \leq C(1 + M).
\end{align*}
It follows immediately that for all $x \in \R$ and $r \in [1,\infty]$
\begin{align}
\| u_{R_{n}} \|_{L^{r}(B_{2}(x))} \leq C(1 + M^{3/4}),  \label{eq: u Rn any Lp est}
\end{align}
independently of both $x, r$ and $R_{n}$. Using (\ref{eq: u Rn phi Rn uniform L inf est}) and (\ref{eq: u Rn any Lp est}), we now obtain uniform local estimates for the right-hand side of (\ref{eq: Delta uRn = rhs eq })
\begin{align*}
- \Delta u_{R_{n}} &=  -\frac{5}{3} u_{R_{n}}^{7/3} + \phi_{R_{n}} u_{R_{n}}
\end{align*}
by
\begin{align*}
\| \smfrac{5}{3} u_{R_{n}}^{7/3} - \phi_{R_{n}} u_{R_{n}} \|_{L^{2}(B_{2}(x))} &\leq 
C \| \smfrac{5}{3} u_{R_{n}}^{7/3} - \phi_{R_{n}} u_{R_{n}} \|_{L^{\infty}(\R)} \\
&\leq C ( \| u_{R_{n}} \|_{L^{\infty}(\R)}^{7/3} + \| \phi_{R_{n}} \|_{L^{\infty}(\R)} \| u_{R_{n}} \|_{L^{\infty}(\R)} ) \\
&\leq C ( 1 + M^{7/4} ).
\end{align*}
Consequently, for any $x \in \R$, the elliptic regularity estimate \cite{Evans} gives
\begin{align}
\| u_{R_{n}} \|_{H^{2}(B_{1}(x))} &\leq C ( \| \smfrac{5}{3} u_{R_{n}}^{7/3} - \phi_{R_{n}} u_{R_{n}} \|_{L^{2}(B_{2}(x))} + \| u_{R_{n}} \|_{L^{2}(B_{2}(x))} ) \nonumber \\
&\leq C( 1 + M^{7/4}) + C( 1 + M^{1/2} ) \leq C( 1 + M^{7/4}). \label{eq: u Rn H2 B1 x est}
\end{align}
As (\ref{eq: u Rn H2 B1 x est}) is independent of $x \in \R$, we obtain
\begin{align}
\| u_{R_{n}} \|_{H^{2}_{\unif}(\R)} \leq C( 1 + M^{7/4}).  \label{eq: uRn H2 unif est}
\end{align}
Applying a similar argument to estimate the right-hand side of (\ref{eq: u phi eqs, phiRn})
\begin{align*}
- \Delta \phi_{R_{n}} &= 4 \pi ( m_{R_{n}} - u_{R_{n}}^{2} )
\end{align*}
yields (\ref{eq: phiRn H2 unif est})
\begin{align*}
\| \phi_{R_{n}} \|_{H^{2}_{\unif}(\R)} \leq C ( 1 + M^{3/2} ).
\end{align*}
Using that $\phi_{R_{n}} \in H^{2}_{\unif}(\R)$ and arguing as in (\ref{eq: uRn H2 unif est}), we obtain the desired estimate (\ref{eq: uRn H4 unif est})
\begin{equation*}
\| u_{R_{n}} \|_{H^{4}_{\unif}(\R)} \leq C( 1 + M^{15/4}). \qedhere
\end{equation*}
We remark that while the constants appearing in the final estimates \eqref{eq: uRn H4 unif est}--\eqref{eq: phiRn H2 unif est} depend on $c_{\varphi}$, they are independent of $M$.
\end{proof}

\begin{remark}
\label{Remark - Energy Definitions are equal}
We now justify the claim that for finite and neutral systems and for $\Omega = \R$, the three energies shown in \eqref{eq: energy three agree} agree. Recall (\ref{eq: energy 1 eq}), which shows that the Coulomb energy can be expressed as
\begin{align*}
\frac{1}{2} \int_{\R} \left( (m_{R_{n}}-u_{R_{n}}^{2}) * \smcb \right)(m_{R_{n}}-u_{R_{n}}^{2}) = \frac{1}{2} \int_{\R} \phi_{R_{n}} (m_{R_{n}}-u_{R_{n}}^{2}) = 
\frac{1}{8\pi} \int_{\R} |\nabla \phi_{R_{n}}|^{2},
\end{align*}
so it follows that the energies defined in \eqref{eq: energy three agree} agree for $\Omega = \R$. \qed
\end{remark}

We now discuss passing to the limit in (\ref{eq: u phi Rn pair}) to obtain regularity for the infinite system.

\begin{proof}[Proof of \Cref{Proposition - General Regularity Est}]
First suppose that $\spt(m)$ is bounded, then for sufficiently large $R_{n}$, $m = m_{R_{n}}$ and hence by \Cref{Proposition - Finite Regularity Est} $(u,\phi) = (u_{R_{n}},\phi_{R_{n}})$ solves (\ref{eq: u phi eq pair}) and satisfies the desired estimates (\ref{eq: gen u H2 unif est})--(\ref{eq: gen phi H2 unif est}).

Now suppose $\spt(m)$ is unbounded, then
the estimates (\ref{eq: uRn H4 unif est})--(\ref{eq: phiRn H2 unif est}) of \Cref{Proposition - Finite Regularity Est} guarantee that the sequences $u_{R_{n}}, \phi_{R_{n}}$ are bounded uniformly in $H^{2}_{\unif}(\R)$. Consequently, there exist $u,  \phi \in H^{2}_{\unif}(\R) \cap L^{\infty}(\R)$ such that along a subsequence $u_{R_{n}}, \phi_{R_{n}}$ converges to $u, \phi$, weakly in $H^{2}(B_{R}(0))$, strongly in $H^{1}(B_{R}(0))$ for all $R>0$ and pointwise almost everywhere. It follows from the pointwise convergence that $u \geq 0$ and
\begin{align*}
\| u \|_{L^{\infty}(\R)} &\leq C( 1 + M^{3/4} ), \\
\| \phi \|_{L^{\infty}(\R)} &\leq C( 1 + M ).
\end{align*}
Passing to the limit of the equations (\ref{eq: u phi Rn pair}) in distribution, we find that the limit $(u,\phi)$ solves
\begin{align*}
- \Delta u &+ \frac{5}{3} u^{7/3} - \phi u = 0, \\
- \Delta \phi &= 4 \pi ( m - u^{2} ).
\end{align*}
Arguing as in (\ref{eq: uRn H4 unif est})--(\ref{eq: phiRn H2 unif est}), we deduce that the desired estimates (\ref{eq: gen u H2 unif est})--(\ref{eq: gen phi H2 unif est}) hold
\begin{equation*}
\| u \|_{H^{4}_{\unif}(\R)} \leq C(1 + M^{15/4}), 
\end{equation*}
\vspace{-17pt}
\begin{equation*}
\hspace{-3pt} \| \phi \|_{H^{2}_{\unif}(\R)} \leq C(1 + M^{3/2}). \qedhere
\end{equation*}
\end{proof}

\begin{proof}[Proof of \Cref{Proposition - Uniform inf u estimate}]
As $m \in \mathcal{M}_{L^{2}}(M,\omega)$, it satisfies (H1)--(H2), 
hence by \Cref{Theorem - C/LB/L Main Result}, the solution $(u,\phi)$ of (\ref{eq: u phi eq pair}) defined in \Cref{Proposition - General Regularity Est} is unique and satisfies $\inf u > 0$.
Now suppose
\begin{align}
\inf_{m \in \mathcal{M}_{L^{2}}(M,\omega)} \inf_{x \in \R} u(x) = 0, \label{eq: inf u is zero contr ass}
\end{align}
we show that this contradicts the assumption that for all $m \in \mathcal{M}_{L^{2}}(M,\omega)$ and $R > 0$
\begin{align*}
\inf_{x \in \R} \int_{B_{R}(x)} m(z) \id z \geq \omega_{0} R^{3} - \omega_{1}.
\end{align*}
It follows from (\ref{eq: inf u is zero contr ass}) that there exists $m_{n} \in \mathcal{M}_{L^{2}}(M,\omega)$ with corresponding solution $(u_{n},\phi_{n})$ and $x_{n} \in \R$ such that for all $n \in \mathbb{N}$
\begin{align*}
u_{n}(x_{n}) \leq \frac{1}{n}.
\end{align*}
Recall the uniform estimates (\ref{eq: uRn H4 unif est})--(\ref{eq: phiRn H2 unif est}) from \Cref{Proposition - General Regularity Est}
\begin{align*}
\| u_{n} \|_{H^{4}_{\unif}(\R)} &\leq C(1 + M^{15/4}), \\
\| \phi_{n} \|_{H^{2}_{\unif}(\R)} &\leq C(1 + M^{3/2}).
\end{align*}
It follows that
\begin{align*}
\left\| \smfrac{5}{3} u_{n}^{4/3} - \phi_{n} u_{n} \right\|_{L^{2}_{\unif}(\R)} &\leq C \| u_{n} \|_{L^{\infty}(\R)}^{4/3} + \| \phi_{n} \|_{L^{2}_{\unif}(\R)} \| u_{n} \|_{L^{\infty}(\R)} \leq C(M).
\end{align*}
As $\smfrac{5}{3} u_{n}^{4/3} - \phi_{n} u_{n} \in L^{2}_{\loc}(\R)$, $u_{n} \in H^{1}_{\unif}(\R)$ and $u_{n} > 0$ solves
\begin{align*}
L_{n}u_{n} := \left( - \Delta + \frac{5}{3} u_{n}^{4/3} - \phi_{n} \right) u_{n} = 0,
\end{align*}
applying the Harnack inequality \cite{Trudinger_MeasurableCoefficients}, and
observing that the coefficients of $L_{n}$ are uniformly estimated by \Cref{Proposition - General Regularity Est}, this yields a uniform Harnack
constant, hence for all $R > 0$, there exists $C = C(R, M) > 0$ such
that
\begin{align*}
\sup_{x \in B_{R}(x_{n})} u_{n}(x) \leq C \inf_{x \in B_{R}(x_{n})} u_{n}(x) \leq \frac{C}{n}.
\end{align*}
It follows that the sequence of functions $u_{n}(\cdot + x_{n})$ converges
uniformly to zero on compact sets. Consider the ground state $(u_{n},\phi_{n})$
corresponding to the nuclear distribution $m_{n}$.

Recall that $\phi_{n}$ solves the following equation in distribution
\begin{align}
- \Delta \phi_{n} = 4 \pi \left( m_{n} - u_{n}^{2} \right). \label{eq: phi mn eq}
\end{align}
We translate the system and then pass to the limit in (\ref{eq: phi mn eq}) as $n$ tends to infinity. To do this, we use the following estimates, which are translation invariant:
\begin{align*}
\| m_{n}(\cdot + x_{n}) \|_{L^{2}_{\unif}(\R)} &\leq M, \\
\| \phi_{n}(\cdot + x_{n}) \|_{H^{2}_{\unif}(\R)} &\leq C(M).
\end{align*}
It follows that, up to a subsequence, $\phi_{n}(\cdot + x_{n})$ converges to
$\wt \phi$, weakly in $H^{2}(B_{R}(0))$, strongly in $H^{1}(B_{R}(0))$ for
all $R > 0$ and pointwise almost everywhere. Moreover, $m_{n}(\cdot + x_{n})$
converges to $\wt m$, weakly in $L^{2}(B_{R}(0))$ for all $R > 0$. By
applying the Lebesgue-Besicovitch Differentiation Theorem \cite{Evans/Gariepy}
we deduce that $\wt m \in \mathcal{M}_{L^{2}}(M,\omega)$. Passing to the
limit in
\begin{align*}
- \Delta \phi_{n}(\cdot + x_{n}) = 4 \pi \left( m_{n}(\cdot + x_{n}) - u_{n}^{2}(\cdot + x_{n}) \right),
\end{align*}
it follows that $\wt \phi$ is a distributional solution of
\begin{align}
- \Delta \wt \phi = 4\pi \wt m. \label{eq: phi bar H2 eq}
\end{align}
Arguing as in \cite[Theorem 6.10]{C/LB/L}, we show that for all $R > 0$
\begin{align}
\int_{B_{R}(0)} \wt m(z) \id z \leq C R. \label{eq: m bar BR contradiction est}
\end{align}
As $\wt m \in \mathcal{M}_{L^{2}}(M,\omega)$, this leads to the contradiction that for all $R > 0$
\begin{align*}
\omega_{0} R^{3} - \omega_{1} \leq \int_{B_{R}(0)} \wt m(z) \id z \leq C R.
\end{align*}
To show (\ref{eq: m bar BR contradiction est}) choose $\varphi \in C^{\infty}_{\textnormal{c}}(B_{2}(0))$ such that $0 \leq \varphi \leq 1$ and $\varphi = 1$ on $B_{1}(0)$. Let $R>0$, then testing (\ref{eq: phi bar H2 eq}) with $\varphi(\cdot/R)$ gives
\begin{align}
- \frac{1}{R^{2}} \int_{B_{2R}(0)} \wt \phi(z) (\Delta \varphi) (z/R) \id z = 4 \pi \int_{B_{2R}(0)} \wt m(z) \varphi(z/R) \id z. \label{eq: overline m contr est}
\end{align}
The left-hand side can be estimated by
\begin{align*}
\frac{1}{R^{2}} \bigg | \int_{B_{2R}(0)} \overline \phi(z) (\Delta \varphi) (z/R) \id z \bigg | \leq \| \overline \phi \|_{L^{\infty}(\R)} \| \Delta \varphi \|_{L^{\infty}} \frac{|B_{2R}(0)|}{R^{2}} \leq C R,
\end{align*}
where the constant $C > 0$ is independent of $R$. As $\wt m \geq 0$, from (\ref{eq: overline m contr est}) we obtain (\ref{eq: m bar BR contradiction est})
\begin{align*}
\int_{B_{R}(0)} \wt m(z) \id z \leq \int_{B_{2R}(0)} \wt m(z) \varphi(z/R) \id z \leq C R.
\end{align*}
The contradiction ensures that there exists a constant $c_{M,\omega} > 0$ such that for all $m \in \mathcal{M}_{L^{2}}(M,\omega)$, the corresponding electron density $u$ satisfies
\begin{equation*}
\inf_{x \in \R} u(x) \geq c_{M,\omega} > 0. \qedhere
\end{equation*}
\end{proof}

\begin{proof}[Proof of \Cref{Corollary - General Est Ck version}]
Our aim is to show by induction that for all $k \in \mathbb{N}_{0}$, if $m \in \mathcal{M}_{H^{k}}(M,\omega)$ then the corresponding solution $(u,\phi)$ to (\ref{eq: u phi eq pair}) satisfies
\begin{align}
\| u \|_{H^{k+4}_{\unif}(\R)} + \| \phi \|_{H^{k+2}_{\unif}(\R)} \leq C ( k, M, \omega ). \label{eq: u phi Hk induction est}
\end{align}

In \Cref{Proposition - General Regularity Est}, by combining the estimates (\ref{eq: gen u H2 unif est}) and (\ref{eq: gen phi H2 unif est}), it follows that (\ref{eq: u phi Hk induction est}) holds for the case $k = 0$: for all $m \in \mathcal{M}_{L^{2}}(M,\omega)$ the corresponding solution $(u,\phi)$ satisfies
\begin{align*}
\| u \|_{H^{4}_{\unif}(\R)} + \| \phi \|_{H^{2}_{\unif}(\R)} \leq C ( M, \omega ).
\end{align*}
We now show the induction step. Suppose the result is true for $k \in \mathbb{N}_{0}$, then consider $m \in \mathcal{M}_{H^{k+1}}(M,\omega) \subset \mathcal{M}_{H^{k}}(M,\omega)$, so by the induction hypothesis the corresponding solution $(u,\phi)$ satisfies
\begin{align}
\| u \|_{H^{k+4}_{\unif}(\R)} + \| \phi \|_{H^{k+2}_{\unif}(\R)} \leq C \left( k,  \| m \|_{H^{k}_{\unif}(\R)}, \omega \right). \label{eq: u phi Hk induction hyp est}
\end{align}
We remark that as $0 < c_{M,\omega} \leq u \leq C(M)$ and $u \in H^{k+4}_{\unif}(\R)$, it follows that for all $r \in \mathbb{R}$, $u^{r} \in H^{k+4}_{\unif}(\R)$. As $(u,\phi)$ solve (\ref{eq: u phi eq pair})
\begin{align*}
- \Delta u &= - \frac{5}{3} u^{7/3} + \phi u, \\
- \Delta \phi &= 4\pi (m - u^{2}), 
\end{align*}
by standard elliptic regularity theory \cite{Evans} for any $x \in \R$
\begin{align*}
\| \phi \|_{H^{k+3}(B_{1}(x))} &\leq C \left( \| m - u^{2} \|_{H^{k+1}(B_{2}(x))} +  \|\phi \|_{L^{2}(B_{2}(x))} \right) \\
&\leq C \left( \| m \|_{H^{k+1}(B_{2}(x))} + \| u^{2} \|_{H^{k+1}(B_{2}(x))} + \|\phi \|_{L^{2}(B_{2}(x))} \right) \\
&\leq C \left( \| m \|_{H^{k+1}_{\unif}(\R)} + \|\phi \|_{L^{2}_{\unif}(\R)} \right) + C \left( k+1, \| u \|_{H^{k+1}_{\unif}(\R)} \right) \\
&\leq C \| m \|_{H^{k+1}_{\unif}(\R)} + C \left( k+1, \| m \|_{H^{k}_{\unif}(\R)}, \omega \right) \\
&\leq C \left( k+1,  \| m \|_{H^{k+1}_{\unif}(\R)}, \omega \right),
\end{align*}
hence
\begin{align}
\| \phi \|_{H^{k+3}_{\unif}(\R)} = \sup_{x \in \R} \| \phi \|_{H^{k+3}(B_{1}(x))} &\leq C \left( k+1,  \| m \|_{H^{k+1}_{\unif}(\R)}, \omega \right). \label{eq: phi Hk3 induction est}
\end{align}
We use an identical argument together and apply the estimate (\ref{eq: phi Hk3 induction est}) to deduce
\begin{align}
\| u \|_{H^{k+5}_{\unif}(\R)} &\leq C \left( \| \smfrac{5}{3} u^{7/3} - \phi u \|_{H^{k+3}_{\unif}(\R)} +  \| u \|_{L^{2}_{\unif}(\R)} \right) \nonumber \\
&\leq C \left( \| \phi \|_{H^{k+3}_{\unif}(\R)}, \| u \|_{H^{k+3}_{\unif}(\R)} \right) \nonumber \\
&\leq C \left( k+1,  \| m \|_{H^{k+1}_{\unif}(\R)}, \omega \right). \label{eq: u Hk5 induction est}
\end{align}
Combining (\ref{eq: phi Hk3 induction est}) and (\ref{eq: u Hk5 induction est}) we obtain the desired estimate
\begin{align*}
\| u \|_{H^{k+5}_{\unif}(\R)} + \| \phi \|_{H^{k+3}_{\unif}(\R)} \leq C \left( k+1,  \| m \|_{H^{k+1}_{\unif}(\R)}, \omega \right),
\end{align*}
which completes the proof of (\ref{eq: u phi Hk induction est}) by induction.
\end{proof}

\subsection{Proofs of Pointwise Stability Estimates}
\label{Subsection - Proofs of Pointwise Stability Estimates}
To prove Theorems \ref{Theorem - One inf pointwise stability estimate alt} and \ref{Theorem - Exponential Est Integral RHS}, we adapt the proof of uniqueness of the TFW equations, shown in \cite{C/LB/L,Blanc_Uniqueness}. Due to the length of the argument, we shall prove several intermediate results. Before showing these results, we outline the structure of the proof.

First, we state two alternative sets of assumptions on nuclear
    distributions $m_1, m_2$:
\begin{itemize}
\item[(A)] 
Let $k = 0$, $m_{1} \in \mathcal{M}_{L^{2}}(M,\omega)$, and let
$(u_{1}, \phi_{1})$ denote the corresponding ground state. Also, let $m_{2} : \R \to \mathbb{R}_{\geq 0}$ satisfy $\| m_{2} \|_{L^{2}_{\unif}(\R)} \leq M'$ and suppose there exists $(u_{2},\phi_{2})$ solving \textnormal{(\ref{eq: u phi eq pair})} corresponding to $m_{2}$, satisfying $u_{2} \geq 0$ and
\begin{align}
\| u_{2} \|_{H^{4}_{\unif}(\R)} &+ \| \phi_{2} \|_{H^{2}_{\unif}(\R)} \leq C(M'). \label{eq: u2 phi2 reg A* est}
\end{align} 
In addition, we assume that either $m_{2} \not\equiv 0$ and $u_{2} > 0$, or $m_{2} = u_{2} = \phi_{2} = 0$. 

\item[(B)] Let $k \in \mathbb{N}_{0}$,
$m_{1}, m_{2} \in \mathcal{M}_{H^{k}}(M, \omega)$ and let
$(u_{1},\phi_{1}), (u_{2},\phi_{2})$ denote the corresponding ground states.
(Note that (B) implies (A), with $M' = C(M)$.)
\end{itemize}
\begin{remark}
We point out that in (A) we assume $u_2 > 0$, while in \Cref{Theorem - One inf pointwise stability estimate alt} we only require $u_2 \geq 0$. The
restriction $u_2 > 0$ allows us to directly use results from
\cite{C/LB/L}, in particular \Cref{Lemma - Eigenvalue Argument}, and will be lifted via a thermodynamic limit argument in the third part of its proof on page \pageref{proof-case2}.
\end{remark} 

Throughout the remainder of the paper we use the notation
\begin{align*}
w = u_{1} - u_{2}, \quad \psi = \phi_{1} - \phi_{2}, \quad R_{m} = 4 \pi (m_{1} - m_{2}).
\end{align*} 

By treating the coupled system of equations as a linear system and by exploiting the coupling between the electron density and electrostatic potential arising from the Coulomb energy term of the TFW functional, we obtain the following initial estimates

\begin{lemma}
\label{Lemma - Exp est lemma 1}
Suppose \textnormal{(A)} holds, then there exists $C = C(M,M',\omega) > 0$ such
that for any $\xi \in H^{1}(\R)$
\begin{align}
\int_{\R} \left( w^{2} + |\nabla w|^{2} + |\nabla \psi|^{2} \right) \xi^{2} \leq C \left( \int_{\R} R_{m} \psi \xi^{2} + \int_{\R} ( w^{2} + \psi^{2} ) |\nabla \xi|^{2} \right). \label{eq: main est 1}
\end{align}
\end{lemma}

To control the $\psi$-dependence on the right-hand side of \eqref{eq: main est
  1}, we require an estimate of the form
\begin{align}
\int_{\R} \psi^{2} \xi^{2} \leq C \left( \int_{\R} R_{m} \psi \xi^{2} + \int_{\R} ( w^{2} + \psi^{2} ) |\nabla \xi|^{2} \right), \label{eq: psi integral est}
\end{align}
which holds for $\xi \in H_{1}$, i.e. $\xi \in H^{1}(\R)$ satisfying $|\nabla \xi| \leq \xi$ on $\R$.

Suppose (\ref{eq: psi integral est}) holds, then applying \Holder's inequality and (\ref{eq: main est 1}) yields
\begin{align*}
\int_{\R} ( w^{2} + \psi^{2} ) \xi^{2} \leq C' \left( \int_{\R} R_{m}^{2} \xi^{2} + \int_{\R} ( w^{2} + \psi^{2} ) |\nabla \xi|^{2} \right).
\end{align*}
To remove the term $\int (w^{2}+\psi^{2})|\nabla \xi|^{2}$ on the right-hand
side, we simply restrict from $\xi \in H_1$ to a narrower class of test
functions,
\begin{align*}
H_{\gamma} = \{ \, \xi \in H^{1}(\R) \, | \, |\nabla \xi(x) | \leq \gamma |\xi(x)| \,\, \forall \, x \in \R \, \},
\end{align*}
where $\gamma = \min \{1, (2C')^{-1/2} \} > 0$, to show
\begin{align}
\int_{\R} ( w^{2} + |\nabla w|^{2} + \psi^{2} + |\nabla \psi|^{2} ) \xi^{2} \leq 2C' \int_{\R} R_{m}^{2} \xi^{2}.
\end{align}
In order to show (\ref{eq: psi integral est}), we adapt the argument used in
\cite{Blanc_Uniqueness} . At the same time, since the equations for $(w,\psi)$ hold pointwise, we obtain additional estimates for $\Delta w, \Delta \psi$.
\begin{lemma}
\label{Lemma - Exp est lemma 2}
Suppose \textnormal{(A)} holds, then there exists $C = C(M,M',\omega), \gamma = \gamma(M,M',\omega) > 0$ such
that for any $\xi \in H_{\gamma}$
\begin{align}
\int_{\R} \bigg( w^{2} + |\nabla w|^{2} + |\Delta w|^{2} &+ \psi^{2} + |\nabla \psi|^{2} + |\Delta \psi|^{2} \bigg) \xi^{2} \leq C \int_{\R} R_{m}^{2} \xi^{2}. \label{eq: w psi H2 xi alt est}
\end{align} 
\end{lemma} 

Clearly Lemmas \ref{Lemma - Exp est lemma 1} and \ref{Lemma - Exp est lemma 2}
hold also under the assumption (B) since (B) implies (A), with $M' = C(M)$.  In the case (B) where
$m_{1}, m_{2}$ are both uniformly bounded below and have higher regularity,
arguing as in \Cref{Corollary - General Est Ck version} and \Cref{Lemma - Exp est lemma 2}, we obtain improved estimates for $w$ and $\psi$. 

Observe that in Case (B), $M' = C(M)$. Due to this, we omit the dependence of
$M'$ in the constants that appear in the following lemmas, whenever we assume
(B) holds.
%

\begin{lemma}
\label{Lemma - Exp Est Final Est} 
Suppose that either \textnormal{(A)} or \textnormal{(B)} holds, then there exist $C = C_{A}(M,M',\omega), \gamma = \gamma_{A}( M,M',\omega) > 0$ or $C = C_{B}(k, M,\omega),$ $\gamma = \gamma_{B}( M,\omega) > 0$, where $\gamma_{B}$ independent of $k$, such that for any $\xi \in H_{\gamma}$
\begin{align}
\int_{\R} \bigg( \sum_{|\alpha_{1}| \leq k+4} |\partial^{\alpha_{1}} w|^{2} + \sum_{|\alpha_{2}| \leq k+2} |\partial^{\alpha_{2}} \psi|^{2} \bigg) \xi^{2} \leq C \int_{\R} \sum_{|\beta| \leq k} |\partial^{\beta} R_{m}|^{2} \xi^{2}. \label{eq: w and psi partial xi global final lemma est}
\end{align}
In particular, for any $y \in \R$,
\begin{align}
\sum_{|\alpha_{1}| \leq k+2} |\partial^{\alpha_{1}} w(y)|^{2} + \sum_{|\alpha_{2}| \leq k} |\partial^{\alpha_{2}} \psi(y)|^{2} \leq C \int_{\R} \sum_{|\beta| \leq k} |\partial^{\beta} R_{m}(x)|^{2} e^{-2\gamma |x - y|} \id x. \label{eq: w and psi pointwise rhs exp integral final lemma est}
\end{align}

\end{lemma}
We remark that in the following proofs, all integrals are taken over $\R$, unless stated otherwise.
\begin{proof}[Proof of \Cref{Lemma - Exp est lemma 1}]
\emph{Case 1.} First suppose that $m_{2} \not\equiv 0$ and $u_{2} > 0$. Recall that $m_{1} \in \mathcal{M}_{L^{2}}(M,\omega)$, hence by Propositions \ref{Proposition - General Regularity Est}, \ref{Proposition - Uniform inf u estimate} and (\ref{eq: u2 phi2 reg A* est})
\begin{align*}
\| u_{1} \|_{H^{4}_{\unif}(\R)} + \| \phi_{1} \|_{H^{2}_{\unif}(\R)} &\leq C(M), \\
\| u_{2} \|_{H^{4}_{\unif}(\R)} + \| \phi_{2} \|_{H^{2}_{\unif}(\R)} &\leq C(M'), \\
\inf_{x \in \R} u_{1}(x) \geq c_{M,\omega} &> 0.
\end{align*}
By the Sobolev embedding: for all $k \in \mathbb{N}_{0}$ and $x \in \R$ $H^{k+2}(B_{1}(x)) \hookrightarrow C^{k,1/2}(B_{1}(x))$, so it follows that
\begin{align*}
\| u_{1} \|_{W^{2,\infty}(\R)} + \| \phi_{1} \|_{L^{\infty}(\R)} &\leq C(M), \\
\| u_{2} \|_{W^{2,\infty}(\R)} + \| \phi_{2} \|_{L^{\infty}(\R)} &\leq C(M'),
\end{align*}
hence $w = u_{1} - u_{2} \in H^{4}_{\unif}(\R) \cap W^{2,\infty}(\R)$, $\psi = \phi_{1} - \phi_{2} \in H^{2}_{\unif}(\R) \cap L^{\infty}(\R)$, and solve
\begin{subequations}
\label{eq: w psi pair eq}
\begin{align}
- \Delta w &= \frac{5}{3} \left( {u_{2}}^{7/3} - u_{1}^{7/3} \right) + \phi_{1} u_{1} - \phi_{2} u_{2}, \label{eq: w initial eq} \\
- \Delta \psi &= 4 \pi \left( u_{2}^{2} - u_{1}^{2} \right) + R_{m}, \label{eq: psi initial eq}
\end{align}
\end{subequations}
pointwise. Let $\xi \in H^{1}(\R)$ then test (\ref{eq: w initial eq}) with $w \xi^{2}$ to obtain
\begin{align}
\int \nabla w \cdot \nabla (w\xi^{2}) + {\frac{5}{3}} \int (u_{1}^{7/3} - u_{2}^{7/3}) w \xi^{2} - \int ( \phi_{1} u_{1} - \phi_{2} u_{2} ) w \xi^{2} &= 0. \label{eq: w with w xi2 est}
\end{align}
We will use the following rearrangements
\begin{align}
\phi_{1} u_{1} - \phi_{2} u_{2} &= \frac{\phi_{1} + \phi_{2}}{2} w + \frac{u_{1} + u_{2}}{2} \psi, \label{eq: phi u rearrange} \\ 
\int \nabla w \cdot \nabla (w \xi^{2}) &= \int | \nabla (w\xi) |^{2} - \int w^{2} |\nabla \xi|^{2}, \label{eq: grad w alt} \\
\int \nabla \psi \cdot \nabla (\psi \xi^{2}) &= \int | \nabla (\psi \xi) |^{2} - \int \psi^{2} |\nabla \xi|^{2}. \label{eq: grad psi alt}
\end{align}
To estimate the second term of (\ref{eq: w with w xi2 est}), observe that \Cref{Proposition - Uniform inf u estimate} and (A) imply that
$\inf u_{1} \geq c_{M,\omega} > 0$ and recall the assumption $u_{2} > 0$. It follows that for
$\nu = \frac{1}{2} \inf (u_{1}^{4/3} + u_{2}^{4/3}) \geq \frac{1}{2}
c_{M,\omega}^{4/3} > 0$
\begin{align}
(u_{1}^{7/3} - u_{2}^{7/3}) (u_{1} - u_{2}) &= (u_{1}^{4/3} + u_{2}^{4/3}) w^{2} + u_{1} u_{2} ( u_{1}^{1/3} - u_{2}^{1/3} ) w \nonumber \\
&\geq (u_{1}^{4/3} + u_{2}^{4/3}) w^{2}\nonumber \\
&\geq \frac{1}{2} ( u_{1}^{4/3} + u_{2}^{4/3} ) w^{2} + \nu w^{2}. \label{eq: nu est}
\end{align}

Combining the estimates (\ref{eq: w with w xi2 est})--(\ref{eq: grad w alt}) and (\ref{eq: nu est}), we obtain
\begin{align}
\int |\nabla (w\xi)|^{2} &+ {\frac{5}{6}} \int ( u_{1}^{4/3} + u_{2}^{4/3}) w^{2} \xi^{2} - {\frac{1}{2}} \int (\phi_{1} + \phi_{2}) w^{2} \xi^{2} + \nu \int w^{2} \xi^{2} \nonumber \\ 
&\leq \int w^{2} |\nabla \xi|^{2} + {\frac{1}{2}} \int \psi ( u_{1}^{2} - u_{2}^{2} ) \xi^{2}. \label{eq: before l}
\end{align}
We define the following operators
\begin{align*}
L_{1} &= - \Delta + \frac{5}{3} u_{1}^{4/3} - \phi_{1}, \\
L_{2} &= - \Delta + \frac{5}{3} u_{2}^{4/3} - \phi_{2}, \\
L_{\phantom{1}} &= \frac{1}{2} L_{1} + \frac{1}{2} L_{2} = - \Delta + \frac{5}{6} (u_{1}^{4/3} + u_{2}^{4/3}) - {\frac{1}{2}} (\phi_{1} + \phi_{2}).
\end{align*}
As $u_{1}, u_{2} > 0$, \Cref{Lemma - Eigenvalue Argument} implies that $L_{1}, L_{2}$ are non-negative operators, hence for any $\varphi \in H^{1}(\R)$
\begin{align}
\langle \varphi , L \varphi \rangle &= {\frac{1}{2}} \langle \varphi , L_{1}\varphi \rangle + {\frac{1}{2}} \langle \varphi , L_{2}\varphi \rangle \geq 0. \label{eq: L varphi geq 0}
\end{align}
Observe that as $w \in W^{2,\infty}(\R)$ and $\xi \in H^{1}(\R)$, $w\xi \in H^{1}(\R)$. We can express (\ref{eq: before l}) as
\begin{align}
\langle w \xi, L(w \xi) \rangle + \nu \int w^{2} \xi^{2} &\leq \int w^{2} |\nabla \xi|^{2} + {\frac{1}{2}} \int \psi ( u_{1}^{2} - u_{2}^{2} ) \xi^{2}. \label{eq: est 1}
\end{align}
To control the final term of (\ref{eq: est 1}), we begin by testing (\ref{eq: psi initial eq}) with $\psi \xi^{2}$ to obtain
\begin{align}
\int \nabla \psi \cdot \nabla (\psi \xi^{2}) &= 4 \pi \int \psi ( u_{2}^{2} - u_{1}^{2} ) \xi^{2} + \int R_{m} \psi \xi^{2}. \label{eq: psi test eq}
\end{align}
Rearranging (\ref{eq: psi test eq}) and applying (\ref{eq: grad psi alt}) yields
\begin{align}
{\frac{1}{2}} \int \psi ( u_{1}^{2} - u_{2}^{2}) \xi^{2} &= {\frac{1}{8 \pi}} \int R_{m} \psi \xi^{2} - {\frac{1}{8 \pi}} \int \nabla \psi \cdot \nabla (\psi \xi^{2}) \nonumber \\
&= {\frac{1}{8 \pi}} \int R_{m} \psi \xi^{2} - {\frac{1}{8 \pi}} \int | \nabla (\psi \xi) |^{2} + {\frac{1}{8 \pi}} \int \psi^{2} |\nabla \xi|^{2}. \label{eq: est 2 start}
\end{align}
Combining (\ref{eq: est 1}) and (\ref{eq: est 2 start}) yields
\begin{align}
\langle w \xi, L(w \xi) \rangle &+ \nu \int w^{2} \xi^{2} + {\frac{1}{8 \pi}} \int | \nabla  (\psi \xi) |^{2} 
 \nonumber \\ 
 &\leq {\frac{1}{8 \pi}} \int R_{m}\psi \xi^{2} 
 + \int w^{2} |\nabla \xi|^{2} + {\frac{1}{8 \pi}} \int \psi^{2} |\nabla \xi|^{2}. \label{eq: est 2 main}
\end{align}
As $\xi \nabla \psi = \nabla (\psi \xi) - \psi \nabla \xi$, we have
\begin{align}
\int | \nabla \psi |^{2} \xi^{2} &\leq C \left( \int | \nabla (\psi \xi)|^{2} + \int \psi^{2} |\nabla \xi|^{2} \right) \nonumber \\
&\leq C \left( \int R_{m} \psi \xi^{2}
 + \int ( w^{2} + \psi^{2} ) |\nabla \xi|^{2} \right). \label{eq: nabla psi rearrangement est}
\end{align}

Combining the estimates (\ref{eq: est 2 main})--(\ref{eq: nabla psi rearrangement est}), we obtain
\begin{align}
\langle w \xi, L(w \xi) \rangle &+ \nu \int w^{2} \xi^{2} + \int | \nabla \psi|^{2} \xi^{2} 
 \nonumber \\ 
 &\leq C \left( \int R_{m} \psi \xi^{2} + \int ( w^{2} + \psi^{2} ) |\nabla \xi|^{2} \right). \label{eq: L w grad psi est} 
\end{align}
Next we obtain an estimate for $\int |\nabla w|^{2} \xi^{2}$, using the fact that $L$ is a non-negative operator. We can express $L$ as
\begin{align*}
L = -\Delta + a, \quad \text{ where } \,\, a = \frac{5(u_{1}^{4/3} + u_{2}^{4/3})}{6} - \frac{\phi_{1} + \phi_{2}}{2} \in H^{2}_{\unif}(\mathbb{R}^{3}).
\end{align*}
From (\ref{eq: L varphi geq 0}), we have shown that $L = - \Delta + a \geq 0$ in the sense that $\langle \varphi , L \varphi \rangle \geq 0$ for every $\varphi \in H^{1}(\mathbb{R}^{3})$. So for $\varepsilon \in (0,1)$
\begin{align*}
L &= (1 - \varepsilon) (-\Delta + a) + \varepsilon(-\Delta) + \varepsilon a \geq \varepsilon(-\Delta)  - \varepsilon \| a \|_{L^{\infty}(\mathbb{R}^{3})}.
\end{align*}
Applying this to (\ref{eq: L w grad psi est}) gives
\begin{align*}
\varepsilon \int &| \nabla (w \xi)|^{2} + (\nu - \varepsilon \| a \|_{L^{\infty}(\mathbb{R}^{3})} ) \int w^{2} \xi^{2}
 \nonumber \\ 
 &\leq C \left( \int R_{m} \psi \xi^{2} + \int ( w^{2} + \psi^{2} ) |\nabla \xi|^{2} \right),
\end{align*}
so choosing $\varepsilon = \min \{ \frac{\nu}{2\| a \|_{L^{\infty}}} , \frac{1}{2} \}$, we deduce
\begin{align*}
\int | \nabla (w \xi)|^{2} &\leq C \left( \int R_{m} \psi \xi^{2} + \int ( w^{2} + \psi^{2} ) |\nabla \xi|^{2} \right)
\end{align*}
and since $\xi \nabla w = \nabla(w \xi) - w \nabla \xi$, we deduce
\begin{align}
\int | \nabla w| ^{2} \xi^{2} &\leq C \left( \int R_{m} \psi \xi^{2} + \int ( w^{2} + \psi^{2} ) |\nabla \xi|^{2} \right). \label{eq: grad w est}
\end{align}
We combine the estimates (\ref{eq: L w grad psi est}) and (\ref{eq: grad w est}) to obtain the desired estimate (\ref{eq: main est 1})
\begin{align*}
\int &w^{2} \xi^{2} + \int | \nabla w| ^{2} \xi^{2} + \int | \nabla \psi|^{2} \xi^{2} \leq C \left( \int R_{m} \psi \xi^{2} + \int ( w^{2} + \psi^{2} ) |\nabla \xi|^{2} \right)
\end{align*}
and observe that this estimate is valid for any $\xi \in H^{1}(\R)$.

\emph{Case 2.} Suppose now that $m_{2} = u_{2} = \phi_{2} = 0$, then the argument used to show (\ref{eq: before l}) holds to give
\begin{align}
\int |\nabla (w\xi)|^{2} &+ {\frac{5}{6}} \int u_{1}^{4/3} w^{2} \xi^{2} - {\frac{1}{2}} \int \phi_{1} w^{2} \xi^{2} + \nu \int w^{2} \xi^{2} \nonumber \\ 
&\leq \int w^{2} |\nabla \xi|^{2} + {\frac{1}{2}} \int \psi  u_{1}^{2} \xi^{2}.
\end{align}
Now using that $L_{1}$ is a non-negative operator, we obtain
\begin{align*}
\frac{1}{2}\int |\nabla (w\xi)|^{2} + \nu \int w^{2} \xi^{2} &\leq
{\frac{1}{2}} \langle \varphi , L_{1}\varphi \rangle + \frac{1}{2}\int |\nabla (w\xi)|^{2} + \nu \int w^{2} \xi^{2} \\
&= \int |\nabla (w\xi)|^{2} + {\frac{5}{6}} \int u_{1}^{4/3} w^{2} \xi^{2} - {\frac{1}{2}} \int \phi_{1} w^{2} \xi^{2} + \nu \int w^{2} \xi^{2} \nonumber \\ 
&\leq \int w^{2} |\nabla \xi|^{2} + {\frac{1}{2}} \int \psi  u_{1}^{2} \xi^{2}.
\end{align*}
Then applying the estimates (\ref{eq: psi test eq})--(\ref{eq: L w grad psi est}) yields the desired estimate (\ref{eq: main est 1}): for all $\xi \in H^{1}(\R)$
\begin{align*}
\int &w^{2} \xi^{2} + \int | \nabla w| ^{2} \xi^{2} + \int | \nabla \psi|^{2} \xi^{2} \leq C \left( \int R_{m} \psi \xi^{2} + \int ( w^{2} + \psi^{2} ) |\nabla \xi|^{2} \right) \qedhere
\end{align*}
\end{proof}

\begin{proof}[Proof of \Cref{Lemma - Exp est lemma 2}]
To obtain an integral estimate for $\psi$, first recall (\ref{eq: w initial eq}), that $w$ solves
\begin{align*}
- \Delta w + \frac{5}{3} \left( \, u_{1}^{7/3} - u_{2}^{7/3} \, \right) - \frac{\phi_{1} + \phi_{2}}{2} w &= \frac{u_{1} + u_{2}}{2} \psi,
\end{align*}
then testing this equation with $\psi \xi^{2}$, for $\xi \in H^{1}(\R)$, yields
\begin{align}
\int \frac{u_{1} + u_{2}}{2} \psi^{2} \xi^{2} = - \int \Delta w \psi \xi^{2} + \frac{5}{3} \int \left( \, u_{1}^{7/3} - u_{2}^{7/3} \, \right) \psi \xi^{2} - \int \frac{\phi_{1} + \phi_{2}}{2} w \psi \xi^{2}. \label{eq: alt psi est w eq}
\end{align}
The first term of the right-hand side can be estimated using integration by parts
\begin{align*}
\bigg| \int& \Delta w \psi \xi^{2} \bigg| = \bigg| \int \nabla w \cdot \nabla \left( \psi \xi^{2} \right) \bigg| \leq \bigg| \int \nabla w \cdot \nabla \psi \xi^{2} \bigg| + 2 \bigg| \int \nabla w \cdot \nabla \xi \psi \xi \bigg| \\
&\leq \left( \int |\nabla w|^{2} \xi^{2} \right)^{1/2} \left( \int |\nabla \psi|^{2} \xi^{2} \right)^{1/2} + 2 \left( \int |\nabla w|^{2} |\nabla \xi|^{2} \right)^{1/2} \left( \int \psi^{2} \xi^{2} \right)^{1/2}.
\end{align*}
By restricting $\xi \in H_{1}$, we have $|\nabla \xi| \leq |\xi|$ hence 
\begin{align}
\bigg| \int& \Delta w \psi \xi^{2} \bigg| \leq 2 \left( \int |\nabla w|^{2} \xi^{2} \right)^{1/2} \left( \int \psi^{2} \xi^{2} \right)^{1/2} + \int \left( |\nabla w|^{2} + |\nabla \psi|^{2}  \right) \xi^{2}. \label{eq: alt psi est Delta term}
\end{align}
We now estimate the remaining terms on the right-hand side of (\ref{eq: alt psi est w eq})
\begin{align}
\bigg | \frac{5}{3} \int \left( \, u_{1}^{7/3} - u_{2}^{7/3} \, \right) &\psi \xi^{2} - \int \frac{\phi_{1} + \phi_{2}}{2} w \psi \xi^{2} \bigg | \nonumber \\ &\leq C \int |w||\psi| \xi^{2} \leq C \left( \int w^{2} \xi^{2} \right)^{1/2} \left( \int \psi^{2} \xi^{2} \right)^{1/2}. \label{eq: alt psi est final terms}
\end{align}
Combining the estimates (\ref{eq: alt psi est Delta term})--(\ref{eq: alt psi
  est final terms}) with (\ref{eq: alt psi est w eq}) and using that
$\inf u_{1} \geq c_{M,\omega} > 0$ and $u_{2} \geq 0$, we obtain
\begin{align}
\int \psi^{2} \xi^{2} \leq \frac{2}{c_{M,\omega}} \int \frac{u_{1} + u_{2}}{2} \psi^{2} \xi^{2} &\leq C \left[ \left( \int |\nabla w|^{2} \xi^{2} \right)^{1/2} + \left( \int w^{2} \xi^{2} \right)^{1/2} \right] \left( \int \psi^{2} \xi^{2} \right)^{1/2} \nonumber \\ 
& \qquad + \int \left( |\nabla w|^{2} + |\nabla \psi|^{2}  \right) \xi^{2} \label{eq: psi est lower bound}
\end{align}
Applying Young's inequality twice and using (\ref{eq: main est 1}) of \Cref{Lemma - Exp est lemma 1} yields
\begin{align*}
\int \psi^{2} \xi^{2} &\leq \frac{1}{2} \int \psi^{2} \xi^{2} + C \int \left( w^{2} + |\nabla w|^{2} + |\nabla \psi|^{2}  \right) \xi^{2} \\ &\leq \frac{1}{2} \int \psi^{2} \xi^{2} + C \left( \int R_{m} \psi \xi^{2} + \int \left( w^{2} + \psi^{2} \right) |\nabla \xi|^{2} \right) \\
& \leq \frac{3}{4} \int \psi^{2} \xi^{2} + C \left( \int R_{m}^{2} \xi^{2} + \int \left( w^{2} + \psi^{2} \right) |\nabla \xi|^{2} \right),
\end{align*}
hence we obtain
\begin{align}
\int \left( w^{2} + |\nabla w|^{2} + \psi^{2} + |\nabla \psi|^{2} \right) \xi^{2} \leq C \left( \int R_{m}^{2} \xi^{2} + \int \left( w^{2} + \psi^{2} \right) |\nabla \xi|^{2} \right). \label{eq: alt w psi H1 xi est}
\end{align}
We further restrict the choice of the test function $\xi$, to remove the terms depending on $w$ and $\psi$ from the right-hand side. Given $C = C(M',M,\omega) > 0$, define $\gamma = \min\{ 1, (2C)^{-1/2} \} > 0$. First note that $H_{\gamma} \subseteq H_{1}$, so for any $\xi \in H_{\gamma}$ the estimate (\ref{eq: alt w psi H1 xi est}) continues to hold. In addition, $|\nabla \xi| \leq \gamma |\xi|$, hence
\begin{align*}
&\int \left( w^{2} + |\nabla w|^{2} + \psi^{2} + |\nabla \psi|^{2} \right) \xi^{2} \nonumber \leq C \left( \int R_{m}^{2} \xi^{2} + \int ( w^{2} + \psi^{2} ) |\nabla \xi|^{2} \right) \\
&\leq C \left( \int R_{m}^{2} \xi^{2} + \wt \gamma^{2} \int ( w^{2} + \psi^{2} ) \xi^{2} \right) \leq C \int R_{m}^{2} \xi^{2} + \frac{1}{2} \int ( w^{2} + \psi^{2} ) \xi^{2}.
\end{align*}
After re-arranging, it follows that for any $\xi \in H_{\gamma}$
\begin{align}
\int \left( w^{2} + |\nabla w|^{2} + \psi^{2} + |\nabla \psi|^{2} \right) \xi^{2} \leq C \int  R_{m}^{2} \xi^{2}. \label{eq: not quite pointwise est}
\end{align} 
Finally, as the equations (\ref{eq: w psi pair eq}) hold pointwise, squaring each equation and integrating them against $\xi^{2}$ yields
\begin{align*}
\int |\Delta w|^{2}\xi^{2} &\leq C \int \left( w^{2} + \psi^{2} \right) \xi^{2} \\
\int |\Delta \psi|^{2}\xi^{2} &\leq C \int \left( R_{m}^{2} + w^{2} \right) \xi^{2}.
\end{align*}
Combining these estimates with (\ref{eq: alt w psi H1 xi est}), we obtain the desired result (\ref{eq: w psi H2 xi alt est})
\begin{equation*}
\int \big( w^{2} + |\nabla w|^{2} + |\Delta w|^{2} + \psi^{2} + |\nabla \psi|^{2} + |\Delta \psi|^{2} \big) \xi^{2} \leq C \int R_{m}^{2} \xi^{2}. \qedhere
\end{equation*}
\end{proof}

\begin{proof}[Proof of \Cref{Lemma - Exp Est Final Est}]
\emph{Case 1.} Suppose (B) holds, so $m_{i} \in \mathcal{M}_{H^{k}}(M, \omega)$ for some $k \in \mathbb{N}_{0}$. 
By \Cref{Corollary - General Est Ck version}, for $i \in \{1,2\}$
\begin{align}
\| u_{i} \|_{H^{k+4}_{\unif}(\R)} +
\| \phi_{i} \|_{H^{k+2}_{\unif}(\R)} \leq C(k,M,\omega). \label{eq: corr 4 u phi extra reg est}
\end{align}
Using integration by parts, we shall obtain integral estimates for derivatives of $w$ in terms of derivatives of $\Delta w$. We will use the Einstein summation convention throughout this proof. 

To begin, we approximate $w \in H^{k+4}_{\unif}(\R)$ by smooth functions $w_{h} \in C^{\infty}(\R)$ such that for all $|\beta| \leq k+4$, $\partial^{\beta} w_{h}$ converges to $\partial^{\beta} w$ pointwise \cite{Jost_PDE}. This approximation is necessary in order to obtain estimates for $\partial^{\alpha} w$ when $|\alpha| = k+4$.

Fix $\xi \in H_{\gamma}$ and let $|\beta| = k' \leq k+2$. Then using integration by parts gives
\begin{align*}
\int& |\Delta \partial^{\beta} w_{h}|^{2} \xi^{2} = \int \partial_{ii} \partial^{\beta} w_{h} \partial_{jj} \partial^{\beta} w_{h} \xi^{2} \nonumber \\
&= - \int \partial_{i} \partial^{\beta} w_{h} \partial_{ijj} \partial^{\beta} w_{h} \xi^{2} - 2 \int  \partial_{i} \partial^{\beta} w_{h} \partial_{jj} \partial^{\beta} w_{h} \xi \partial_{i} \xi \nonumber \\
&= \int \partial_{ij} \partial^{\beta} w_{h} \partial_{ij} \partial^{\beta} w_{h} \xi^{2} + 2 \int \partial_{i} \partial^{\beta} w_{h} \partial_{ij} \partial^{\beta} w_{h} \xi \partial_{j} \xi - 2 \int \partial_{i} \partial^{\beta} w_{h} \partial_{jj} w_{h} \xi \partial_{i} \xi \nonumber \\
&= \int \sum_{|\alpha| = 2}| \partial^{\alpha + \beta} w|^{2} \xi^{2} + 2 \int \partial_{i} \partial^{\beta} w_{h} \partial_{ij} \partial^{\beta} w_{h} \xi \partial_{j} \xi - 2 \int \partial_{i} \partial^{\beta} w_{h} \partial_{jj} \partial^{\beta} w_{h} \xi \partial_{i} \xi.
\end{align*}
Summing over $|\beta| = k'$ and rearranging yields
\begin{align*}
\int& \sum_{|\alpha| = k' + 2}| \partial^{\alpha} w_{h}|^{2} \xi^{2} = \int \sum_{|\beta| = k'} |\Delta \partial^{\beta} w_{h}|^{2} \xi^{2} \\ &+ 2 \sum_{|\beta| = k'} \sum_{i,j=1}^{3} \bigg( \int \partial_{i} \partial^{\beta} w_{h} \partial_{ij} \partial^{\beta} w_{h} \xi \partial_{j} \xi - \int \partial_{i} \partial^{\beta} w_{h} \partial_{jj} \partial^{\beta} w_{h} \xi \partial_{i} \xi \bigg).
\end{align*}
Then, using that $\xi \in H_{\gamma} \subseteq H_{1}$, hence $|\nabla \xi| \leq |\xi|$, we can estimate the right-hand side using \Holder's inequality,
\begin{align*}
\int& \sum_{|\alpha| = k' + 2}| \partial^{\alpha} w_{h}|^{2} \xi^{2} \leq \int \sum_{|\beta| = k'} |\Delta \partial^{\beta} w_{h}|^{2} \xi^{2} \\
&+ C \sum_{|\beta| = k'} \sum_{i,j=1}^{3} \bigg( \int |\partial_{i} \partial^{\beta} w_{h}| |\partial_{ij} \partial^{\beta} w_{h}| \xi^{2} + \int |\partial_{i} \partial^{\beta} w_{h}| |\partial_{jj} \partial^{\beta} w_{h}| \xi^{2} \bigg) \\
&\leq \frac{1}{2} \int \sum_{|\alpha| = k' + 2}| \partial^{\alpha} w_{h}|^{2} \xi^{2} + C \bigg( \int \sum_{|\beta_{1}| = k'} |\Delta \partial^{\beta_{1}} w_{h}|^{2} \xi^{2} + \int \sum_{|\beta_{2}| = k'+1} |\partial^{\beta_{2}} w_{h}|^{2} \xi^{2} \bigg).
\end{align*}
Re-arranging this and letting $h \to 0$, we obtain
\begin{align}
\sum_{|\alpha| = k' + 2} &\int | \partial^{\alpha} w|^{2} \xi^{2} \leq C \bigg( \int \sum_{|\beta_{1}| = k'} |\partial^{\beta_{1}} \Delta w|^{2} \xi^{2} + \int \sum_{|\beta_{2}| = k'+1} |\partial^{\beta_{2}} w|^{2} \xi^{2} \bigg). \label{eq: w k' partial est}
\end{align}
Using an identical argument, we obtain similar estimates for $\psi$, for $k' \leq k$,
\begin{align}
\sum_{|\alpha| = k' + 2} &\int | \partial^{\alpha} \psi|^{2} \xi^{2} \leq C \bigg( \int \sum_{|\beta_{1}| = k'} |\partial^{\beta_{1}} \Delta \psi|^{2} \xi^{2} + \int \sum_{|\beta_{2}| = k'+1} |\partial^{\beta_{2}} \psi|^{2} \xi^{2} \bigg). \label{eq: psi k' partial est}
\end{align}
In the case $k' = 0$, combining (\ref{eq: w k' partial est}), (\ref{eq: psi k' partial est}) and (\ref{eq: w psi H2 xi alt est}) of Lemma \ref{Lemma - Exp est lemma 2} yields: there exists $C, \gamma > 0$ such that for all $\xi \in H_{\gamma}$
\begin{align}
\int \sum_{|\alpha| = 2} ( |\partial^{\alpha} w|^{2} &+ |\partial^{\alpha} \psi|^{2} ) \xi^{2} \nonumber \\ &\leq C\int \left( |\nabla w|^{2} + |\Delta w|^{2} + |\nabla \psi|^{2} + |\Delta \psi|^{2} \right) \xi^{2} \leq C \int R_{m}^{2} \xi^{2}.
\end{align} 

We will now provide estimates for the right-hand terms of the form $\partial^{\beta} \Delta w, \partial^{\beta} \Delta \psi$. Recall (\ref{eq: w psi pair eq})
\begin{align*}
- \Delta w &= \frac{5}{3} \left( {u_{2}}^{7/3} - u_{1}^{7/3} \right) + \frac{\phi_{1} + \phi_{2}}{2} w + \frac{u_{1} + u_{2}}{2} \psi =: f_{1}, \\
- \Delta \psi &= 4 \pi \left( u_{2}^{2} - u_{1}^{2} \right) + R_{m} =: f_{2}.
\end{align*}

From (\ref{eq: corr 4 u phi extra reg est}) it follows that $f_{1} \in H^{k+2}_{\unif}(\R), f_{2} \in H^{k}_{\unif}(\R)$. Let $|\alpha_{1}| = j_{1} \leq k+2, |\alpha_{2}| = j_{2} \leq k$, then differentiating (\ref{eq: w psi pair eq}) yields
\begin{align}
|\partial^{\alpha_{1}} \Delta w| &\leq C(j_{1},M,\omega) \sum_{|\beta_{1}| \leq j_{1}} \left( |\partial^{\beta_{1}} w| + |\partial^{\beta_{1}} \psi| \right), \label{eq: Delta partial w pointwise est} \\
|\partial^{\alpha_{2}} \Delta \psi| &\leq C(j_{2},M,\omega) \sum_{|\beta_{2}| \leq j_{2}} \left( |\partial^{\beta_{2}} R_{m}| + |\partial^{\beta_{2}} w| \right). \label{eq: Delta partial psi pointwise est}
\end{align}
Squaring (\ref{eq: Delta partial w pointwise est})--(\ref{eq: Delta partial psi pointwise est}), summing over partial derivatives and integrating against $\xi^{2}$ we deduce
\begin{align}
\int \sum_{|\alpha_{1}| = j_{1}} |\partial^{\alpha_{1}} \Delta w|^{2} \xi^{2} &\leq
C \int \sum_{|\beta_{1}| \leq j_{1}} \left( |\partial^{\beta_{1}} w|^{2} + |\partial^{\beta_{1}} \psi|^{2} \right) \xi^{2}, \label{eq: Delta partial w j part 2 est} \\
\int \sum_{|\alpha_{2}| = j_{2}} |\partial^{\alpha_{2}} \Delta \psi|^{2} \xi^{2}
&\leq C \int \sum_{|\beta_{2}| \leq j_{2}} \left( |\partial^{\beta_{2}} R_{m}|^{2} + |\partial^{\beta_{2}} w|^{2} \right) \xi^{2}. \label{eq: Delta partial psi j part 2 est}
\end{align}
Substituting (\ref{eq: Delta partial w j part 2 est}) into (\ref{eq: w k' partial est}) gives for $i \leq k+4$
\begin{align}
\int \sum_{|\alpha| = i_{1}}| \partial^{\alpha} w|^{2} \xi^{2} \leq C \int \bigg( \sum_{|\beta_{1}| = i_{1}-1}  |\partial^{\beta_{1}} w|^{2} + \sum_{|\beta_{2}| = i_{1}-2}  |\partial^{\beta_{2}} \Delta w|^{2}  \bigg) \xi^{2}& \nonumber \\
\leq C \int \bigg( \sum_{|\beta_{1}| = i_{1}-1}  |\partial^{\beta_{1}} w|^{2} +  \sum_{|\beta_{1}| \leq i_{1}-2} \left( |\partial^{\beta_{1}} w|^{2} + |\partial^{\beta_{1}} \psi|^{2} \right) \bigg) \xi^{2}&. \label{eq: partial w i1 est}
\end{align}
Similarly, substituting (\ref{eq: Delta partial psi j part 2 est}) into (\ref{eq: psi k' partial est}) gives for $i_{2} \leq k+2$
\begin{align}
\int \sum_{|\alpha| = i_{2}}| \partial^{\alpha} \psi|^{2} \xi^{2} \leq C \int \bigg( \sum_{|\beta_{1}| = i_{2}-1}  |\partial^{\beta_{1}} \psi|^{2} + \sum_{|\beta_{2}| = i_{2}-2}  |\partial^{\beta_{2}} \Delta \psi|^{2}  \bigg)\xi^{2}& \nonumber \\
\leq C \int \bigg( \sum_{|\beta_{1}| = i_{2}-1}  |\partial^{\beta_{1}} \psi|^{2} + \sum_{|\beta_{2}| \leq i_{2}-2} \left( |\partial^{\beta_{2}} R_{m}|^{2} + |\partial^{\beta_{2}} w|^{2} \right) \bigg) \xi^{2}&. \label{eq: partial psi i2 est}
\end{align}
Using (\ref{eq: partial w i1 est}) and (\ref{eq: partial psi i2 est}), arguing by induction over $i_{1}, i_{2}$ simultaneously gives
\begin{align*}
\int \sum_{|\alpha| \leq k+2} \left( |\partial^{\alpha} w|^{2} + | \partial^{\alpha} \psi|^{2} \right) \xi^{2} \leq C \int \sum_{|\beta| \leq k} |\partial^{\beta} R_{m}|^{2} \xi^{2}.
\end{align*} 
To show the remaining estimate for the derivatives of $w$, applying (\ref{eq: partial w i1 est}) with $i_{1} = k+3, k+4$ yields the estimate (\ref{eq: w and psi partial xi global final lemma est})
\begin{equation*}
\int \bigg( \sum_{|\alpha_{1}| \leq k+4} |\partial^{\alpha_{1}} w|^{2} + \sum_{|\alpha_{2}| \leq k+2} | \partial^{\alpha_{2}} \psi|^{2} \bigg) \xi^{2} \leq C \int \sum_{|\beta| \leq k} |\partial^{\beta} R_{m}|^{2} \xi^{2}.
\end{equation*}
Now fix $y \in \R$ and choose $\xi(x) = e^{-\gamma |x - y|}$. We will now show the lower pointwise lower bound for $w$ and $\psi$
\begin{align}
\sum_{|\alpha_{1}| \leq k+2} &|\partial^{\alpha_{1}} w(y)|^{2} + \sum_{|\alpha_{2}| \leq k} |\partial^{\alpha_{2}} \psi(y)|^{2} \nonumber \\
&\leq C \int \bigg( \sum_{|\alpha_{1}| \leq k+4} |\partial^{\alpha_{1}} w|^{2} + \sum_{|\alpha_{2}| \leq k+2} |\partial^{\alpha_{2}} \psi|^{2} \bigg) e^{-2 \gamma |x - y|} \id x, \label{eq: w psi lhs pointwise est}
\end{align}
where the constant $C$ is independent of $y$ and $\gamma$. 

By \Cref{Corollary - General Est Ck version},
$w \in H^{k+4}(B_{1}(y)), \psi \in H^{k+2}(B_{1}(y))$, hence by the Sobolev
embedding theorem \cite{Evans}
$w \in C^{k+2,1/2}(B_{1}(y)), \psi \in C^{k,1/2}(B_{1}(y))$ and
\begin{align*}
\| w \|_{C^{k+2}(B_{1}(y))} \leq C \| w \|_{H^{k+4}(B_{1}(y))}, \\
\| \psi \|_{C^{k}(B_{1}(y))} \leq C \| \psi \|_{H^{k+2}(B_{1}(y))}.
\end{align*}
We use these estimates to show (\ref{eq: w psi lhs pointwise est})
\begin{align*}
\sum_{|\alpha_{1}| \leq k+2} |\partial^{\alpha_{1}} w(y)|^{2} &+ \sum_{|\alpha_{2}| \leq k} |\partial^{\alpha_{2}} \psi(y)|^{2} \\
&\leq \| w \|_{C^{k+2,1/2}(B_{1}(y))}^{2} + \| \psi \|_{C^{k,1/2}(B_{1}(y))}^{2} \\
&\leq C \left( \| w \|_{H^{k+4}(B_{1}(y))}^{2} + \| \psi \|_{H^{k+2}(B_{1}(y))}^{2} \right) \\
&= C \int_{B_{1}(y)} \bigg( \sum_{|\alpha_{1}| \leq k+4} |\partial^{\alpha_{1}} w|^{2} + \sum_{|\alpha_{2}| \leq k+2} |\partial^{\alpha_{2}} \psi|^{2} \bigg) \\
&\leq C \int_{\R} \bigg( \sum_{|\alpha_{1}| \leq k+4} |\partial^{\alpha_{1}} w|^{2} + \sum_{|\alpha_{2}| \leq k+2} |\partial^{\alpha_{2}} \psi|^{2} \bigg) e^{-2 \gamma |x - y|} \id x.
\end{align*}
Combining (\ref{eq: w and psi partial xi global final lemma est}) and (\ref{eq: w psi lhs pointwise est}), we obtain the desired estimate (\ref{eq: w and psi pointwise rhs exp integral final lemma est})
\begin{equation*}
\sum_{|\alpha_{1}| \leq k+2} |\partial^{\alpha_{1}} w(y)|^{2} + \sum_{|\alpha_{2}| \leq k} |\partial^{\alpha_{2}} \psi(y)|^{2} \leq C \int \sum_{|\beta| \leq k} |\partial^{\beta} R_{m}(x)|^{2} e^{-2 \gamma |x-y|} \id x.
\end{equation*}

\emph{Case 2.} Suppose (A) holds, then as
$m_{1} \in \mathcal{M}_{L^{2}}(M,\omega)$, by Proposition \ref{Proposition -
  General Regularity Est} and (\ref{eq: u2 phi2 reg A* est}),
\begin{align*}
\| u_{1} \|_{H^{4}_{\unif}(\R)} +
\| \phi_{1} \|_{H^{2}_{\unif}(\R)} &\leq C(M), \\
\| u_{2} \|_{H^{4}_{\unif}(\R)} +
\| \phi_{2} \|_{H^{2}_{\unif}(\R)} &\leq C(M').
\end{align*}
The argument used to show (\ref{eq: w k' partial est}) holds for $k' \leq 2$, so
for $\xi \in H_{1}$
\begin{align*}
\sum_{|\alpha_{1}| \leq 4} \int | \partial^{\alpha_{1}} w|^{2} \xi^{2} &\leq C \bigg( \int \sum_{|\beta_{1}| \leq 2} |\partial^{\beta_{1}} \Delta w|^{2} \xi^{2} + \int \sum_{|\beta_{2}| \leq 2} |\partial^{\beta_{2}} w|^{2} \xi^{2} \bigg).
\end{align*}
Then, as (\ref{eq: Delta partial w j part 2 est}) holds with $j_{1} \leq 2$, applying this and (\ref{eq: w k' partial est}) for $k' = 0$ yields
\begin{align}
\sum_{|\alpha_{1}| \leq 4} \int | \partial^{\alpha_{1}} w|^{2} \xi^{2}  &\leq
C \int \sum_{|\beta_{1}| \leq 2} \left( |\partial^{\beta_{1}} w|^{2} + |\partial^{\beta_{1}} \psi|^{2} \right) \xi^{2} \nonumber \\
&\leq
C \bigg( \int |\Delta w|^{2} \xi^{2} + \int \sum_{|\beta_{1}| \leq 1}|\partial^{\beta_{1}} w|^{2} \xi^{2} + \sum_{|\beta_{2}| \leq 2} |\partial^{\beta_{2}} w|^{2} \xi^{2} \bigg). \label{eq: w A est}
\end{align}
Similarly, the argument used to show (\ref{eq: psi k' partial est}) holds for $k' = 0$, to give
\begin{align}
\sum_{|\alpha_{2}| \leq 2} \int | \partial^{\alpha_{2}} \psi|^{2} \xi^{2} &\leq C \bigg( \int |\Delta \psi|^{2} \xi^{2} + \int \sum_{|\beta_{2}| \leq 1} |\partial^{\beta_{2}} \psi|^{2} \xi^{2} \bigg). \label{eq: psi A est}
\end{align}
Finally, combining (\ref{eq: w A est})--(\ref{eq: psi A est}) and applying (\ref{eq: w psi H2 xi alt est}) from Lemma \ref{Lemma - Exp est lemma 2}, we obtain the desired estimate (\ref{eq: w and psi partial xi global final lemma est}) with $k = 0$
\begin{align*}
\sum_{|\alpha_{1}| \leq 4} \int &| \partial^{\alpha_{1}} w|^{2} \xi^{2} + \sum_{|\alpha_{2}| \leq 2} \int | \partial^{\alpha_{2}} \psi|^{2} \xi^{2} \\ &\leq
C \bigg( \int \left( |\Delta w|^{2} + |\Delta \psi|^{2} \right) \xi^{2} + \int \sum_{|\beta_{1}| \leq 1}| \left( \partial^{\beta_{1}} w|^{2} + \partial^{\beta_{1}} \psi|^{2} \right) \xi^{2} \bigg) 
\leq C \int R_{m}^{2} \xi^{2}.
\end{align*} 
The argument used in Case 1 holds for $k = 0$ to show the desired estimate (\ref{eq: w and psi pointwise rhs exp integral final lemma est})
\begin{equation*}
\sum_{|\alpha_{1}| \leq 2} |\partial^{\alpha_{1}} w(y)|^{2} + |\psi(y)|^{2} \leq C \int |R_{m}(x)|^{2} e^{-2 \gamma |x-y|} \id x. \qedhere
\end{equation*} 
\end{proof} 
We have now established all technical prerequisites to prove Theorems
\ref{Theorem - One inf pointwise stability estimate alt} and \ref{Theorem -
  Exponential Est Integral RHS}.

\begin{proof}[Proof of \Cref{Theorem - Exponential Est Integral RHS}]
Applying Lemmas \ref{Lemma - Exp est lemma 1} -- \ref{Lemma - Exp Est Final Est} with the assumption (B) yields the desired estimates (\ref{eq: w and psi partial xi global est})--(\ref{eq: w and psi pointwise rhs exp integral est}).
\end{proof}

\begin{proof}[Proof of \Cref{Theorem - One inf pointwise stability estimate alt}]\emph{Case 1.} Suppose $\spt(m_{2})$ is bounded and $m_{2} \not\equiv 0$. We show assumption (A) is satisfied, so by applying Lemmas \ref{Lemma - Exp est lemma 1} -- \ref{Lemma - Exp Est Final Est} we obtain the desired estimates (\ref{eq: w and psi partial xi global onesided est})--(\ref{eq: w and psi pointwise rhs exp integral onesided est}).

Since $m_{2} \in L^{2}_{\unif}(\R)$, it follows that $m_{2} \in L^{1}(\R)$ and since $m_{2} \geq 0$ and $m_{2} \not\equiv 0$, it follows that $\int m_{2} > 0$. Then, define the minimisation problem
\begin{align*}
I^{\TFW}(m_{2}) = \inf \left\{ \, E^{\TFW}(v,m_{2})  \, \bigg| \, v \in H^{1}(\R), v \geq 0, \int_{\R} v^{2} = \int_{\R} m_{2} > 0 \, \right \},
\end{align*}
which yields a unique solution $(u_{2},\phi_{2})$ to (\ref{eq: u phi Rn pair}), satisfying $u_{2} > 0$, using \cite[Theorem 7.19]{Lieb_Summary}. Applying Proposition \ref{Proposition - General Regularity Est}, we obtain the uniform estimates
\begin{align*}
\| u_{2} \|_{H^{4}_{\unif}(\R)} + \| \phi_{2} \|_{H^{2}_{\unif}(\R)} &\leq C(M'),
\end{align*}
\emph{Case 2.} Suppose $m_{2} = u_{2} = \phi_{2} = 0$, then by definition $(u_{2},\phi_{2})$ solve (\ref{eq: u phi eq pair}) and (A) is satisfied, so Lemmas \ref{Lemma - Exp est lemma 1} -- \ref{Lemma - Exp Est Final Est} imply (\ref{eq: w and psi partial xi global onesided est})--(\ref{eq: w and psi pointwise rhs exp integral onesided est}).

\emph{Case 3.} \label{proof-case2} Suppose $\spt(m_{2})$ is unbounded. By \Cref{Proposition - Finite Regularity Est}, there exists $(u_{2},\phi_{2})$ solving (\ref{eq: u phi eq pair}) corresponding to $m_{2}$ and satisfying $u_{2} \geq 0$. As we can not guarantee that $u_{2} > 0$, we can not apply Lemmas \ref{Lemma - Exp est lemma 1} -- \ref{Lemma - Exp Est Final Est} directly to compare $(u_{1},\phi_{1})$ with $(u_{2},\phi_{2})$. Instead we follow the proof of \Cref{Proposition - Finite Regularity Est} and use a thermodynamic limit argument to construct a sequence of functions $(u_{2,R_{n}},\phi_{2,R_{n}})$ that satisfy (A) for sufficiently large $R_{n}$, which converges to $(u_{2},\phi_{2})$.

Let $R_{n} \uparrow \infty$ and define $m_{2,R_{n}} := m_{2} \cdot \chi_{B_{R_{n}}(0)}$, then as $m_{2} \in L^{2}_{\unif}(\R)$, $m_{2} \geq 0$ and $m_{2} \not \equiv 0$, it follows that $m_{2,R_{n}} \in L^{1}(\R)$ and for sufficiently large $R_{n}$, $\int m_{2,R_{n}} > 0$. By \Cref{Proposition - Finite Regularity Est}, the minimisation problem
\begin{align*}
I^{\TFW}(m_{2,R_{n}}) = \inf \left\{ \, E^{\TFW}(v,m_{2,R_{n}})  \, \bigg| \, v \in H^{1}(\R), v \geq 0, \int_{\R} v^{2} = \int_{\R} m_{2,R_{n}} \, \right \},
\end{align*}
defines a unique solution $(u_{2,R_{n}},\phi_{2,R_{n}})$ to (\ref{eq: u phi Rn pair}), satisfying $u_{2,R_{n}} > 0$ and
\begin{align}
\| u_{2,R_{n}} \|_{H^{4}_{\unif}(\R)} + \| \phi_{2,R_{n}} \|_{H^{2}_{\unif}(\R)} &\leq C(M'), \label{eq: u 2Rn unif est}
\end{align}
where the constant is independent of $R_{n}$. Passing to the limit in (\ref{eq: u 2Rn unif est}), there exist $u_{2} \in H^{4}_{\unif}(\R), \phi_{2} \in H^{2}_{\unif}(\R)$ such that, respectively, along a subsequence $u_{2,R_{n}}, \phi_{2,R_{n}}$ converges to $u_{2}, \phi_{2}$, weakly in $H^{4}(B_{R}(0))$ and $H^{2}(B_{R}(0))$, strongly in $H^{2}(B_{R}(0))$ and $L^{2}(B_{R}(0))$ for all $R>0$ and for all $|\alpha| \leq 2$, $\partial^{\alpha} u_{2,R_{n}}, \phi_{2,R_{n}}$ converges to $\partial^{\alpha} u_{2}, \phi_{2}$ pointwise. It follows that $(u_{2}, \phi_{2})$ is a solution of (\ref{eq: u phi eq pair}) corresponding to $m_{2}$, satisfying $u_{2} \geq 0$ and (\ref{eq: u2 phi2 reg est})
\begin{equation*}
\| u_{2} \|_{H^{4}_{\unif}(\R)} + \| \phi_{2} \|_{H^{2}_{\unif}(\R)} \leq C(M').
\end{equation*}
In addition, $(u_{1}',\phi_{1}') = (u_{1},\phi_{1})$ and $(u_{2}',\phi_{2}') = (u_{2,R_{n}},\phi_{2,R_{n}})$ satisfy assumption (A) for large $R_{n}$, so by Lemmas \ref{Lemma - Exp est lemma 1} -- \ref{Lemma - Exp Est Final Est} that there exist $C,\gamma > 0$, independent of $R_{n}$, such that for large $R_{n}$ and any $\xi \in H_{\gamma}$
\begin{align}
\int_{\R} \bigg( \sum_{|\alpha_{1}| \leq 4} |\partial^{\alpha_{1}} (u_{1} - u_{2,R_{n}})|^{2} + \sum_{|\alpha_{2}| \leq 2} |\partial^{\alpha_{2}} (\phi_{1} - \phi_{2,R_{n}})|^{2} \bigg) \xi^{2} \leq C \int_{\R} (m_{1} - m_{2,R_{n}})^{2} \xi^{2}, \label{eq: u1 u2Rn xi est}
\end{align}
and for any $y \in \R$,
\begin{align}
\sum_{|\alpha_{1}| \leq 2} |\partial^{\alpha_{1}} (u_{1} - u_{2,R_{n}})(y)|^{2} + |(\phi_{1} - \phi_{2,R_{n}})(y)|^{2} \leq C \int_{\R} |(m_{1} - m_{2,R_{n}})(x)|^{2} e^{-2\gamma |x - y|} \id x. \label{eq: u1 u2Rn exp est}
\end{align}
Using the pointwise convergence of $(u_{2,R_{n}},\phi_{2,R_{n}})$ to $(u_{2},\phi_{2})$, applying the Dominated Convergence Theorem and sending $R_{n} \to \infty$ in (\ref{eq: u1 u2Rn xi est})--(\ref{eq: u1 u2Rn exp est}) we obtain the desired estimates (\ref{eq: w and psi partial xi global onesided est})--(\ref{eq: w and psi pointwise rhs exp integral onesided est}).
\end{proof}

\subsection{Proofs of Applications}
\label{Subsection - Proofs of Applications}
The proof of \Cref{Proposition - Infinite Finite Ground state comparison} is an 
application of \Cref{Theorem - One inf pointwise stability estimate alt}.

\begin{proof}[Proof of \Cref{Proposition - Infinite Finite Ground state comparison}]

Observe that $(u_{1},\phi_{1}) = (u,\phi)$ and $(u_{2},\phi_{2}) = (u_{\Omega},\phi_{\Omega})$ satisfy the conditions of \Cref{Theorem - One inf pointwise stability estimate alt}, there exist $C,\wt \gamma > 0$, independent of $\Omega$, such that for all $y \in \R$
\begin{align*}
\sum_{|\alpha| \leq 2} |\partial^{\alpha} (u - u_{\Omega})(y)|^{2} + |(\phi - \phi_{\Omega})(y)|^{2} &\leq C \int_{\R} |(m - m_{\Omega})(x)|^{2} e^{-2\gamma |x - y|} \id x.
\end{align*}
Now let $y \in \Omega$, $d = {\rm dist}(y, \partial\Omega)$ and observe that
$m - m_{\Omega} \in
L^{2}_{\unif}(\R)$.
Since $\sup_{x \in A} e^{- 2\wt \gamma |x|} \leq C \inf_{x \in A} e^{-2 \wt \gamma |x|}$ for any
  $A \subset B_1(z), z \in \R,$ with $C = C(\wt \gamma)$ independent of $z$, we have the bound
\begin{displaymath}
  \int_{B_d(y)^c} |(m - m_{\Omega})(x)|^{2} e^{- 2 \wt \gamma |x-y|} \id x \leq C \left( \|m\|_{L^2_\unif(\R)}^2 + \|m_{\Omega}\|_{L^2_\unif(\R)}^2 \right) \int_{B_{d}(y)^c} e^{- 2\wt \gamma |x-y|} \id x.
\end{displaymath}
Therefore, we obtain the desired estimate (\ref{eq: w and psi Rbuf est})
\begin{align*}
\int_{\R} &|(m - m_{\Omega})(x)|^{2} e^{-2\gamma |x - y|} \id x =
\int_{\Omega^{c}} |(m - m_{\Omega})(x)|^{2} e^{-2\gamma |x - y|} \id x \\ &\leq \int_{\Omega_{\rm buf}^{\rm c}} m(x)^{2} e^{-2 \wt \gamma |x - y|} \id x \leq C M^2 \int_{B_{R_{\rm buf}}(0)^{\rm c}} e^{-2\wt\gamma |x-y|} \id x 
\leq C M^2 \int_{B_d(y)^c} e^{-2\wt\gamma |x-y|} \id x \\
&= C M^2 (1+d^2) e^{-2\wt \gamma d} 
\leq C M^2 e^{- 2\gamma d},
\end{align*}
for any given $0 < \gamma < \wt\gamma$, where $C = C(\wt\gamma,\gamma)$.
\end{proof}

Next, we now prove \Cref{Corollary - Exponential Estimates Consequences} as a
direct consequence of Theorems \ref{Theorem - One inf pointwise stability estimate alt} and \ref{Theorem - Exponential Est Integral RHS}. 

\begin{proof}[Proof of \Cref{Corollary - Exponential Estimates Consequences}] 
  Let $k \in \mathbb{N}_{0}$ and
  $m_{1}, m_{2} \in \mathcal{M}_{H^{k}}(M,\omega)$ and recall the estimate
  (\ref{eq: w and psi pointwise rhs exp integral est}) of \Cref{Theorem -
    Exponential Est Integral RHS}, that there exists $C, \wt \gamma > 0$ such
  that
\begin{align}
\sum_{|\alpha_{1}| \leq k+2} |\partial^{\alpha_{1}} w(y)|^{2} + \sum_{|\alpha_{2}| \leq k} |\partial^{\alpha_{2}} \psi(y)|^{2} &\leq C \int \sum_{|\beta| \leq k} |\partial^{\beta} R_{m}(x)|^{2} e^{-2 \wt \gamma |x-y|} \id x. \label{eq: Local est}
\end{align}

(1) $R_m$ having compact support is a special case of exponential
decay, hence we consider only the case
$\sum_{|\beta| \leq k} |\partial^{\beta} R_{m}(x)|^{2} \leq C e^{- 2\gamma' |x -
  z|}$.
It is straightforward to see that there exists $C, \gamma > 0$, depending on
$\wt \gamma, \gamma'$, such that
\begin{align}
\int_{\R} e^{- 2\gamma' |x - z|} e^{-2 \wt \gamma |x - y|} \id x \leq C e^{- 2 \gamma |y - z|}.
\label{eq: exp exp est}
\end{align}
Hence (\ref{eq: w psi local exp est}) follows immediately from 
combining (\ref{eq: Local est}) and (\ref{eq: exp exp est}).

(2) Suppose that $R_{m}$ satisfies the algebraic decay
$\sum_{|\beta| \leq k}| \partial^{\beta} R_{m}(x)|^{2} \leq C (1 + |x|)^{-2r}$.
It is again elementary to show that there exists $C = C(r) > 0$ such that for
all $y \in \R$
\begin{equation}
 \int_{\R} (1 + |x|)^{-2r} e^{-2\gamma |x-y|} \id x 
 \leq C (1 + |y|)^{-2r}. \label{eq: decay rhs est}
\end{equation}
Combining (\ref{eq: decay rhs est}) with (\ref{eq: Local est}) gives the desired estimate (\ref{eq: w psi decay est}).

(3) Now suppose that $R_{m} \in H^{k}(\R)$ and recall (\ref{eq: w and psi partial xi global est}) of \Cref{Theorem - Exponential Est Integral RHS}, that there exists $C, \wt \gamma > 0$ such that for all $\xi \in H_{\wt \gamma}$
\begin{align}
\int_{\R} \bigg( \sum_{|\alpha_{1}| \leq k+4} |\partial^{\alpha_{1}} w|^{2} + \sum_{|\alpha_{2}| \leq k+2} |\partial^{\alpha_{2}} \psi|^{2} \bigg) \xi^{2} \leq C \int_{\R} \sum_{|\beta| \leq k} |\partial^{\beta} R_{m}|^{2} \xi^{2}. \label{eq: Xi est}
\end{align}
For any $0 < \gamma \leq \wt \gamma$, the function $\xi_{\gamma}(x) = e^{-\gamma|x|} \in H_{\wt \gamma}$. Then substituting $\xi_{\gamma}$ into (\ref{eq: Xi est}) yields
\begin{align*}
  &\int_{\R} \bigg( \sum_{|\alpha_{1}| \leq k+4} |\partial^{\alpha_{1}} w(x)|^{2} + \sum_{|\alpha_{2}| \leq k+2} |\partial^{\alpha_{2}} \psi(x)|^{2} \bigg) e^{-2\gamma |x|} \id x \\
& \hspace{2cm} \leq C \int_{\R} \sum_{|\beta| \leq k} |\partial^{\beta} R_{m}(x)|^{2} e^{-2\gamma |x|} \id x 
\leq C \int_{\R} \sum_{|\beta| \leq k} |\partial^{\beta} R_{m}(x)|^{2} \id x.
\end{align*}
Sending $\gamma \to 0$ and applying the Dominated Convergence Theorem yields the desired estimate
(\ref{eq: w psi Sobolev est}).

Under the assumptions of Theorem \ref{Theorem - One inf pointwise stability
estimate alt} with $k = 0$, other than applying Theorem \ref{Theorem - One
inf pointwise stability estimate alt} instead of Theorem \ref{Theorem - Exponential Est Integral RHS}, the proof is identical.
\end{proof}

We turn to the proofs of the charge-neutrality estimates.

\begin{proof}[Proof of \Cref{Theorem - Neutrality Estimate}]
Recall that $\rho_{12} = m_{1} - u_{1}^{2} - m_{2} + u_{2}^{2}$. Let $R > 0$ and choose
  $\varphi_{R} \in C^{\infty}_{\textnormal{c}}(\R)$ satisfying
  $0 \leq \varphi_{R} \leq 1,$ $\varphi_{R} = 1$ on $B_{R}(0)$,
  $\varphi_{R} = 0$ outside $B_{R+1}(0)$ and
  $\|\varphi_R \|_{W^{2,\infty}(\R)} \leq c_\varphi$. Let $A_R := B_{R+1}(0)
  \setminus B_R(0)$.  Recall (\ref{eq: psi initial
    eq}), that the difference $\psi := \phi_{1} - \phi_{2}$ solves
\begin{equation}
  \label{eq:charge-neutr-prf-Delta-psi}
- \Delta \psi  = 4 \pi \rho_{12}
\end{equation}
pointwise. Testing (\ref{eq:charge-neutr-prf-Delta-psi}) with $\varphi_{R}$ and using integration
by parts yields
\begin{align*}
\int_{B_{R+1}(0)} \rho_{12} \varphi_{R} 
  = - \frac{1}{4 \pi} \int_{A_R} \psi \Delta \varphi_{R}&. \end{align*}
Since $\varphi_{R} = 1$ on $B_{R}(0)$, we deduce
\begin{align*}
  \int_{B_{R}(0)} \rho_{12} = - \frac{1}{4 \pi} \int_{A_R} \psi \Delta \varphi_{R} - \int_{A_R} \rho_{12} \varphi_{R},
\end{align*}
and hence
\begin{align}
\bigg| \int_{B_{R}(0)} \rho_{12} \bigg| \leq C \int_{A_R} 
  \left( |m_{1} - m_{2}| + |u_{1} - u_{2}| + |\phi_{1} - \phi_{2}| \right),
\label{eq: charge rearranged eq}
\end{align}
where $C$ depends only on $c_\varphi$. Observe that $|A_R| \leq C R^2$.

(1) By (\ref{eq: w psi local exp est}) of \Cref{Corollary - Exponential
  Estimates Consequences} there exists $C, \wt \gamma > 0$ such that
\begin{align*}
|(\phi_{1} - \phi_{2})(x)| + |(m_{1} - m_{2})(x)| +|(u_{1} - u_{2})(x)| \leq C e^{- \wt \gamma|x|}. 
\end{align*}
Then using (\ref{eq: charge rearranged eq}) we deduce
\begin{align}
\bigg| \int_{B_{R}(0)} \rho_{12} \bigg| &\leq C \int_{A_R} \left( |m_{1} - m_{2}| + |u_{1} - u_{2}| + |\phi_{1} - \phi_{2}| \right) \nonumber \\
&\leq C \int_{A_R} e^{-\wt \gamma |x|} \id x \leq C (1+R^{2}) e^{-\wt \gamma R}, \label{eq: exp charge est}
\end{align}
which implies (\ref{eq: local pert BR exp est}) for any $0 < \gamma < \wt\gamma$.


(2) Suppose now that  $|(m_{1}- m_{2})(x)| \leq C (1 + |x|)^{-r}$, then
 using (\ref{eq: charge rearranged eq}) we obtain
\begin{align*}
\bigg| \int_{B_{R}(0)} \rho_{12} \bigg| &\leq C \int_{A_R} (1+|x|)^{-r} \leq C (1+R)^{2-r}.
\end{align*}

(3) Suppose $m_{1} - m_{2} \in L^{2}(\R)$, then by \Cref{Corollary - Exponential
  Estimates Consequences}, $u_{1} - u_{2}, \phi_{1} - \phi_{2} \in
H^{2}(\R)$, hence by Proposition \ref{Proposition - General Regularity Est} $u_{1}^{2} - u_{2}^{2} \in L^{2}(\R)$.
Taking the Fourier transform,
$\widehat{f}(k) = \int_{\R} f(x) e^{-2\pi i k \cdot x} \id x$, of
\eqref{eq:charge-neutr-prf-Delta-psi} and rearranging, we obtain
\begin{align*}
\frac{\widehat{\rho}_{12}(k)}{|k|^2} = \pi \widehat{\psi}(k) \in L^{2}(\R).
\end{align*}
Arguing as in \cite{C/Ehr} we show that $0$ is a Lebesgue point for $\widehat{\rho}_{12}$. For $\varepsilon >0$,
\begin{align*}
\frac{1}{|B_{\varepsilon}(0)|} \int_{B_{\varepsilon}(0)} |\widehat{\rho}_{12}(k)| \id k &\leq \frac{1}{|B_{\varepsilon}(0)|} \left( \int_{B_{\varepsilon}(0)} |k|^4 \id k \right)^{1/2} \left( \int_{B_{\varepsilon}(0)} \frac{|\widehat{\rho}_{12}(k)|^{2}}{|k|^4} \id k \right)^{1/2} \\
&\leq C \varepsilon^{1/2} \| \phi_{1} - \phi_{2} \|_{L^{2}(\R)},
\end{align*}
which tends to 0 as $\varepsilon \to 0$, as claimed.
%
\end{proof}

\subsection{Proof of energy locality}
\label{sec:proofs energy locality}
To prove \Cref{Theorem - Forcing Est}, we first establish the existence, uniqueness and regularity of the solutions to the linearised TFW equations.

Fix $Y = (Y_j)_{j \in \mathbb{N}} \in \mathcal{Y}_{L^{2}}(M,\omega)$ and let
$m = m_{Y} \in \mathcal{M}_{L^{2}}(M,\omega)$.  Let
$V \in \R \smallsetminus \{0\}$, $k \in \mathbb{N}$ and for $h \in [0,1]$ define
\begin{align}
Y^{h} = \{ \, Y_{j} + \delta_{jk} h V  \, | \, j \in \mathbb{N} \, \}, \label{eq: Y h def}
\end{align}
and the associated nuclear configuration
\begin{align}
m_{h}(x) = m(x) + \eta(x - Y_{k} - h V) - \eta(x - Y_{k}). \label{eq: mh def}
\end{align}

\begin{lemma}
\label{Lemma - M1 omega1 Space}
There exist $M', \omega_{0}', \omega_{1}' > 0$,
such that for $\omega' = (\omega_{0}', \omega_{1}')$, $m_{h} \in \mathcal{M}_{L^{2}}(M',\omega')$ for all $h \in [0,1]$.
In particular, $Y^{h} \in \mathcal{Y}_{L^{2}}(M',\omega')$ for all $h \in [0,1]$.
\end{lemma}

\begin{proof} [Proof of \Cref{Lemma - M1 omega1 Space}]
Recall that $m_{h}, \eta \geq 0$, $\eta \in C^{\infty}(\R)$ and $\int_{\R} \eta = 1$, then
\begin{align*}
\sup_{x \in \R} \| m_{h} \|_{L^{2}(B_{1}(x))} &\leq \sup_{x \in \R} \left( \| m \|_{L^{2}(B_{1}(x))} + \left( \int_{B_{1}(x)} \eta(z- Y_{k} - hV)^{2} \id z \right)^{1/2} \right) \\ &\leq M + \| \eta \|_{L^{2}(\R)} =: M'.
\end{align*}

Since $m \in \mathcal{M}_{L^{2}}(M,\omega)$, with $\omega = (\omega_{0},\omega_{1})$, for all $R > 0$,
\begin{align*}
\inf_{x \in \R} \int_{B_{R}(x)} m_{h}(z) \id z \geq 
\inf_{x \in \R} \int_{B_{R}(x)} m(z) \id z - \int_{B_{R}(x)} \eta(z - Y_{k}) \id z \geq \omega_{0} R^{3} - \omega_{1} - 1,
\end{align*}
hence for $\omega' = (\omega_{0}, \omega_{1}+1)$,
$m_{h} \in \mathcal{M}_{L^{2}}(M',\omega')$ for all $h \in [0,1]$, as claimed.
%
\end{proof}

As $m_{h} \in \mathcal{M}_{L^{2}}(M',\omega')$ for all $h \in [0,1]$, by \Cref{Theorem - C/LB/L Main Result} there exists a corresponding ground state $(u_{h},\phi_{h})$. Also, let $(u,\phi) = (u_{0},\phi_{0})$. We now use \Cref{Corollary - Exponential Estimates Consequences} to compare $(u_{h},\phi_{h})$ with $(u,\phi)$ and rigorously linearise the TFW equations.

\begin{lemma}
\label{Lemma - Linearised TFW Results}
Let $Y \in \mathcal{Y}_{L^{2}}(M,\omega)$ and let $m = m_{Y} \in \mathcal{M}_{L^{2}}(M,\omega)$. Also, let $k \in \mathbb{N}$, $V \in \R \smallsetminus\{0\}$ and $h_{0} = \min\{ 1, |V|^{-1}\}$. For $h \in [0,h_{0}]$ define
\begin{align*}
m_{h}(x) = m(x) + \eta(x - Y_{k} - h V) - \eta(x - Y_{k}).
\end{align*}
There exist $C = C(M',\omega')$, $\gamma_{0} = \gamma_{0}(M',\omega') > 0$, independent of $h$ and $|V|$, such that
\begin{align}
\sum_{|\alpha| \leq 2} \big( |\partial^{\alpha}(u_{h} - u)(x)| + |\partial^{\alpha}(\phi_{h} - \phi)(x)| \big) + |(m_{h} - m)(x)| &\leq C h e^{-\gamma_{0} |x - Y_{k}|}, \label{eq: uh - u0 exp est}
\\ \|u_{h} - u\|_{H^{4}(\R)} + \|\phi_{h} - \phi\|_{H^{2}(\R)} \leq C \| m_{h} - m \|_{L^{2}(\R)} &\leq C h. \label{eq: uh - u0 L2 est}
\end{align}
Moreover, the limits
\begin{align*}
\ou = \lim_{h \to 0} \frac{u_{h} - u}{h}, \quad \ophi = \lim_{h \to 0} \frac{\phi_{h} - \phi}{h}, \quad \om = \lim_{h \to 0} \frac{m_{h} - m}{h},
\end{align*}
exist and are the unique solution to the linearised TFW equations
\begin{subequations}
\label{eq: ou ophi eq pair}
\begin{align}
- \Delta \ou &+ \left( \frac{35}{9} u^{4/3} - \phi \right) \ou - u \ophi = 0, \label{eq: ou eq} \\
- \Delta \ophi &= 4 \pi \left( \om - 2u \ou \right). \label{eq: ophi eq}
\end{align}
\end{subequations}
Moreover, $\ou \in H^{4}(\R), \ophi \in H^{2}(\R), \om \in C^{\infty}_{\textnormal{c}}(\R)$ and satisfy
\begin{align}
&\sum_{|\alpha| \leq 2} \big( |\partial^{\alpha}\ou(x)| + |\partial^{\alpha}\ophi(x)| \big) + |\om(x)| \leq C e^{- \gamma_{0} |x - Y_{k}|}, \label{eq: ou ophi exp est} \\
&\|\ou\|_{H^{4}(\R)} + \|\ophi\|_{H^{2}(\R)} \leq C \| \om \|_{L^{2}(\R)}. \label{eq: ou ophi L2 est}
\end{align}

\end{lemma}

\begin{proof}[Proof of \Cref{Lemma - Linearised TFW Results}]
By \Cref{Proposition - General Regularity Est} and \Cref{Proposition - Uniform inf u estimate}, for $h \in [0,h_{0}]$ the ground state $(u_{h},\phi_{h})$ satisfies
\begin{align}
&\| u_{h} \|_{H^{4}_{\unif}(\R)} + \| \phi_{h} \|_{H^{2}_{\unif}(\R)} \, \leq C(M'), \label{eq: uh phih reg est} \\
&\inf_{x \in \R} u_{h}(x) \geq c_{M',\omega'} > 0, \label{eq: inf uh est}
\end{align}
independently of $h$. From (\ref{eq: mh def}), it follows that
\begin{align}
|(m_{h}-m)(x)| &= |\eta(x - Y_{k} - h V) - \eta(x - Y_{k})| \nonumber \\ &\leq h |V| \int_{0}^{1} |\nabla \eta(x - Y_{k} - thV )| \id t \leq \int_{0}^{1} |\nabla \eta(x - Y_{k} - thV )| \id t. \label{eq: m mh rhs est}
\end{align}
For all $h \in [0,h_{0}]$, $ \spt(m_{h} - m) \subset B_{R_{0}+1}(Y_{k})$, so by \Cref{Corollary - Exponential Estimates Consequences} and (\ref{eq: m mh rhs est}) it follows that there exists $\gamma_{0} > 0$ such that
\begin{align}
&\sum_{|\alpha| \leq 2} |\partial^{\alpha}(u_{h} - u)(x)| + |(\phi_{h} - \phi)(x)| + |(m_{h} - m)(x)| \leq C h e^{-\gamma_{0} |x - Y_{k}|}, \label{eq: uh phih diff exp est 1}
\end{align}
and (\ref{eq: uh - u0 L2 est}) holds
\begin{align}
&\|u_{h} - u\|_{H^{4}(\R)} + \|\phi_{h} - \phi\|_{H^{2}(\R)} \leq C \| m_{h} - m \|_{L^{2}(\R)} \leq C h. \label{eq: u h - u 0 L2 est}
\end{align}
Due to the uniform estimates (\ref{eq: uh phih reg est})--(\ref{eq: inf uh est}) and (\ref{eq: m mh rhs est}), the constants appearing on the right-hand side are independent of $h$.

We now show
\begin{align}
\sum_{|\alpha| \leq 2} |\partial^{\alpha} (\phi_{h} - \phi)(x)| \leq C e^{-\gamma_{0}|x - Y_{k}|}. \label{eq: phi h C2 exp est}
\end{align}
Observe that for $h \in (0,h_{0}]$ as $\spt(m_{h} - m) \subset B_{R_{0}+1}(Y_{k})$, by the triangle inequality $x \in B_{R_{0}+3}^{\rm c}(Y_{k})$ implies $B_{2}(x) \subset B_{R_{0}+1}^{\rm c}(Y_{k})$. Consequently, for $x \in B_{R_{0}+3}^{\rm c}(Y_{k})$
\begin{align}
\| m_{h} - m \|_{C^{0,1/2}(B_{2}(x))} = 0, \label{eq: mh C0 part 1}
\end{align}
and for $x \in B_{R_{0}+3}(Y_{k})$, by (\ref{eq: mh def}) it follows that
\begin{align}
\| m_{h} - m \|_{C^{0,1/2}(B_{2}(x))} \leq 2 \| \eta \|_{C^{0,1/2}(B_{2}(x))}. \label{eq: mh C0 part 2}
\end{align} 
By (\ref{eq: mh C0 part 1})--(\ref{eq: mh C0 part 2}) we deduce that $x \mapsto \| m_{h} - m_{0} \|_{C^{0,1/2}(B_{2}(x))}$ is a bounded function with support in $B_{R_{0}+3}(Y_{k})$, hence there exists $C > 0$ such that
\begin{align}
\| m_{h} - m \|_{C^{0,1/2}(B_{2}(x))} \leq C e^{-\gamma_{0}|x - Y_{k}|}. \label{eq: mh diff exp est}
\end{align}
Then we apply the Schauder estimates \cite[Theorem 10.2.1, Lemma
10.1.1]{Jost_PDE} together with (\ref{eq: uh phih diff exp est 1}) and (\ref{eq:
  mh diff exp est}) to estimate
\begin{align}
\| \phi_{h} - \phi \|_{C^{2,1/2}(B_{1}(x))} &\leq C \left( \| m_{h} - m - u_{h}^{2} + u^{2} \|_{C^{0,1/2}(B_{2}(x))} + \| \phi_{h} - \phi \|_{L^{2}(B_{2}(x))} \right),\nonumber \\
&\leq C \left( \| m_{h} - m \|_{C^{0,1/2}(B_{2}(x))} + \| u_{h}^{2} - u^{2} \|_{C^{0,1/2}(B_{2}(x))} + \| \phi_{h} - \phi \|_{L^{2}(B_{2}(x))} \right),\nonumber \\
&\leq C \left( \| (u_{h}+u )( u_{h} - u ) \|_{C^{0,1/2}(B_{2}(x))} + e^{-\gamma_{0}|x - Y_{k}|} \right),\nonumber \\
&\leq C \left( \| u_{h}+u \|_{C^{0,1/2}(B_{2}(x))} \| u_{h} - u  \|_{C^{0,1/2}(B_{2}(x))} + e^{-\gamma_{0}|x - Y_{k}|} \right). \label{eq: phi C2 est part 1}
\end{align}
Applying the Sobolev embedding $C^{0,1/2}(B_{2}(x)) \hookrightarrow H^{2}(B_{2}(x))$ and using (\ref{eq: inf uh est}), it follows that
\begin{align}
\| u_{h}+u \|_{C^{0,1/2}(B_{2}(x))} \leq C \| u_{h}+u \|_{H^{2}(B_{2}(x))} \leq C \left( \| u_{h} \|_{H^{2}_{\unif}(\R)} + \| u \|_{H^{2}_{\unif}(\R)} \right) \leq C. \label{eq: uh + u est}
\end{align}
Applying (\ref{eq: uh + u est}) and (\ref{eq: uh phih diff exp est 1}) to (\ref{eq: phi C2 est part 1}), we obtain the desired estimate (\ref{eq: phi h C2 exp est}): for any multi-index $\alpha$ satisfying $|\alpha| \leq 2$
\begin{align}
|\partial^{\alpha}(\phi_{h} - \phi)(x)| &\leq \| \phi_{h} - \phi \|_{W^{2,\infty}(B_{1}(x))} \leq \| \phi_{h} - \phi \|_{C^{2,1/2}(B_{1}(x))} \nonumber \\
&\leq C \left( \| u_{h}+u \|_{C^{0,1/2}(B_{2}(x))} \| u_{h} - u  \|_{C^{0,1/2}(B_{2}(x))} + e^{-\gamma_{0}|x - Y_{k}|} \right) \nonumber \\
&\leq C \left( \| u_{h} - u  \|_{W^{1,\infty}(B_{2}(x))} + e^{-\gamma_{0}|x - Y_{k}|} \right) 
\leq C e^{-\gamma_{0}|x - Y_{k}|}. \label{eq: phi h diff exp est upto 2}
\end{align}

We will show next that there exist $\ou \in H^{4}(\R), \ophi \in H^{2}(\R)$
such that $\smfrac{u_{h}- u}{h}, \smfrac{\phi_{h}- \phi}{h}$ converge to
$\ou, \ophi$ respectively, weakly in $H^{4}(\R)$ and $H^{2}(\R)$, strongly in
$H^{3}(B_{R}(0))$ and $H^{1}(B_{R}(0))$ for all $R > 0$ and pointwise almost
everywhere, along with their derivatives as $h \to 0$. 

First consider any decreasing sequence $h_{n} \to 0$, then there exists a
subsequence (still denoted by $h_{n}$) such that
$\smfrac{u_{h_{n}}- u}{h_{n}}, \smfrac{\phi_{h_{n}}- \phi}{h_{n}}$ converge to
$\ou \in H^{4}(\R), \ophi \in H^{2}(\R)$ respectively, weakly in $H^{4}(\R)$ and
$H^{2}(\R)$, strongly in $H^{3}(B_{R}(0))$ and $H^{1}(B_{R}(0))$ for all $R > 0$
and pointwise almost everywhere, along with their derivatives. In addition, it follows that $(\ou,\ophi)$ satisfy (\ref{eq: ou ophi exp est})--(\ref{eq: ou ophi L2 est}).

We now verify
that the limiting functions are independent of the choice of sequence. First,
observe that by passing to the limit as $h_{n} \to 0$ in the equations
\begin{align*}
&- \Delta \left( \frac{ u_{h_{n}} - u }{h_{n}}  \right) + \frac{5}{3} \frac{ u_{h_{n}}^{7/3} - u^{7/3} }{h_{n}} - \frac{ \phi_{h_{n}} u_{h_{n}} - \phi u }{h_{n}}  = 0, \\
&- \Delta \left( \frac{ \phi_{h_{n}} - \phi }{h_{n}} \right) = 4 \pi \left( \frac{ m_{h_{n}} - m }{h_{n}} - \frac{ u_{h_{n}}^{2} - u^{2} }{h_{n}}  \right),
\end{align*}
it follows that $(\ou,\ophi)$ solve the linearised TFW equations (\ref{eq: ou
  ophi eq pair}) pointwise, 
\begin{align*}
  \notag
- \Delta \ou &+ \left( \frac{35}{9} u^{4/3} - \phi \right) \ou - u \ophi = 0, \\
  \notag
- \Delta \ophi &= 4 \pi \left( \om - 2u \ou \right), \\
  \text{where} \quad
\om(x) &= \lim_{h_{n} \to 0} \frac{(m_{h_{n}}- m)(x)}{h_{n}} = - \nabla \eta (x - Y_{k}) \cdot V.
\end{align*}
Clearly $\om$ is independent of the sequence $h_{n}$. Applying \cite[Corollary 2.3]{Blanc_Uniqueness}, it follows that the $(\ou, \ophi)$ is the unique solution to the linearised system (\ref{eq: ou ophi eq pair}), hence is independent of the sequence $(h_{n})$. It then follows that $\smfrac{u_{h}- u}{h}, \smfrac{\phi_{h}- \phi}{h}$ converge to $\ou, \ophi$ as $h \to 0$ as stated above.
\end{proof}


We are now in a position to prove \Cref{Theorem - Forcing Est}. 

\begin{proof}[Proof of \Cref{Theorem - Forcing Est}]
We will repeatedly use the fact that there exists $C, \gamma > 0$ such that, for all $h \in [0,h_{0}]$, $p \in [1,2]$,
\begin{align}
\int_{\R} (1 + m_{h}(x) + |\nabla \phi_{h}(x)|)^{p} e^{- \gamma_{0}|x-Y_{k}|} e^{- \wt \gamma|x-Y_{j}|} \id x \leq C e^{-\gamma |Y_{j} - Y_{k}|}, \label{eq: 1 + m exp est}
\end{align}
which is a consequence of the uniform bounds on $m_h, \phi_h$ and of \eqref{eq:
  exp exp est}.


Further, we require that there exist $C, \wt \gamma > 0$ such that, for $j \in
\mathbb{N}$, $h \in (0,h_{0}]$, $x \in \R$, 
\begin{align}
\left| \frac{ \varphi_j(Y^{h};x) - \varphi_j(Y;x)  }{h}  \right| \leq C e^{-\wt \gamma |x - Y_{j}|} e^{-\wt \gamma |x - Y_{k}|}, \label{eq: varphi deriv pre-limit est}
\end{align}
which follows directly from (\ref{eq: partition condition 2}).


For $i = 1,2$ and $j \in \mathbb{N}$, consider the difference
\begin{align}
&\frac{E^{i}_{j}(Y^{h}) - E^{i}_{j}(Y)}{h} = \int_{\R} \frac{ \mathcal{E}_{i}(Y^{h};x) \varphi_{j}(Y^{h};x) - \mathcal{E}_{i}(Y;x) \varphi_{j}(Y;x)}{h} \id x \nonumber \\
&= \int_{\R} \left( \frac{ \mathcal{E}_{i}(Y^{h};x) - \mathcal{E}_{i}(Y;x) }{h} \right) \varphi_{j}(Y^{h};x) \id x  \nonumber \\
 & \quad
   + \int_{\R} \mathcal{E}_{i}(Y;x) \left( \frac{ \varphi_{j}(Y^{h};x) - \varphi_{j}(Y;x) }{h} \right) \id x \label{eq: E diff two part est 2}
\end{align}
We wish to show that the limit of (\ref{eq: E diff two part est 2})  exists as
$h \to 0$ to obtain
\begin{align}
\frac{\partial E^{i}_{j} }{\partial Y_{k}} &= \int_{\R} \frac{ \partial \mathcal{E}_{i} }{\partial Y_{k} }(Y;x) \varphi_{j}(Y;x) \id x + \int_{\R} \mathcal{E}_{i}(Y;x) \frac{ \partial \varphi_{j}}{\partial Y_{k}}(Y;x) \id x, \label{eq: Forcing limit eq}
\end{align}
where
\begin{align}
\frac{ \partial \mathcal{E}_{1} }{\partial Y_{k} }(Y;\cdot) &= 2 \nabla u \cdot \nabla \ou  + \frac{10}{3} u^{7/3} \ou + \frac{1}{2} \ophi ( m - u^{2} ) + \frac{1}{2} \phi ( \om - 2u \ou ), \label{eq: curly E 1 derivative eq} \\
\frac{ \partial \mathcal{E}_{2} }{\partial Y_{k} }(Y;\cdot) &= 2 \nabla u \cdot \nabla \ou  + \frac{10}{3} u^{7/3} \ou + \frac{1}{4\pi} \nabla \ophi \cdot \nabla \phi. \label{eq: curly E 2 derivative eq}
\end{align}

\emph{Case 1.} First consider the energy density
\begin{align}
\mathcal{E}_{1}(Y;x) = |\nabla u(x)|^{2} + u^{10/3}(x) + \frac{1}{2} \phi(x)(m-u^{2})(x). \label{eq: E1 density def}
\end{align}

To show (\ref{eq: curly E 1 derivative eq}), consider the difference
\begin{align}
\frac{\mathcal{E}_{1}(Y^{h};\cdot) - \mathcal{E}_{1}(Y;\cdot)}{h} &= \nabla(u_{h}+u) \cdot \nabla\left( \frac{ u_{h}- u}{h} \right) + \left( \frac{ u_{h}^{10/3}- u^{10/3}}{h} \right) \nonumber \\ & \quad + \frac{1}{2h} \left( \phi_{h}(m_{h}-u_{h}^{2}) - \frac{1}{2} \phi (m-u^{2}) \right) \nonumber \\
&= \nabla(u_{h}+u) \cdot \nabla\left( \frac{ u_{h}- u}{h} \right) + \left( \frac{ u_{h}^{10/3}- u^{10/3}}{h} \right) \nonumber \\
& \quad + \frac{1}{2} \left( \frac{\phi_{h} - \phi}{h} \right)(m -u^{2}) + \frac{1}{2} \phi_{h} \left( \frac{m_{h} - m -u_{h}^{2} + u^{2} }{h} \right). \label{eq: curly E 1 difference eq}
\end{align}

It follows from (\ref{eq: curly E 1 difference eq}) and pointwise convergence of $u_{h},\nabla u_{h},\phi_{h}$ to $u,\nabla u,\phi$ and \newline $\frac{u_{h}-u}{h},\nabla \left( \frac{u_{h}-u}{h} \right),\frac{\phi_{h}-\phi }{h},\frac{m_{h}-m}{h}$ to $\ou, \nabla \ou,\ophi,\om$ as $h \to 0$, that (\ref{eq: curly E 1 derivative eq}) holds
\begin{align*}
\lim_{h \to 0} \frac{\mathcal{E}_{1}(Y^{h};\cdot) - \mathcal{E}_{1}(Y;\cdot)}{h} = 2 \nabla u \cdot \nabla \ou  + \frac{10}{3} u^{7/3} \ou + \frac{1}{2} \ophi ( m - u^{2} ) + \frac{1}{2} \phi ( \om - 2u \ou ) = \frac{ \partial \mathcal{E}_{1} }{\partial Y_{k} } .
\end{align*}

Applying (\ref{eq: uh - u0 exp est}) to (\ref{eq: curly E 1 difference eq}) yields
\begin{align}
\left| \mathcal{E}_{1}(Y^{h};x) - \mathcal{E}_{1}(Y;x) \right| &\leq C \left( |(u_{h} - u)(x)| + |\nabla(u_{h} - u)(x)| + |(m_{h} - m)(x)| \right) \nonumber \\ 
& \qquad + C (1 + m(x)) |(\phi_{h}-\phi)(x)| \nonumber \\ &\leq C h (1 + m(x)) e^{- \gamma_{0} |x - Y_{k}|}. \label{eq: curly E derivative h est}
\end{align}
Combining (\ref{eq: curly E derivative h est}) and (\ref{eq: partition condition
  1}), we deduce
\begin{align}
\left| \frac{ \mathcal{E}_{1}(Y^{h};x) - \mathcal{E}_{1}(Y;x) }{h} \varphi_{j}(Y;x) \right| &\leq C (1 + m(x)) e^{- \gamma_{0}|x-Y_{k}|} e^{- \wt \gamma|x-Y_{j}|}, \label{eq: E curly 1 est}
\end{align}
hence by (\ref{eq: 1 + m exp est}) and the Dominated Convergence Theorem,
\begin{align}
\int_{\R} \frac{ \partial \mathcal{E}_{1} }{\partial Y_{k} }(Y;x) \varphi_{j}(Y;x) \id x = \lim_{h \to 0}\int_{\R} \left( \frac{ \mathcal{E}_{1}(Y^{h};x) - \mathcal{E}_{1}(Y;x) }{h} \right) \varphi_{j}(Y;x) \id x. \label{eq: curly E 1 limit}
\end{align}
It follows from (\ref{eq: E curly 1 est}) and (\ref{eq: 1 + m exp est}) that
\begin{align}
\left| \int_{\R} \frac{ \partial \mathcal{E}_{1} }{\partial Y_{k} }(Y;x) \varphi_{j}(Y;x) \id x \right| \leq C \int_{\R} (1+m(x)) e^{- \gamma_{0}|x-Y_{k}|} e^{- \wt \gamma|x-Y_{j}|} \id x \leq C e^{-\gamma |Y_{j} - Y_{k}|}. \label{eq: E curly 1 part 1 est}
\end{align}

It remains to show that (\ref{eq: E diff two part est 2}) converges using (\ref{eq: 1 + m exp est}) and (\ref{eq: varphi deriv pre-limit est}). As $\varphi_{j}(Y;x)$ is differentiable with respect to $Y_{k}$, for all $x \in \R$
\begin{align*}
\mathcal{E}_{1}(Y;x) \, \frac{ \partial \varphi_{j} }{\partial Y_{k}}(Y;x) = \lim_{h \to 0} \mathcal{E}_{1}(Y;x) \left( \frac{ \varphi_{j}(Y^{h};x) - \varphi_{j}(Y;x) }{h} \right),
\end{align*}
and combining (\ref{eq: E1 density def}) with (\ref{eq: varphi deriv pre-limit est}) implies
\begin{align*}
\Big| \mathcal{E}_{1}(Y;x) \left(\frac{ \varphi_{j}(Y^{h};x) - \varphi_{j}(Y;x) }{h} \right) \Big| &\leq C  ( 1 + m(x)) e^{- \gamma_{0}|x-Y_{k}|} e^{- \wt\gamma|x - Y_{j}|},
\end{align*}
hence by (\ref{eq: 1 + m exp est}) and the Dominated Convergence Theorem,
\begin{align}
  &\int_{\R} \mathcal{E}_{1}(Y;x) \frac{ \partial \varphi_{j} }{\partial Y_{k}}(Y;x) \id x = \lim_{h \to 0} \int_{\R} \mathcal{E}_{1}(Y;x) \, 
   \left( \frac{ \varphi_{j}(Y^{h};x) - \varphi_{j}(Y;x) }{h} \right)
 \id x,  \notag \\
  & \text{and} \qquad 
\left | \int_{\R} \mathcal{E}_{1}(Y;x) \frac{ \partial \varphi_{j} }{\partial Y_{k}}(Y;x) \id x \right | \leq C e^{-\gamma |Y_{j} - Y_{k}|}. \label{eq: E curly 1 part 2 est}
\end{align}
Combining (\ref{eq: E curly 1 part 1 est}) and (\ref{eq: E curly 1 part 2 est})
yields the desired estimate (\ref{eq: forcing est}).

The second case, using $\mathcal{E}_2$ instead of $\mathcal{E}_1$, is analogous.
\end{proof}

\begin{proof}[Proof of \textnormal{(\ref{eq: Forcing 1 2 equal est})}]
We will use
\begin{align}
\sum_{j \in \mathbb{N}} e^{-\gamma |Y_{j} - Y_{k}|} < \infty, \label{eq: sum exp Yj finite}
\end{align}
which is a consequence of (H1) and that $Y \in \mathcal{Y}_{L^{2}}(M,\omega)$.
Then for $i \in \{1,2\}$
\begin{align*}
\sum_{j \in \mathbb{N}} \left| \frac{\partial E^{i}_{j} }{\partial Y_{k}} \right| &\leq \sum_{j \in \mathbb{N}} \left| \int_{\R} \frac{ \partial \mathcal{E}_{i} }{\partial Y_{k} }(Y;x) \varphi_{j}(Y;x) \id x \right| + \sum_{j \in \mathbb{N}} \left| \int_{\R} \mathcal{E}_{i}(Y;x) \frac{ \partial \varphi_{j}}{\partial Y_{k}}(Y;x) \id x \right| \\ &\leq C \sum_{j \in \mathbb{N}} e^{-\gamma |Y_{j} - Y_{k}|} < \infty,
\end{align*}
hence by the Monotone Convergence Theorem, the sum is well-defined
\begin{align*}
\sum_{j \in \mathbb{N}} \frac{\partial E^{1}_{j} }{\partial Y_{k}} &= \int_{\R} \frac{ \partial \mathcal{E}_{i} }{\partial Y_{k} }(Y;x) \left( \sum_{j \in \mathbb{N}} \varphi_{j}(Y;x) \right) \id x + \int_{\R} \mathcal{E}_{i}(Y;x) \left( \sum_{j \in \mathbb{N}} \frac{ \partial \varphi_{j}}{\partial Y_{k}}(Y;x) \right) \id x.
\end{align*}
As $(\varphi_{j})_{j \in \mathbb{N}}$ satisfies (\ref{eq: sum 1 eq}) for all $h \in [0,h_{0}]$, it follows that
\begin{align*}
\sum_{j \in \mathbb{N}} \frac{\partial \varphi_{j} }{\partial Y_{k}}(Y; x) = 0,
\end{align*}
and consequently, 
\begin{align*}
\sum_{j \in \mathbb{N}} \frac{\partial E^{i}_{j} }{\partial Y_{k}} &= \int_{\R} \frac{ \partial \mathcal{E}_{i} }{\partial Y_{k} }(Y;x) \id x.
\end{align*}
Now consider the difference of (\ref{eq: curly E 1 derivative eq})--(\ref{eq: curly E 2 derivative eq})
 \begin{align}
 \left(\frac{\partial \mathcal{E}_{1}}{\partial Y_{k}} - \frac{\partial \mathcal{E}_{2}}{\partial Y_{k}} \right)(Y;\cdot) &= \frac{1}{2} \ophi (m - u^{2}) + \frac{1}{2} \phi (\om - 2u \ou) - \frac{1}{4\pi} \nabla \ophi \cdot \nabla \phi, \label{eq: E 1 2 deriv difference eq}
 \end{align}
and applying integration by parts yields
 \begin{align*}
 \int_{\R} \left(\frac{\partial \mathcal{E}_{1}}{\partial Y_{k}} - \frac{\partial \mathcal{E}_{2}}{\partial Y_{k}} \right)(Y;x) \id x &= \int_{\R} \left( \frac{1}{2} \ophi (m - u^{2}) + \frac{1}{2} \phi (\om - 2u \ou) - \frac{1}{4\pi} \nabla \ophi \cdot \nabla \phi \right)\\ &= \frac{1}{8\pi} \int_{\R} \left( \ophi (-\Delta \phi) + \phi (-\Delta \ophi) - 2 \nabla \ophi \cdot \nabla \phi \right) \\
 &= \frac{1}{8\pi} \int_{\R} \left( 2\nabla \ophi \cdot\nabla \phi - 2 \nabla \ophi \cdot \nabla \phi \right) = 0.
 \end{align*}
 In addition, since
 \begin{align*}
 \frac{1}{4\pi} \int_{\R} \nabla \phi \cdot \nabla \ophi = \frac{1}{4\pi} \int_{\R} \phi (-\Delta \ophi) = \int_{\R} \phi (\om - 2u \ou)
 \end{align*}
 and since $u$ solves (\ref{eq: u inf eq}),
 $- \Delta u + \frac{5}{3} u^{7/3} - \phi u = 0,$
 the desired result (\ref{eq: Forcing 1 2 equal est}) holds:
 \begin{equation*}
 \int_{\R} \frac{\partial \mathcal{E}_{2}}{\partial Y_{k}}(Y;x) \id x = 2 \int_{R} \left( \nabla u \cdot \nabla \ou + \frac{5}{3} u^{7/3} \ou - \phi u \ou \right) + \int_{\R} \phi \, \om = \int_{\R} \phi \, \om. \qedhere
 \end{equation*}
\end{proof}

Finally, we establish (\ref{eq: TFW energy 1})--(\ref{eq: TFW energy 2}): if the
partition functions $\varphi_j$ are invariant under permutations and isometries,
then so are the site energies.

\begin{lemma}
\label{Lemma - Site Energy Invariance}
If the partition $(\varphi_j)_{j \in \mathbb{N}}$ is permutation and  isometry
invariant \textnormal{(\ref{eq: Permutation Invariance def})--(\ref{eq: Isometry
    Invariance def})}, then for $i =1,2$, for any bijection $P: \mathbb{N} \to \mathbb{N}$, isometry $A: \R \to \R$, $j \in \mathbb{N}$ and $Y \in \mathcal{Y}_{L^{2}}(M,\omega)$
\begin{align}
E^{i}_{j}(Y \circ P) = E^{i}_{j}(Y), \label{eq:SiteEnergyPermInv} \\
E^{i}_{j}(AY) = E^{i}_{j}(Y). \label{eq:SiteEnergyIsoInv}
\end{align}
\end{lemma}

\begin{proof}[Proof of \Cref{Lemma - Site Energy Invariance}]
Let $Y \in \mathcal{Y}_{L^{2}}(M,\omega)$ and $m = m_{Y}$, then as $P: \mathbb{N} \to \mathbb{N}$ is a bijection,
\begin{align*}
m_{Y \circ P}(x) = \sum_{j\in \mathbb{N}} \eta(x - Y_{P_{j}}) = \sum_{j\in \mathbb{N}} \eta(x - Y_{j}) = m_{Y}(x).
\end{align*} 
Since (\ref{eq: u phi eq pair}) has a unique solution,
$(u_{Y},\phi_{Y}) = (u_{Y \circ P}, \phi_{Y \circ P})$. Consequently, the energy
densities agree, $\mathcal{E}_{i}(Y\circ P;\cdot) =
\mathcal{E}_{i}(Y;\cdot)$.
Together with (\ref{eq: Permutation Invariance def}) this implies
(\ref{eq:SiteEnergyPermInv}).

We now show isometry invariance (\ref{eq:SiteEnergyIsoInv}). First consider a
translation $A_{1}(x) = x + c$, for $c \in \R$, then 
\begin{align*}
m_{A_{1}Y}(x) = \sum_{j\in \mathbb{N}} \eta(x - Y_{j} - c) = m_{Y}(x - c) = m_{Y}(A_{1}^{-1}(x)).
\end{align*}
Then, by the uniqueness of the TFW equations, it follows that $(u_{A_{1}Y},\phi_{A_{1}Y})(\cdot) = (u_{Y},\phi_{Y})(\cdot - c)$, so $\mathcal{E}_{i}(A_{1}Y;\cdot) = \mathcal{E}_{i}(Y;\cdot - c)$ and thus
\begin{align*}
E^{i}_{j}(A_{1}Y) &= \int_{\R} \mathcal{E}_{i}(A_{1}Y;x) \varphi_{j}(A_{1}Y;x) \id x = \int_{\R} \mathcal{E}_{i}(Y;x-c) \varphi_{j}(Y;x-c) \id x \\ &= \int_{\R} \mathcal{E}_{i}(Y;z) \varphi_{j}(Y;z) \id z = E^{i}_{j}(Y).
\end{align*}
Similarly, for a rotation $A_{2}(x) = Rx$, $R \in {\rm O}(3)$, since we assumed
that $\eta$ is radially symmetric,
\begin{align}
m_{A_{2}Y}(x) = \sum_{j\in \mathbb{N}} \eta(x - RY_{j}) = \sum_{j\in \mathbb{N}} \eta(R(R^{T}x - Y_{j})) = \sum_{j\in \mathbb{N}} \eta(R^{T}x - Y_{j}) =  m_{Y}(R^{T}x). \label{eq: m Y rot eq}
\end{align}
As $(u_{Y},\phi_{Y})$ solve (\ref{eq: u phi eq pair})
\begin{align*}
- \Delta u_{Y} &+ \frac{5}{3} u^{7/3}_{Y} - \phi_{Y} u_{Y} = 0,  \\
- \Delta \phi_{Y} &= 4\pi (m_{Y} - u^{2}_{Y}),
\end{align*}
then by (\ref{eq: m Y rot eq}) and as the Laplacian is invariant under
rotations, it follows that \newline $(u,\phi) = (u_{Y}, \phi_{Y}) \circ A_2^{-2}$ solves
\begin{align*}
- \Delta u &+ \frac{5}{3} u^{7/3} - \phi u = 0,  \\
- \Delta \phi &= 4\pi (m_{Y} \circ R^{T} - u^{2}) = 4\pi (m_{A_{2}Y} - u^{2}),
\end{align*}
hence the uniqueness of (\ref{eq: u phi eq pair}) implies $(u_{A_{2}Y},\phi_{A_{2}Y}) = (u_{Y}, \phi_{Y}) \circ A_2^{-1}$. It follows that $E_{i}(A_{2}Y;\cdot) = E_{i}(Y;R^{T}\cdot)$, hence as $\det(R) = 1$, a change of variables shows
\begin{align*}
E^{i}_{j}(A_{2}Y) &= \int_{\R} \mathcal{E}_{i}(A_{2}Y;x) \varphi_{j}(A_{2}Y;x) \id x = \int_{\R} \mathcal{E}_{i}(Y;R^{T}x) \varphi_{j}(Y;R^{T}x) \id x \\ &= \int_{\R} \mathcal{E}_{i}(Y;z) \varphi_{j}(Y;z) |\det(R)| \id z = \int_{\R} \mathcal{E}_{i}(Y;z) \varphi_{j}(Y;z) \id z = E^{i}_{j}(Y).
\end{align*}
As the site energies are invariant under both translations and rotations, they are invariant under all isometries of $\R$.
\end{proof}

\bibliographystyle{plain}
\bibliography{sample}

\end{document}